\pdfoutput=1
\documentclass[11pt]{article}
\usepackage[T1]{fontenc}
\usepackage{lmodern}
\usepackage[a4paper,margin=1in]{geometry}
\usepackage{amsmath,amssymb,amsthm}
\usepackage{tikz}
\usetikzlibrary{arrows.meta,positioning,calc}
\usepackage{booktabs}
\usepackage{listings}
\usepackage{algorithm}
\usepackage{algpseudocode}
\usepackage[hidelinks]{hyperref}
\theoremstyle{plain}
\newtheorem{theorem}{Theorem}[section]
\newtheorem{lemma}[theorem]{Lemma}
\newtheorem{proposition}[theorem]{Proposition}
\newtheorem{corollary}[theorem]{Corollary}
\theoremstyle{definition}
\newtheorem{definition}[theorem]{Definition}
\newtheorem{problem}[theorem]{Problem}
\theoremstyle{remark}
\newtheorem{remark}[theorem]{Remark}

\newcommand{\Suff}{\mathrm{Suff}}
\newcommand{\Ach}{\mathrm{Ach}}
\newcommand{\stops}{\mathrm{stops}}
\newcommand{\Non}{\mathrm{Non}}
\newcommand{\Unq}{\mathrm{Unq}}
\newcommand{\auth}{\mathrm{auth}}
\newcommand{\conf}{\mathrm{conf}}
\newcommand{\secure}{\mathrm{sec}}
\newcommand{\trustlat}{\mathbb{T}}
\newcommand{\goals}{\mathbb{G}}
\newcommand{\completion}{\mathcal{D}}
\newcommand{\attacks}{\mathcal{H}}
\newcommand{\blocker}{\mathrm{b}}
\newcommand{\weakest}{\mathrm{W}}
\newcommand{\strongest}{\mathrm{S}}
\newcommand{\sig}[2]{\mathsf{sig}_{#1}(#2)}
\newcommand{\ach}{\Vdash}
\newcommand{\numNProtocols}{ten}

\newcommand{\numNRows}{18}

\newcommand{\numMaxN}{80}

\newcommand{\numLBoundMaxStatesOne}{3\,247}
\newcommand{\numLBoundMaxStatesTwo}{3\,309\,263}

\newcommand{\numLBoundTimeout}{600}
\newcommand{\numLBtwoBudget}{2\,097\,152}
\newcommand{\numLCpsaLoadMax}{6.4}
\newcommand{\numLCpsaLoadMin}{3.6}
\newcommand{\numLCpsaMaxAny}{1.4}
\newcommand{\numLCpsaMaxIqr}{0.04}
\newcommand{\numLCpsaReps}{5}
\newcommand{\numLCpsaRunEnd}{2026-10-06 22:49:27}
\newcommand{\numLCpsaRunStart}{2026-10-06 22:46:06}
\newcommand{\numLCpsaShareMax}{99}
\newcommand{\numLCpsaShareMin}{96}
\newcommand{\numLCpsaStepLimit}{2000}
\newcommand{\numLCpsaStrandBound}{12}
\newcommand{\numLCpsaTimeout}{600}
\newcommand{\numLCpsaVersion}{4.4.9}
\newcommand{\numLDuneVersion}{3.19.1}
\newcommand{\numLGadgetMaxN}{6}

\newcommand{\numLLibLines}{2\,483}
\newcommand{\numLLoadMax}{4.9}
\newcommand{\numLLoadMin}{1.7}
\newcommand{\numLMachine}{Apple M3 Max, 14 cores, 36 GB}
\newcommand{\numLMaxH}{2\,048}
\newcommand{\numLMaxN}{80}
\newcommand{\numLMaxU}{6}
\newcommand{\numLMaxW}{29\,886}
\newcommand{\numLMinU}{2}
\newcommand{\numLNAlgAtBound}{459}

\newcommand{\numLNAlgRuns}{459}

\newcommand{\numLNAnalyserLeastAtBound}{5}
\newcommand{\numLNAnalyserPrunedAtBound}{3}
\newcommand{\numLNAnalyserRows}{5}
\newcommand{\numLNBoundRuns}{30}
\newcommand{\numLNBoundTimeouts}{8}

\newcommand{\numLNBthreeRuns}{520}
\newcommand{\numLNBtwoBudget}{362}
\newcommand{\numLNChannelAgree}{6}
\newcommand{\numLNChannelChanged}{4}
\newcommand{\numLNChannelRows}{6}
\newcommand{\numLNCpsaGoalRuns}{41}

\newcommand{\numLNCpsaMonotone}{13}
\newcommand{\numLNCpsaOverBound}{8}
\newcommand{\numLNCpsaPoints}{216}
\newcommand{\numLNCpsaRejected}{29}
\newcommand{\numLNCpsaRejectedRows}{5}
\newcommand{\numLNCpsaReplayRuns}{221}
\newcommand{\numLNCpsaRows}{13}
\newcommand{\numLNCpsaShrinkReplayRuns}{131}
\newcommand{\numLNDiffBoundedLines}{460}
\newcommand{\numLNDiffCore}{300}
\newcommand{\numLNDiffCoreMismatch}{0}
\newcommand{\numLNDiffCpsaFields}{17}
\newcommand{\numLNDiffCpsaRows}{13}
\newcommand{\numLNDigestInstances}{466}
\newcommand{\numLNDigestMismatch}{0}
\newcommand{\numLNFalse}{6}
\newcommand{\numLNGadget}{30}
\newcommand{\numLNGadgetAgree}{30}
\newcommand{\numLNGadgetTrusts}{1\,072}
\newcommand{\numLNGreedyFastest}{10}
\newcommand{\numLNGreedyInW}{13}
\newcommand{\numLNLevelwiseAgree}{13}
\newcommand{\numLNLevelwiseFasterThanLoop}{9}
\newcommand{\numLNLevelwiseFewerCalls}{8}
\newcommand{\numLNLoopFasterThanExh}{7}

\newcommand{\numLNPaddedRuns}{145}
\newcommand{\numLNRows}{18}
\newcommand{\numLNRowsAgree}{18}
\newcommand{\numLNScaleBlockerDominates}{15}
\newcommand{\numLNScaleRuns}{2\,305}
\newcommand{\numLNScaleSlowRuns}{15}

\newcommand{\numLNScaleTimeouts}{49}
\newcommand{\numLNSeeds}{10}
\newcommand{\numLNShrinkAgree}{13}
\newcommand{\numLNShrinkFasterThanLoop}{6}
\newcommand{\numLNShrinkFewerReplays}{8}
\newcommand{\numLNShrinkMoreCalls}{7}
\newcommand{\numLNShrinkRuns}{458}

\newcommand{\numLNSteps}{11}
\newcommand{\numLNTestsFailed}{0}
\newcommand{\numLNTestsPassed}{13\,333}
\newcommand{\numLOS}{macOS 26.5.2}
\newcommand{\numLOcamlVersion}{5.1.0}
\newcommand{\numLPaddedTimeoutN}{25}
\newcommand{\numLRowCap}{16}
\newcommand{\numLRunEnd}{2026-10-04 18:07:22}
\newcommand{\numLRunStart}{2026-10-04 14:28:50}
\newcommand{\numLScaleHeapCap}{3\,000}
\newcommand{\numLScaleMaxTime}{189}
\newcommand{\numLScaleTimeout}{600}
\newcommand{\numLScaleWorkers}{10}
\newcommand{\numLSelectTimeout}{60}
\newcommand{\numLSignedAsFound}{5}
\newcommand{\numLSignedPruned}{3}
\newcommand{\numLSourcesSha}{12dda0ec723b}
\newcommand{\numLTwoKeyPruned}{4}

\title{Trust a Few:\\ The Weakest Assumptions a Protocol Needs}
\author{\begin{tabular}[t]{c@{\hspace{3em}}c}
Bhumika Mittal & Aalok Thakkar\\
{\normalfont\small School of Computer Science} & {\normalfont\small Vachani School of Advanced Computing}\\
{\normalfont\small Georgia Institute of Technology} & {\normalfont\small Ashoka University}
\end{tabular}}
\date{}

\begin{document}
\maketitle
\begin{abstract}
Protocol verifiers check whether a protocol meets a security goal under
stated trust assumptions, such as a key remaining secret, a value being
fresh, or a channel being authentic. They do not identify which
assumptions the goal requires. Rowe, Guttman, and Liskov asked for the
weakest assumptions under which a protocol achieves a goal and left the
question for future research. We answer this question for assumptions
about keys, values, and channels. Call a run that violates the goal an
attack, and call the assumptions that would rule it out its stopping set.
We show that the weakest trust is exactly the set of minimal hitting sets
of all minimal attacks: the smallest sets of assumptions that contain a
member of every minimal attack's stopping set. Several such sets may
exist, none weaker than another. The answer is unique when either/or
assumptions are allowed. A single weakest trust exists exactly when every
minimal attack has only one stopping assumption. A Galois connection
between assumptions and goals explains the need for disjunction. Goal
conclusions allow disjunction, so the strongest goal for given
assumptions exists whenever shape analysis terminates. Hypotheses do not,
so the weakest assumptions for a goal need not exist. With no
protocol-specific argument, the same structure gives the trust required
by a conjunction of goals, an exact condition under which composed
protocols need no extra trust, and a restatement of the
attack-containment tenet of Rowe et al., with attacks compared by their
stopping sets. We compute the weakest trust using a loop that asks a protocol verifier whether a candidate suffices,
records what stops each attack it reports, and recomputes the candidates.
Deciding whether some trust of at most $k$ assumptions suffices is
NP-hard. Using a bounded analyser of our own and CPSA as verifiers, we
find the weakest trust for \numNRows{} goals across \numNProtocols{}
protocols (counting our running example) and their variants. On every goal, the results agree with
evaluating every trust using the same verifier. In our model of the
adopted fix of Kerberos \textsc{pkinit}, the analysis identifies an
assumption in the stated hypothesis that is not needed for the client's
authentication of the server. It also finds channel assumptions under
which Needham--Schroeder meets its goal in our bounded analyser, with one
instance per role and channels read as limits on injecting and reading
messages. On synthetic instances with up to \numMaxN{} candidate
assumptions and a set-family oracle in place of a verifier, every
completed run of the loop makes exactly as many goal queries as there are
minimal attacks and weakest trusts together, when each reported attack
has a minimal stopping set. The loop is Reiter's, the query bound is an
instance of known results on learning monotone Boolean functions, and the
verifier is a standard one. What is new is the characterisation in the
goal language of Rowe et al.\ and what follows from it.
\end{abstract}

\section{Introduction}
\label{sec:intro}

\begin{quote}
\emph{Love all, trust a few, do wrong to none.}\\
\hspace*{\fill}Shakespeare, \emph{All's Well That Ends Well}, Act~1, Scene~1
\end{quote}

The Countess gives this advice to her son Bertram as he leaves home for
the King's court, a place full of people whose honesty he cannot check. A
security protocol is sent into the same position. It must exchange messages with
anyone on the network, and the adversary may control every message it
receives. It relies on a few things it cannot check itself, such as a key
never being leaked, a value being fresh, or a channel being authentic. It
must still meet its goal: no one is deceived, and no secret leaks.

Protocol verifiers such as ProVerif, Tamarin, and CPSA check the last part
of that advice for a chosen middle part \cite{proverif,tamarin,cpsaspec}.
Given the assumptions, they decide whether the protocol meets its goal.
When the check succeeds, the designer knows that the assumptions are
enough. However, the designer does not learn which of them the goal
requires, unless the check is repeated under each combination of
assumptions, as some of the analyses discussed in Section~\ref{sec:related}
do. This paper computes the middle part, the few things a protocol must trust for a
given goal.

In 2005, the Kerberos public-key extension \textsc{pkinit}, already widely
deployed, was found to admit a man-in-the-middle attack. An insider could
relay a client's request to the authentication server under its own
identity and learn the session keys \cite{cervesato}. Two repairs were
proposed. The researchers who found the flaw suggested binding the
client's identity into the server's signature. The authors of the
standard, wanting to keep the message format, suggested replacing a nonce
with a keyed hash of the client's request. (Our CPSA model of this fix,
which is not shipped with CPSA, simplifies it to an unkeyed hash, as
Section~\ref{sec:eval-rq1} explains.) Both repairs stop the attack, and a
verifier confirms this in seconds. The debate between them was about what
each repair relies on: which keys must remain uncompromised, which values
must be fresh, and what a deployment must provide for the guarantee on
paper to hold. A standards committee weighs this when it chooses a
repair, and so does the implementer of a constrained device deciding what
the hardware must protect. The closest tool support checks a claim
against each member of a lattice of adversary or threat models and
reports the models under which the claim holds. These lattices are built
from capabilities such as compromising keys or state \cite{basincremers},
infecting devices or compromising keys \cite{jacommemfa,jacommeedhoc}, or
exploiting weak hash functions \cite{chevalhash}. This addresses the dual
question, which we compare in Section~\ref{sec:related}.

Write $\tau \ach \chi$ for ``under the assumptions $\tau$, protocol $\Pi$
achieves goal $\chi$ against a Dolev--Yao adversary''
\cite{dolevyao,thayer}, and call a set of assumptions $\tau$ a
\emph{trust}. The relation links three objects. Fixing two and solving for
the third gives three problems (Fig.~\ref{fig:triangle}). Solving for the
protocol is protocol synthesis, attempted since Perrig and Song
\cite{apg} and still hard. Solving for the goal is the problem Rowe,
Guttman, and Liskov solved \cite{rowe}. For a fixed protocol and fixed
assumptions, the protocol achieves a single strongest goal, which their
shape analysis computes. Shape analysis is the search that CPSA performs.
This paper solves for the trust.

\begin{figure}[t]
\centering
\begin{tikzpicture}[scale=0.9]
  \node[label={[font=\footnotesize\itshape]above:love all}] (P) at (90:1.7) {protocol $\Pi$};
  \node[label={[font=\footnotesize\itshape]below:trust a few}] (T) at (210:2.1) {trust $\tau$};
  \node[label={[font=\footnotesize\itshape]below:do wrong to none}] (G) at (330:2.1) {goal $\chi$};
  \draw (P) -- (T); \draw (T) -- (G); \draw (G) -- (P);
  \node[align=center,font=\scriptsize] at (0,-2.9)
    {fix $\Pi,\tau$, solve for strongest $\chi$: shape analysis \cite{rowe}\\
     fix $\Pi,\chi$, solve for weakest $\tau$: this paper\\
     fix $\tau,\chi$, solve for $\Pi$: protocol synthesis \cite{apg}};
\end{tikzpicture}
\caption{The three objects a protocol analysis relates, with the clause of
the Countess's advice that each answers, and the three problems obtained by
solving for each in turn.}
\label{fig:triangle}
\end{figure}
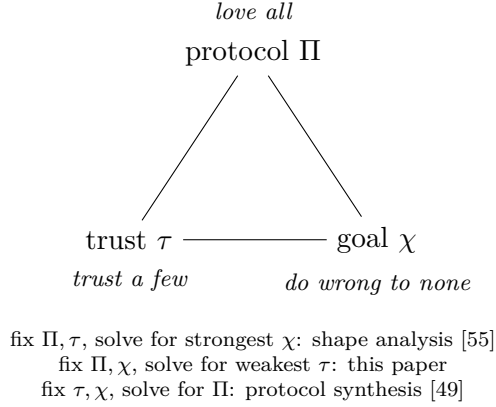

Rowe, Guttman, and Liskov also posed the trust problem \cite{rowe}. They
express a security goal as an implication $\forall\bar{x}.\, (\Phi
\Rightarrow \Psi)$, where the hypothesis $\Phi$ is a conjunction of atoms
and the conclusion $\Psi$ is a disjunction of existentially quantified
conjunctions. For a fixed hypothesis, the goals the protocol achieves
have a single strongest member. For a fixed conclusion, however, Rowe et
al.\ observe that there need not be a single weakest hypothesis under
which the protocol achieves it. They describe finding such hypotheses as
a form of abduction. They report knowing of no tools designed to support
the ``impoverish-as-possible'' reasoning it requires, and leave it for
future research \cite[Section~6.4]{rowe}. Alves-Foss and Soule generate
necessary preconditions for protocol goals in a belief logic
\cite{alvesfoss}, and the threat-model lattices above answer the dual
question for adversary capabilities by enumeration. To our knowledge, the
question has not been answered for hypotheses of the goal language of
\cite{rowe}.

\paragraph{The answer.}
Call a run of the protocol that violates the goal an \emph{attack}, and
call the assumptions that would rule out an attack its \emph{stopping
set}. The \emph{minimal attacks} are those whose stopping sets contain no
other attack's stopping set. The least a protocol must trust to meet a goal
is exactly the set of minimal ways to stop all of its minimal attacks: the
smallest trusts that contain a member of every minimal attack's stopping
set. There may be several, none weaker than another. The answer is unique
once \emph{either/or} assumptions, such as ``$k_1$ or $k_2$ is
uncompromised'', are allowed. A protocol verifier can compute it. There is
a single weakest trust, with no either/or, exactly when every minimal
attack is stopped by one assumption. The weakest trust and the minimal
attacks determine each other, by a classical result of Edmonds and
Fulkerson and of Isbell \cite{edmonds,isbell}. Section~\ref{sec:example}
works all of this out by hand on a signed challenge-response. The same
example refutes the informal criterion for uniqueness, which leaves out the
word ``minimal''.

\paragraph{Why either/or is needed.}
Assuming more trust cannot make a goal fail, and weaker goals are no
harder to achieve. Because of these two properties, when both the
strongest goal achieved under a trust and the weakest trust sufficient for
a goal exist, these maps form a Galois connection, a pair of maps that
determine each other. Birkhoff's polarities \cite{birkhoff} always give a
connection between sets of trusts and sets of goals. Between single trusts
and single goals, the connection exists exactly when the weakest trust is
unique. Both maps produce a disjunction. For the strongest goal, the
disjunction ranges over the minimal runs compatible with a hypothesis. The
goal language of \cite{rowe} allows disjunctions in conclusions, so the
strongest goal exists whenever shape analysis terminates
\cite{rowe,guttmantosca,ramsdell}. For the weakest trust, the disjunction
ranges over the minimal sufficient trusts. However, hypotheses are
restricted to conjunctions, so a single weakest trust need not exist. This
gap comes from the goal language and is the same for every protocol.
Allowing either/or hypotheses resolves this asymmetry through the
disjunctive completion of abstract interpretation \cite{cousot}. With
finitely many candidate assumptions, this completion is computable, though
not cheaply.

\paragraph{What is ours.}
We propose no new search algorithm. The order theory, the duality between
attacks and trusts, and the bound on verifier queries are known results,
whose sources Section~\ref{sec:related} names. Chatterjee, Henzinger, and
Jobstmann showed through an either/or example that weakest environment
assumptions in synthesis need not be unique \cite{chatterjee}. The finite
goal sets of Rowe et al.\ \cite[Section~3.4]{rowe} can also express
either/or hypotheses by using one goal per disjunct. What is new is the
weakest trust in the goal language of \cite{rowe}, studied as the
counterpart of the strongest goal under achievement. We explain how the
restriction to conjunctive hypotheses causes non-uniqueness. We
characterise the weakest trust as the minimal ways to stop every minimal
attack and derive the uniqueness criterion of Theorem~\ref{thm:unique} from
blocker duality. This criterion corrects the informal one. We also observe
that restricting hypotheses to trust atoms (A2) is what keeps the set of
candidate assumptions finite (A1), which both the completion and the
computation require. We compute the answer with a protocol verifier as an
oracle in a hitting-set loop, and a replay query recovers the stopping
sets of reported attacks. The restriction to trust atoms also defines the
scope of our results. Rowe et al.\ ask about hypotheses that are arbitrary
conjunctions of atoms, including atoms about role positions and
parameters. We answer the question for hypotheses built from trust atoms,
and Section~\ref{sec:limits} discusses what the general case would
require.

\paragraph{Contributions.}
\begin{enumerate}
\item \emph{A finite space of trusts} (Section~\ref{sec:lattice}). Six
  modelling assumptions fix the setting, each stated with what it gives and
  what it costs. We read the channel levels of M\"odersheim and Vigan\`o
  \cite{modersheim} as restrictions on adversary transitions in strand
  spaces. Under this reading, Lemma~\ref{lem:generators} shows that a
  secure channel is exactly an authentic and confidential one. Thus, the
  candidate assumptions are a finite set $U$ of independent atoms about
  keys, values, and channels, and a trust is any subset of $U$. The only
  property of the protocol semantics used later is that assuming more only
  removes runs, together with its refinement, separable stopping (A6).
\item \emph{Why the answer needs either/or}
  (Section~\ref{sec:galois}). The strongest-goal map $\strongest$ and the
  weakest-trust map $\weakest$ determine each other whenever both exist.
  Two sufficient trusts can have an insufficient intersection, so a single
  weakest trust need not exist. With either/or trusts, the weakest trust
  exists and is unique. As a formula, it is the disjunction of the minimal
  sufficient trusts (Theorem~\ref{thm:completion}). Thus, every goal has a
  unique weakest hypothesis when hypotheses may be either/or combinations
  of trust atoms.
\item \emph{The weakest trust is the set of minimal ways to stop every
  minimal attack} (Section~\ref{sec:blocker}). Each attack has a stopping
  set, and a trust rules out a run exactly when one of its members does.
  Thus, a trust suffices exactly when it meets the stopping set of every
  minimal attack, and $\weakest(\chi)$ is the family of minimal such
  trusts. This duality is Reiter's \cite{reiter}. Theorem~\ref{thm:unique}
  follows from it: the weakest trust is unique if and only if every minimal
  attack has a single stopping assumption. The informal criterion, which
  leaves out ``minimal'', is false, as the example of
  Section~\ref{sec:example} shows. An attack that no assumption stops gives
  the answer $\mathsf{false}$. This separates ``assume more'' from ``no
  assumption over these atoms suffices''.
\item \emph{Consequences} (Section~\ref{sec:consequences}). The trust
  required by a conjunction of goals is computed from the trust each goal
  requires. Provided every component attack carries over to the
  composition, two composed protocols need at least the combined trust of
  their parts. Theorem~\ref{thm:composition} shows that they need exactly
  that trust if and only if every minimal attack on the composition is
  stopped by no less than some minimal attack on a part. We conjecture that
  disjoint encryption, under which neither protocol creates or opens the
  other's encrypted messages, suffices. Table~\ref{tab:composition} gives
  two pairs where the combination is exact and one where it is not. In the
  last pair, multi-protocol attacks \cite{cremersares}, which pass messages
  of one protocol off as messages of the other, make the composition need
  more. The tight trusts, those equal to the weakest trust of their own
  strongest goal, correspond one to one, in order, with the goals that are
  the strongest their own weakest trust achieves. Because a conclusion may
  contain trust atoms, nearly every trust is tight. Hence, this
  correspondence does not by itself remove unneeded assumptions. Rowe et
  al.\ \cite{rowe} state the tenet that a system is at least as secure as
  another exactly when every attack on the first is also an attack on the
  second. For two protocols over a common situation and goal language, we
  restate this tenet as a theorem by comparing attacks through their
  stopping sets. This is a restatement of their tenet in our terms rather
  than a derivation.
\item \emph{A computation} (Section~\ref{sec:computing}). Running CPSA's
  search once from the empty trust misses attacks. A loop that asks a
  verifier for attacks computes $\weakest(\chi)$ and the minimal attacks
  together, with the verifier as its only oracle. A replay query computes
  stopping sets from any verifier's attacks. The \textsc{pkinit} analysis
  of \cite{rowe} follows the pattern of this loop by hand. Given the
  minimal attacks, checking a candidate answer takes quasi-polynomial time
  \cite{fredman}. Given only a verifier, checking it requires the blocker
  of the candidate and one goal query per member. Whether the answer can be
  listed in time polynomial in its size is open. The loop makes at most
  $|\attacks(\chi)| + |\weakest(\chi)|$ goal queries when the verifier's
  attacks have minimal stopping sets, not counting the replay queries. This
  bound is an instance of known results on learning monotone Boolean
  functions with membership queries and on dualisation
  \cite{angluin,bshouty,biochibaraki,gunopulos}. Deciding whether some
  trust of at most $k$ assumptions suffices is NP-hard, by a reduction that
  encodes the hitting-set problem as a protocol.
\item \emph{An implementation and an evaluation} (Section~\ref{sec:eval}).
  We evaluate an OCaml implementation using our bounded analyser and CPSA.
  It computes $\weakest(\chi)$ for \numNRows{} goals across
  \numNProtocols{} protocols, counting the running example
  (Table~\ref{tab:results}). On each goal, the result agrees with
  evaluating every trust using the same oracle. This comparison tests the
  loop and the stopping-set computation, but not the oracle. The
  evaluation also reports which answers depend on the analyser's bound
  (Table~\ref{tab:bound}). Every reported attack is a run of the protocol,
  so every trust that suffices without a bound meets each stopping set the
  loop recovers. Sufficiency over the bounded analyser holds only for the
  finite set of sessions it searches. We also use synthetic instances from
  the hardness reduction, with up to \numMaxN{} candidate assumptions and a
  set-family oracle in place of a verifier. On these, the loop makes exactly
  the number of goal queries that the cost bound allows whenever the
  reported attacks have minimal stopping sets, while enumeration needs up
  to $2^n$ calls (Table~\ref{tab:scale}). The conjunct-dropping heuristic
  of \cite{rowe} returns one weakest trust and cannot detect
  non-uniqueness.
\end{enumerate}

\section{An Example: What the Weakest Trust Looks Like}
\label{sec:example}

On a small protocol, the weakest trust can be found by hand by listing the
attacks and what stops each. In a signed challenge-response protocol, two
parties $A$ and $B$ exchange messages so that $B$ can be sure it has just
spoken with $A$. Challenger $B$ sends a fresh challenge $N$ to prover $A$.
The prover responds by signing $N$ together with both names, and $B$
accepts if the signature verifies under $A$'s public key:
\[
B \rightarrow A :\ N
\qquad\qquad
A \rightarrow B :\ \sig{sk_A}{N, A, B}.
\]
The situation $\sigma$ specifies a completed run of $B$ as challenger with
$A$ on challenge $N$. The goal $\chi$ is $B$'s recent agreement with $A$:
$A$ must have run the prover role with $B$ on the same $N$ and sent its
response after $B$ sent $N$. Consider the three candidate assumptions
\[
U = \{\Non(sk_A),\ \Non(sk_B),\ \Unq(N)\},
\]
stating that the signing keys are uncompromised and the challenge is
fresh. With all three assumed, $B$'s conclusion is sound. Only $A$ can
produce a signature under $sk_A$, since $\Non(sk_A)$ keeps the key from the
adversary. $A$ signs $(N, A, B)$ only as prover for $B$ on challenge $N$.
Since $N$ originates only at $B$'s transmission, $A$ received it, and
signed, after $B$ sent it.

The goal requires only two of the three. The unique weakest trust is
$\{\Non(sk_A),\Unq(N)\}$, because every attack on $B$'s goal is stopped by
one of these two assumptions and each of them is needed to stop some
attack. Thus, $B$'s goal does not require $sk_B$ to remain uncompromised.

The answer follows from the attacks. An \emph{attack} is a run that
satisfies $\sigma$, with nothing assumed beyond it, but violates $\chi$. Its
\emph{stopping set} $\stops(x)$ is the set of atoms of $U$ that rule out
the run when assumed. Section~\ref{sec:blocker} gives the general
definitions. Fig.~\ref{fig:runs} illustrates three attacks as strand
diagrams.
\begin{itemize}
\item \emph{Forgery} (Fig.~\ref{fig:runs}(a)). The adversary holds $sk_A$,
  receives $N$, and signs $(N, A, B)$ itself. $B$ completes its run without
  $A$ running the prover role at all, so $\chi$ fails. Because the run uses
  $sk_A$, assuming $\Non(sk_A)$ rules it out. It uses no other key, and
  $N$ originates once, so neither of the other atoms rules it out. Hence,
  $\stops(x) = \{\Non(sk_A)\}$.
\item \emph{Replay} (Fig.~\ref{fig:runs}(b)). An earlier session, in which
  $B$ challenged $A$ with the same value $N$ and $A$ answered, has left
  $\sig{sk_A}{N, A, B}$ in the adversary's knowledge. In the test session,
  $B$ sends $N$ again, and the adversary replays the recorded signature in
  response. $A$ did sign $(N, A, B)$ for $B$, but before $B$'s test
  challenge, so the recency required by $\chi$ fails. $N$ originates twice,
  once on each of $B$'s challenger strands, so $\Unq(N)$ rules the run out.
  The run uses no key, so neither key assumption does, and $\stops(x) =
  \{\Unq(N)\}$.
\item \emph{Double forgery} (Fig.~\ref{fig:runs}(c)). Both parties run both
  roles. In addition to the forgery, the run contains a second session in
  which $A$, acting as challenger, sends $N_a$ and accepts
  $\sig{sk_B}{N_a, B, A}$, which the adversary also forged. $B$'s run fails
  $\chi$ as in the forgery. The run uses both signing keys, so each key
  assumption rules it out on its own, and $N$ originates once. Thus,
  $\stops(x) = \{\Non(sk_A), \Non(sk_B)\}$. The second session plays no
  part in $B$'s goal. However, an attack need not be a minimal run, and a
  search may find this run before the forgery, as happens in
  Section~\ref{sec:eval}.
\end{itemize}
A \emph{reflection} attack, where $B$'s challenge is sent back to $B$ as
prover, does not work. In this case, $B$ would sign $(N, B, A)$, which it
rejects when acting as challenger. Nor does a run in which the adversary,
as a principal $I$, obtains from $A$ a signature on $(N, A, I)$, since $B$
rejects it for the same reason. The names inside the signed message
exclude both runs.

\begin{figure}[t]
\centering
\tikzset{
  rn/.style={circle,fill,inner sep=0pt,minimum size=4pt},
  an/.style={rectangle,draw,fill=black!12,inner sep=0pt,minimum size=5pt},
  msg/.style={-{Stealth[length=4pt]},semithick},
  rec/.style={-{Stealth[length=4pt]},semithick,dashed},
  strand/.style={double,double distance=1.2pt,-{Implies},shorten <=2pt,shorten >=2pt},
  lab/.style={font=\scriptsize},
  act/.style={font=\scriptsize,align=left},
  hd/.style={font=\small,anchor=base},
}
\begin{minipage}[b]{0.38\textwidth}
\centering
\begin{tikzpicture}
  \node[hd] at (0,0.25) {$B$};
  \node[hd] at (2.9,0.25) {adversary};
  \node[rn] (b1) at (0,0) {};
  \node[rn] (b2) at (0,-1.4) {};
  \node[an] (p1) at (2.9,0) {};
  \node[an] (p2) at (2.9,-1.4) {};
  \draw[strand] (b1) -- (b2);
  \draw[strand] (p1) -- node[act,right] {signs\\with $sk_A$} (p2);
  \draw[msg] (b1) -- node[lab,above] {$N$} (p1);
  \draw[msg] (p2) -- node[lab,above] {$\sig{sk_A}{N,A,B}$} (b2);
  \node[lab] at (1.45,-2.15) {$\stops(x) = \{\Non(sk_A)\}$};
\end{tikzpicture}

\smallskip
(a) forgery
\end{minipage}
\hfill
\begin{minipage}[b]{0.58\textwidth}
\centering
\begin{tikzpicture}
  \node[hd] at (0,0.25) {$B$};
  \node[hd] at (3,0.25) {$A$};
  \node[hd] at (6,0.25) {adversary};
  \node[font=\scriptsize,rotate=90] at (-1.35,-0.6) {earlier};
  \node[font=\scriptsize,rotate=90] at (-1.35,-3.2) {test};
  \node[rn,label={[lab]left:$N$ orig.}] (e1) at (0,0) {};
  \node[rn] (e2) at (0,-1.2) {};
  \node[rn] (a1) at (3,0) {};
  \node[rn] (a2) at (3,-1.2) {};
  \node[rn,label={[lab]left:$N$ orig.}] (b1) at (0,-2.6) {};
  \node[rn] (b2) at (0,-3.8) {};
  \node[an] (p0) at (6,-1.2) {};
  \node[an] (p1) at (6,-2.6) {};
  \node[an] (p2) at (6,-3.8) {};
  \draw[strand] (e1) -- (e2);
  \draw[strand] (a1) -- (a2);
  \draw[strand] (b1) -- (b2);
  \draw[strand] (p0) -- node[act,right] {records} (p1);
  \draw[strand] (p1) -- node[act,right] {replays} (p2);
  \draw[msg] (e1) -- node[lab,above] {$N$} (a1);
  \draw[msg] (a2) -- node[lab,above] {$\sig{sk_A}{N,A,B}$} (e2);
  \draw[rec] (a2) -- node[lab,above] {copy} (p0);
  \draw[msg] (b1) -- node[lab,above,pos=0.3] {$N$} (p1);
  \draw[msg] (p2) -- node[lab,above,pos=0.7] {$\sig{sk_A}{N,A,B}$} (b2);
  \node[lab] at (3,-4.55) {$\stops(x) = \{\Unq(N)\}$};
\end{tikzpicture}

\smallskip
(b) replay
\end{minipage}

\bigskip
\begin{minipage}[b]{0.56\textwidth}
\centering
\begin{tikzpicture}
  \node[hd] at (0,0.25) {$B$};
  \node[hd] at (3.2,0.25) {adversary};
  \node[hd] at (6.4,0.25) {$A$};
  \node[rn] (b1) at (0,0) {};
  \node[rn] (b2) at (0,-1.2) {};
  \node[an] (p1) at (3.2,0) {};
  \node[an] (p2) at (3.2,-1.2) {};
  \node[an] (p3) at (3.2,-2.0) {};
  \node[an] (p4) at (3.2,-3.2) {};
  \node[rn] (a1) at (6.4,-2.0) {};
  \node[rn] (a2) at (6.4,-3.2) {};
  \draw[strand] (b1) -- (b2);
  \draw[strand] (a1) -- (a2);
  \draw[strand] (p1) -- node[act,right] {signs with $sk_A$} (p2);
  \draw[strand] (p3) -- node[act,left,align=right] {signs with $sk_B$} (p4);
  \draw[msg] (b1) -- node[lab,above] {$N$} (p1);
  \draw[msg] (p2) -- node[lab,above] {$\sig{sk_A}{N,A,B}$} (b2);
  \draw[msg] (a1) -- node[lab,above] {$N_a$} (p3);
  \draw[msg] (p4) -- node[lab,above] {$\sig{sk_B}{N_a,B,A}$} (a2);
  \node[lab] at (3.2,-3.95) {$\stops(x) = \{\Non(sk_A), \Non(sk_B)\}$};
\end{tikzpicture}

\smallskip
(c) double forgery
\end{minipage}
\hfill
\begin{minipage}[b]{0.4\textwidth}
\centering
\begin{tikzpicture}
  \node[hd] at (0,0.25) {$B$};
  \node[hd] at (3.1,0.25) {adversary};
  \node[rn] (b1) at (0,0) {};
  \node[rn] (b2) at (0,-2.0) {};
  \node[an] (p1) at (3.1,0) {};
  \node[an] (p2) at (3.1,-1.0) {};
  \node[an] (p3) at (3.1,-2.0) {};
  \draw[strand] (b1) -- (b2);
  \draw[strand] (p1) -- node[act,right] {encrypts\\with $k_1$} (p2);
  \draw[strand] (p2) -- node[act,right] {encrypts\\with $k_2$} (p3);
  \draw[msg] (b1) -- node[lab,above] {$N$} (p1);
  \draw[msg] (p3) -- node[lab,above] {$\{\!|\{\!|N,A,B|\!\}_{k_1}|\!\}_{k_2}$} (b2);
  \node[lab] at (1.55,-2.75) {$\stops(x) = \{\Non(k_1), \Non(k_2)\}$};
\end{tikzpicture}

\smallskip
(d) forgery, two-key variant
\end{minipage}
\caption{Strand diagrams of attacks on the signed challenge-response
(a)--(c) and a forgery on its two-key variant (d). Each column represents a
principal. Filled circles mark regular-strand nodes, while shaded squares
mark adversary actions, labelled beside them. Double arrows connect
successive nodes of a strand, single arrows represent transmissions, and
the dashed arrow in (b) represents the adversary copying a transmission.
``$N$ orig.'' marks where $N$ originates. The stopping set is shown below
each run.}
\label{fig:runs}
\end{figure}
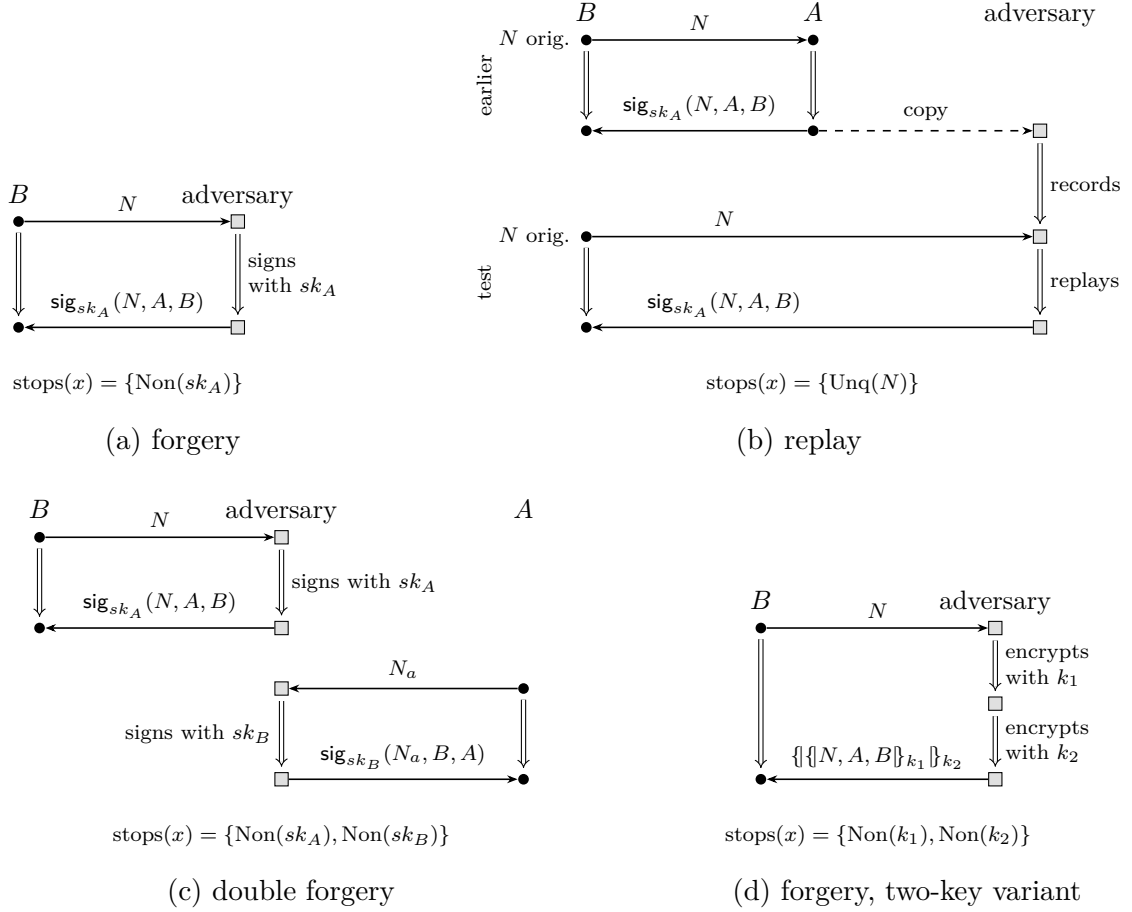

\begin{table}[t]
\caption{Attacks on the signed challenge-response and on its two-key
variant, with their stopping sets. A trust is sufficient when it meets
every row. The double-forgery row contains the forgery row, so every trust
that meets the forgery row also meets it. Deleting it leaves the minimal
attacks $\attacks(\chi)$, whose stopping sets determine whether the
weakest trust is unique (Section~\ref{sec:blocker}).}
\label{tab:example}
\centering
\begin{tabular}{llll}
\toprule
protocol & attack & $\stops(x)$ & status\\
\midrule
signed & forgery & $\{\Non(sk_A)\}$ & minimal\\
 & replay & $\{\Unq(N)\}$ & minimal\\
 & double forgery & $\{\Non(sk_A), \Non(sk_B)\}$ & redundant\\
\midrule
two-key & forgery & $\{\Non(k_1), \Non(k_2)\}$ & minimal\\
 & replay & $\{\Unq(N)\}$ & minimal\\
\bottomrule
\end{tabular}
\end{table}

Table~\ref{tab:example} collects the stopping sets. A trust is sufficient
exactly when it contains at least one atom from every row, as
Section~\ref{sec:blocker} proves for every protocol. The forgery and
replay rows have one entry each, so every sufficient trust contains
$\Non(sk_A)$ and $\Unq(N)$. Together, these assumptions also stop the
double forgery, through $\Non(sk_A)$. Since the double-forgery row
contains the forgery row, it is not minimal, and the minimal attacks are
$\attacks(\chi) = \{\{\Non(sk_A)\}, \{\Unq(N)\}\}$. Dropping either atom of
$\{\Non(sk_A),\Unq(N)\}$ leaves a row unmet. $\Non(sk_B)$ occurs only in
the row that is not minimal, so the goal does not need it.

The usual informal account of non-uniqueness goes wrong here. It says that
two incomparable weakest trusts appear when some attack has two stopping
assumptions. The double-forgery row has two entries, yet the answer is
unique, because that row is not minimal. Thus, an attack with two stopping
assumptions does not necessarily lead to two incomparable weakest trusts.
Non-uniqueness happens exactly when some minimal attack has at least two
stopping assumptions, as Section~\ref{sec:blocker} proves. Such an attack is a
single run that either of two assumptions rules out, so the run uses both
of the things those assumptions protect.

The two-key variant provides such a minimal attack. It replaces the
signature with a nested encryption under two symmetric keys $k_1$ and
$k_2$ shared by $A$ and $B$:
\[
B \rightarrow A :\ N
\qquad\qquad
A \rightarrow B :\ \{\!|\{\!|N, A, B|\!\}_{k_1}|\!\}_{k_2},
\]
with $U = \{\Non(k_1), \Non(k_2), \Unq(N)\}$ and with $\sigma$ and $\chi$
as before. To forge a response, the adversary needs both keys, since it
cannot produce the inner encryption without $k_1$ or the outer one without
$k_2$ (Fig.~\ref{fig:runs}(d)). Thus, the forgery is a single run, and
protecting either key stops it. Its stopping set $\{\Non(k_1),
\Non(k_2)\}$ represents an either/or choice over one run, rather than a
requirement to assume both. Replay is as before, with stopping set
$\{\Unq(N)\}$. The double forgery of this variant uses the same two keys,
has the same stopping set as the forgery, and adds no row. There are now
two weakest trusts,
\[
\{\Non(k_1),\Unq(N)\}
\qquad\text{and}\qquad
\{\Non(k_2),\Unq(N)\},
\]
neither weaker than the other. Their intersection $\{\Unq(N)\}$ is not
sufficient, since the forgery survives it. Thus, the sufficient trusts are
not closed under intersection, which Section~\ref{sec:galois} discusses in
general. The hardness proof of Section~\ref{sec:cost} uses the same
nesting, with one key per element of a set.

The example contains the main result. Under the modelling assumptions of
Section~\ref{sec:lattice}, the following hold for every protocol,
situation, and goal. Section~\ref{sec:galois} proves the last item, and
Section~\ref{sec:blocker} proves the others.
\begin{itemize}
\item A trust suffices exactly when it contains an atom from the stopping
  set of every minimal attack. The weakest trust is the family of smallest
  such trusts, that is, the minimal ways to stop all the minimal attacks.
\item The weakest trust and the minimal attacks determine each other. Read
  as formulas, the attacks give a conjunction of either/or clauses, here
  $\Unq(N) \wedge (\Non(k_1) \vee \Non(k_2))$. The weakest trusts give
  the same condition as an either/or of conjunctions, here $(\Non(k_1)
  \wedge \Unq(N)) \vee (\Non(k_2) \wedge \Unq(N))$.
\item A single weakest trust exists exactly when every minimal attack has
  one stopping assumption.
\item When there is no single weakest trust, no hypothesis of the goal
  language of \cite{rowe} expresses the answer, because hypotheses there
  are conjunctions of atoms. (A finite set of goals in the sense of
  \cite[Section~3.4]{rowe}, one per weakest trust, can express it.) The
  either/or trust above does express it, and when either/or trusts are
  allowed, the weakest trust is unique.
\end{itemize}
Computing the answer requires every minimal attack, and a single analysis
can miss some. Started from the empty trust, an enrich-by-need analysis
(the search CPSA performs) finds the forgery and stops, because the replay
appears only after $\Non(sk_A)$ is assumed, as Section~\ref{sec:obvious}
explains.
Instead, the loop of Section~\ref{sec:computing} asks the verifier about
each candidate trust in turn. Section~\ref{sec:eval} recovers both answers
of this section in this way (Table~\ref{tab:results}).

\section{Background: Protocols, Goals, and Achievement}
\label{sec:background}

To state the main result precisely, we need four things: what a run of a
protocol is, what the trust assumptions say about a run, what a goal is,
and when a protocol achieves a goal. We take all four from the strand-space
model of Thayer, Herzog, and Guttman \cite{thayer}, the goal language of
Guttman \cite{guttmanshapes,guttmanestablishing}, and its use by Rowe,
Guttman, and Liskov \cite{rowe}, whose question we answer.
Everything in this section is cited except two objects: the situation
$\sigma$ (Definition~\ref{def:situation}) and Problem~\ref{prob:weakest}.
Where the sources differ in presentation, we say which one we follow.
Nothing later depends on the points where they differ. The running example
is the signed challenge-response of Section~\ref{sec:example}.

\subsection{Messages and Strands}

Messages are terms built from names, nonces, and keys. What one principal
does in one session is a strand, the list of messages it sends and
receives.

\begin{definition}[Terms {\cite{thayer}}]
\label{def:terms}
Fix disjoint sets of \emph{atoms}: names, texts (nonces and other data), and
keys. Each key $K$ has an inverse $K^{-1}$. A symmetric key is its own
inverse, and the private key $sk_A$ and public key $pk_A$ of a principal $A$
are inverse to each other. \emph{Terms} form the free algebra over the atoms
with two constructors: pairing $(t_0, t_1)$ and encryption $\{\!|t|\!\}_K$
with $K$ a key. There is no equational theory. We write $\sig{sk_A}{t}$ for
$\{\!|t|\!\}_{sk_A}$. Since $sk_A^{-1} = pk_A$ is public, anyone can extract
$t$, but only a holder of $sk_A$ can produce the term. The \emph{subterm}
relation $\sqsubseteq$ is the least reflexive relation with $t \sqsubseteq
(t_0, t_1)$ whenever $t \sqsubseteq t_0$ or $t \sqsubseteq t_1$, and $t
\sqsubseteq \{\!|t_0|\!\}_K$ whenever $t \sqsubseteq t_0$. A key used to
encrypt is not thereby a subterm of the ciphertext.
\end{definition}

In the running example, the response $\sig{sk_A}{N, A, B}$ has $N$, $A$, $B$,
and the triple $(N, A, B)$ among its subterms, but not $sk_A$. Thus, $A$ can
send the response without sending its key.

\begin{definition}[Strands {\cite{thayer}}]
\label{def:strands}
A \emph{node} is a pair of a sign and a term: $+t$ (transmission of $t$) or
$-t$ (reception of $t$). A \emph{strand} is a finite sequence of nodes, the
local behaviour of one principal in one session or of one adversary step. A
\emph{strand space} is a set of strands. For nodes we write $n_1 \Rightarrow
n_2$ when $n_2$ immediately follows $n_1$ on the same strand, and $n_1
\rightarrow n_2$ when $n_1 = +t$ and $n_2 = -t$ lie on any strands, a
possible delivery of $t$. $\mathrm{msg}(n)$ is the term of node $n$.
\end{definition}

A \emph{protocol} $\Pi$ is a finite set of \emph{roles}, each a strand whose
terms contain variables, the role's \emph{parameters}. The \emph{regular}
strands of $\Pi$ are the instances of its roles under substitutions of terms
for parameters. They are the honest principals' sessions. In the running
example, the challenger role is $\langle +N,\ -\sig{sk_a}{N, a, b}\rangle$
and the prover role is $\langle -N,\ +\sig{sk_a}{N, a, b}\rangle$, with
parameters $a$, $b$, and $N$. $B$'s session with $A$ is the challenger
strand with $a = A$ and $b = B$.

\begin{definition}[Adversary strands {\cite{thayer,guttmanshapes}}]
\label{def:adversary}
The \emph{adversary} (or \emph{penetrator}) strands are the instances of
the following traces, for atoms $a$, terms $g$, $h$, and keys $K$:
\begin{center}
\begin{tabular}{ll@{\qquad}ll}
create & $\langle +a\rangle$ &
tee & $\langle -g,\ +g,\ +g\rangle$\\
flush & $\langle -g\rangle$ &
pair & $\langle -g,\ -h,\ +(g, h)\rangle$\\
separate & $\langle -(g, h),\ +g,\ +h\rangle$ &
encrypt & $\langle -K,\ -h,\ +\{\!|h|\!\}_K\rangle$\\
decrypt & $\langle -K^{-1},\ -\{\!|h|\!\}_K,\ +h\rangle$ & &
\end{tabular}
\end{center}
\end{definition}

These strands give the Dolev--Yao adversary \cite{dolevyao}. It controls
every delivery, and it can pair, separate, encrypt with keys it has, and
decrypt with keys whose inverses it has, but nothing else. Thayer et al.\
give the penetrator text strands and key strands, the latter only for keys
in a fixed set $K_P$ \cite{thayer}. We follow the later presentations
\cite{guttmanshapes,guttmanestablishing}, in which the adversary may create
any atom and the keys it may not have are named by non-origination.

\subsection{Bundles and Origination}

A run of a protocol is formally a bundle. The trust assumptions about keys
and nonces are statements about where in a run a value first appears.

\begin{definition}[Bundle {\cite{thayer}}]
\label{def:bundle}
A \emph{bundle} is a finite set of nodes $\mathcal{B}$ of a strand space,
with a set of edges $\rightarrow_{\mathcal{B}} \subseteq\, \rightarrow$
between nodes of $\mathcal{B}$, such that
\begin{enumerate}
\item every reception node $n_2 \in \mathcal{B}$ has exactly one $n_1$ with
  $n_1 \rightarrow_{\mathcal{B}} n_2$;
\item if $n_2 \in \mathcal{B}$ and $n_1 \Rightarrow n_2$, then $n_1 \in
  \mathcal{B}$;
\item the graph with edges $\rightarrow_{\mathcal{B}} \cup \Rightarrow$ on
  $\mathcal{B}$ is acyclic.
\end{enumerate}
Its reflexive transitive closure is a partial order $\preceq_{\mathcal{B}}$,
the causal order of the bundle. A bundle of $\Pi$ is one whose strands are
regular strands of $\Pi$ and adversary strands.
\end{definition}

A bundle is one complete run, and we call it a run in prose. The three
conditions say that every reception was sent by exactly one earlier
transmission, that each strand present has all its earlier nodes, and that
nothing depends on its own future. A transmission may be delivered to
several receptions. Every nonempty set of nodes of a bundle has
$\preceq_{\mathcal{B}}$-minimal members \cite{thayer}. This is the
induction principle used throughout strand-space proofs. Each diagram of
Fig.~\ref{fig:runs} is a bundle, with the adversary's strands drawn as
shaded squares.

\begin{definition}[Origination {\cite{thayer}}]
\label{def:origination}
An atom $v$ \emph{originates} at a node $n$ if $n$ is a transmission, $v
\sqsubseteq \mathrm{msg}(n)$, and $v \not\sqsubseteq \mathrm{msg}(n')$ for
every earlier node $n'$ of the same strand. In a bundle $\mathcal{B}$, $v$ is
\emph{non-originating}, written $\Non(v)$, if it originates at no node of
$\mathcal{B}$, and \emph{uniquely originating}, written $\Unq(v)$, if it
originates at exactly one node of $\mathcal{B}$.
\end{definition}

A value originates where it is first sent on a strand. The two predicates
are the strand-space forms of the two trust assumptions on values.
$\Non(sk_A)$ says that $sk_A$ is never sent by anyone, the adversary's
create strands included. $A$ may still sign with it, since a key used to
encrypt is not a subterm. Thus, $\Non(sk_A)$ means that $sk_A$ is
uncompromised. $\Unq(N)$ says that $N$ enters the run at one point only,
so no earlier session can have used it. It means that $N$ is fresh. Both
are properties of a whole bundle, and whether a bundle satisfies one of
them does not depend on any other atom.

In the forgery of Fig.~\ref{fig:runs}(a), the adversary obtains $sk_A$ on a create strand, so
$sk_A$ originates and the run violates $\Non(sk_A)$. In the replay of
Fig.~\ref{fig:runs}(b), $N$ originates on both of $B$'s challenger strands,
so the run violates $\Unq(N)$.

\begin{definition}[Skeletons and shapes {\cite{guttmanshapes,cpsaspec}}]
\label{def:shapes}
A \emph{skeleton} describes the regular part of a partial execution: a finite
set of regular nodes closed under $\Rightarrow$-predecessors, a partial order
on them, and sets of atoms assumed non-originating and uniquely originating.
It is \emph{realised} if it is the regular part of some bundle that extends
its order and satisfies its origination assumptions. Skeletons are related by
homomorphisms, which map strands to strands by a substitution and preserve
order and origination assumptions. The \emph{shapes} of a skeleton
$\mathbb{A}$ are the realised skeletons $\mathbb{B}$ with a homomorphism
$\mathbb{A} \to \mathbb{B}$ that are minimal in the homomorphism order of
\cite{guttmanshapes}. Every realised skeleton reachable
from $\mathbb{A}$ by a homomorphism is reachable from one of its shapes.
\end{definition}

In words, a skeleton is the honest part of a partial run together with the
assumptions made about it. Its shapes are the minimal ways to extend it,
by adding regular strands, ordering nodes, or identifying values, into the
regular part of a run. For $B$'s completed challenger strand with
$\Non(sk_A)$ and $\Unq(N)$ assumed, there is one shape. It adds a prover
strand of $A$ that receives $N$ after $B$ sends it and sends
the signature that $B$ receives. With nothing assumed, $B$'s strand is
already realised, since the adversary can forge the signature, and it is
its own shape. Section~\ref{sec:obvious} returns to this case.

We use skeletons and shapes only through the tool that computes them.
Enrich-by-need analysis, implemented in CPSA \cite{cpsaspec}, starts from a
skeleton and adds regular behaviour as needed to explain each reception the
adversary cannot produce, until every branch is realised. When it terminates,
it returns finitely many shapes. Termination is not guaranteed in
general, so Section~\ref{sec:galois} states it as an assumption where a
result requires it. Section~\ref{sec:computing} uses CPSA as an oracle and
relies only on this output.

\subsection{The Goal Language}

A goal says: whenever the honest activity named in its hypothesis occurred
and the trust assumptions hold, the activity named in its conclusion also
occurred. The weakest trust is a statement about the hypotheses of such
goals, so we need their exact form.

\begin{definition}[Atoms and goals {\cite{guttmanestablishing,rowe}}]
\label{def:goals}
The goal language of $\Pi$ is a first-order language with variables ranging
over nodes and over messages. Its atoms are:
\begin{itemize}
\item \emph{role-position} atoms $\rho_i(n)$, one for each role $\rho$ of
  $\Pi$ and position $i$ on it: node $n$ is the $i$-th node of an instance of
  $\rho$;
\item \emph{parameter} atoms $p(n, v)$, one for each parameter $p$ of a role:
  the strand of $n$ instantiates $p$ with $v$;
\item $\mathsf{Preceq}(m, n)$: $m$ precedes $n$ in the causal order;
  $\mathsf{Coll}(m, n)$: $m$ and $n$ lie on the same strand; and equality of
  terms;
\item the origination atoms $\Non(v)$, $\Unq(v)$, and $\mathsf{UnqAt}(n,
  v)$, the last saying that $v$ originates uniquely, at $n$.
\end{itemize}
A \emph{security goal} is a closed formula
\[
\forall \bar{x}.\,\bigl(\Phi \Rightarrow \Psi\bigr),
\qquad
\Psi \;=\; \bigvee_{1 \le i \le j} \exists \bar{y}_i.\, \Delta_i,
\]
where $\Phi$ and each $\Delta_i$ are conjunctions of atoms and $j \ge 0$. The
grammar is
\begin{align*}
\Phi, \Delta &::= \alpha \mid \Phi \wedge \Phi
&
\Psi &::= \mathsf{false} \mid \exists \bar{y}.\, \Delta \mid \Psi \vee \Psi
&
G &::= \forall \bar{x}.\,(\Phi \Rightarrow \Psi)
\end{align*}
with $\alpha$ ranging over atoms. $\Phi$ is the \emph{hypothesis} and $\Psi$
the \emph{conclusion}. There is no negation and no implication inside $\Phi$
or $\Psi$ (formulas of this form are called geometric).
\end{definition}

The hypothesis is a conjunction, but the conclusion may be a disjunction.
This difference explains why the forward question below has a single
answer and the inverse one may not, as Section~\ref{sec:galois} shows. The case
$j = 0$ makes $\Psi$ the empty disjunction $\mathsf{false}$, expressing that
no run satisfies the hypothesis. Secrecy goals take this form, with a hypothesis that includes a node of a
\emph{listener} role $\langle -v\rangle$ receiving the secret $v$
\cite{guttmanestablishing}. The case $j = 1$ covers authentication and
agreement. In the running example, writing $\mathsf{Chl}_i$ and
$\mathsf{Prv}_i$ for the positions of the challenger and prover roles,
recent agreement of $B$ with $A$ is the goal
\begin{align*}
\forall m_1, m_2, a, b, N.\ \bigl(\,
&\mathsf{Chl}_1(m_1) \wedge \mathsf{Chl}_2(m_2) \wedge
  \mathsf{Coll}(m_1, m_2) \wedge a(m_2, a) \wedge b(m_2, b) \wedge N(m_2, N)\\
&{}\wedge \Non(sk_a) \wedge \Unq(N)\\
\Rightarrow\ &\exists n.\ \mathsf{Prv}_2(n) \wedge a(n, a) \wedge b(n, b)
  \wedge N(n, N) \wedge \mathsf{Preceq}(m_1, n)\,\bigr).
\end{align*}
The first line of the hypothesis describes $B$'s completed run, and the
second line is the trust. Without $\Non(sk_a)$, the adversary can create $sk_a$ and
sign. Without $\Unq(N)$, a prover strand from a session that precedes $m_1$
can supply the signature, and $\mathsf{Preceq}(m_1, n)$ fails.

\begin{definition}[Satisfaction and achievement {\cite{guttmanestablishing,rowe}}]
\label{def:achieve}
Let $\mathcal{B}$ be a bundle of $\Pi$ and $\eta$ an assignment of nodes of
$\mathcal{B}$ to node variables and terms to message variables. An atom
holds at $(\mathcal{B}, \eta)$ when its reading above is true of the nodes
and terms $\eta$ assigns, with $\mathsf{Preceq}$ read as
$\preceq_{\mathcal{B}}$ and origination read in $\mathcal{B}$ as in
Definition~\ref{def:origination}. Conjunction, disjunction, existential and
universal quantification are read as usual, and $\mathcal{B} \models G$ when
$G$ holds at $(\mathcal{B}, \eta)$ for every $\eta$. The protocol $\Pi$
\emph{achieves} $G$ if every bundle of $\Pi$ satisfies $G$.
\end{definition}

Rowe et al.\ and Guttman state achievement over the executions of $\Pi$,
and we take executions to be bundles, as above. Two facts follow from the
definitions and are recorded in the sources. Achievement is closed under
consequence: if $\Pi$ achieves $G'$ and $G'$ implies $G$, then $\Pi$
achieves $G$ \cite[Lemma~1]{rowe}. Because goals are positive existential
in their conclusions, they are also preserved along homomorphisms. This
lets a goal be checked on the finitely many shapes instead of on all
bundles \cite{guttmanestablishing}.

\subsection{What Rowe, Guttman, and Liskov Did}

Rowe et al.\ answered the forward question. For a fixed hypothesis, the
goals the protocol achieves have a strongest member. Our question is the
inverse, so we state their answer in the form we use. Fix a hypothesis
$\Phi$ and consider the goals with hypothesis $\Phi$ that $\Pi$ achieves,
ordered by implication. Let $\mathbb{A}_\Phi$ be the skeleton that $\Phi$
describes, and for a skeleton $\mathbb{S}$ let $\delta(\mathbb{S})$ be the
conjunction of atoms describing its nodes, roles, parameters, order, and
origination assumptions. Shape analysis sentences are due to Guttman
\cite{guttmantosca} and Ramsdell \cite{ramsdell}. We state the form in
which Rowe et al.\ use them.

\begin{theorem}[Shape analysis sentence {\cite[Theorem~2]{rowe}}]
\label{thm:sas}
Suppose enrich-by-need terminates on $\mathbb{A}_\Phi$ with shapes
$\mathbb{S}_1, \dots, \mathbb{S}_j$. Let
\[
\Psi_\Phi \;=\; \bigvee_{1 \le i \le j} \exists \bar{y}_i.\,
\delta(\mathbb{S}_i),
\]
where $\bar{y}_i$ are the variables of $\delta(\mathbb{S}_i)$ not among the
variables $\bar{x}$ of $\Phi$. Then
\begin{enumerate}
\item $\Pi$ achieves $\forall \bar{x}.\,(\Phi \Rightarrow \Psi_\Phi)$, the
  \emph{shape analysis sentence} for $\Phi$; and
\item for every conclusion $\Psi$, if $\Pi$ achieves $\forall
  \bar{x}.\,(\Phi \Rightarrow \Psi)$ then the shape analysis sentence implies
  $\forall \bar{x}.\,(\Phi \Rightarrow \Psi)$.
\end{enumerate}
\end{theorem}

The statement is a paraphrase in our notation. Rowe et al.\ \cite{rowe}
give the exact form, including how the variables of each shape are matched
to those of $\Phi$. When there are no shapes, $j = 0$ and $\Psi_\Phi =
\mathsf{false}$. In this case, no run satisfies $\Phi$, which is how a
secrecy goal is established. In the running example, with $\Phi$ the
hypothesis of the goal above, there is one shape, the one described after
Definition~\ref{def:shapes}. Its $\delta$ names $A$'s prover strand and
its place after $B$'s challenge, so the shape analysis sentence implies
recent agreement.

The theorem presupposes that enrich-by-need terminates on
$\mathbb{A}_\Phi$, and it characterises the strongest conclusion when it
does.
It does not say when the search terminates. Thus, the forward map, from a
hypothesis to the strongest achieved conclusion, has a value wherever
enrich-by-need terminates, and CPSA computes it there.

\subsection{Situations and the Inverse Question}

The inverse question varies the trust and keeps the rest of the hypothesis
fixed, so the two parts must be separated. The hypothesis of a goal mixes
two kinds of atoms. Role-position, parameter, order, and equality atoms
specify which runs are under discussion. Origination atoms, and the channel
assumptions of Section~\ref{sec:lattice}, specify what is trusted. In the
goal above, the first line of the hypothesis is of the first kind, and
$\Non(sk_a) \wedge \Unq(N)$ is of the second.

\begin{definition}[Situation]
\label{def:situation}
A \emph{situation} $\sigma$ is a conjunction of atoms describing the local
runs under discussion: which roles reached which positions with which
parameters, in which order. It may also contain origination atoms that are
held fixed and not varied, such as freshness that the scenario itself
supplies. The \emph{trust} is the remaining part of the hypothesis, a
conjunction $\tau$ of atoms drawn from a finite set $U$ fixed in
Section~\ref{sec:lattice}, so that every hypothesis considered has the form
$\sigma \wedge \tau$.
\end{definition}

In the running example, $\sigma$ is the first line of the hypothesis above,
$B$'s completed challenger run with $A$ on $N$, and $\tau$ is drawn from
$U = \{\Non(sk_A), \Non(sk_B), \Unq(N)\}$. The question is how much must be
assumed beyond $\sigma$.

\begin{problem}[Weakest trust]
\label{prob:weakest}
Given a protocol $\Pi$, a situation $\sigma$, and a goal conclusion $\chi$,
determine the weakest trusts $\tau$ such that $\Pi$ achieves
$\forall\bar{x}.\,(\sigma \wedge \tau \Rightarrow \chi)$.
\end{problem}

Rowe et al.\ pose this problem as ordering the hypotheses that share a
conclusion. They observe that maximal elements need not be unique and
describe the search for them as abductive, a search for the weakest
explanation of a known consequence, which requires
``impoverish-as-possible'' reasoning. They write that they are ``unaware of any tools that are specifically designed to support
this type of reasoning'' \cite[Section~6.4]{rowe}.
Problem~\ref{prob:weakest} is their question restricted to hypotheses of the
form $\sigma \wedge \tau$ (assumption A2 of Section~\ref{sec:lattice}).

\begin{remark}[What minimality means, and what it is for]
\label{rem:minimality}
Minimality is relative to a given protocol, situation, and goal. If
authentication requires $\Non(sk_A)$ and $\Unq(N)$ but not $\Non(sk_B)$,
then $\Non(sk_B)$ is unnecessary \emph{for this authentication claim},
although it may be essential for secrecy or for another goal proved from
the same assumptions. Designers also do not err by assuming more. A uniform
baseline, such as assuming that all honest long-term keys are
uncompromised, is simpler to state, and it is the standard against which
protocols are compared \cite{basincremers}. Assuming more trust cannot make
a goal fail, so a stronger hypothesis is never
wrong, only possibly unneeded. Computing $\weakest(\chi)$ identifies which
assumptions a goal requires. It does not ask anyone to deploy at the
minimum. A designer needs this information as soon as an assumption
carries a cost, such as a hardware trust anchor that costs money in every
device.
\end{remark}

Theorem~\ref{thm:sas} and Problem~\ref{prob:weakest} ask about one relation
in opposite directions, and Section~\ref{sec:galois} relates the two.

\section{The Trust Lattice}
\label{sec:lattice}

The main result concerns one relation, $\tau \ach \chi$: under trust
$\tau$, the protocol achieves goal $\chi$. Section~\ref{sec:background}
gave the runs and goals. To state the result, we also need the set of
trusts and the relation itself. Six modelling assumptions fix these. For
each, we state what we assume, what it gives, and what it costs.

\begin{description}
\item[A1 (Finitely many candidates).] \emph{Assume:} the set $U$ of
  assumption atoms, the \emph{candidate assumptions}, is finite.
  \emph{Gives:} the trusts are the subsets of $U$, which form a finite
  Boolean lattice, so every family of trusts has a least upper bound and
  a greatest lower bound, and computing the minimal ways to stop every
  attack is a finite computation.
  \emph{Costs:} the keys, nonces, and channels in play must be named by the
  situation. Thus, the answer concerns a bounded scenario, such as $B$'s
  one run with $A$ in the running example, rather than a protocol
  quantified over unboundedly many sessions.
\item[A2 (Trust hypotheses only).] \emph{Assume:} only the trust part of
  the hypothesis varies, that is, the atoms $\Non(v)$ and $\Unq(v)$ and the
  channel assumptions. Which runs are under discussion, as given by
  role-position and parameter atoms, is fixed by $\sigma$. \emph{Gives:} a hypothesis is a
  subset of a finite set, so the hypotheses form a lattice. \emph{Costs:}
  Rowe et al.\ ask about all conjunctions of atoms of the goal language
  \cite{rowe}, and we do not answer their question in that generality. The
  restriction narrows the question itself, and Section~\ref{sec:lim-a2}
  discusses what the general case would require.
\item[A3 (Safety goals).] \emph{Assume:} every violation of a goal has a
  finite witness. Secrecy, aliveness (the partner ran the protocol), and
  non-injective agreement (the partner ran it with the same values) are of
  this kind. \emph{Gives:} a violation, when there is one, is a finite run.
  Hence, a search through runs finds it, and failure of a goal is
  semi-decidable. Deciding achievement requires more: an oracle that also
  terminates when the goal holds. An enrich-by-need analysis is such an oracle
  wherever it terminates, and Section~\ref{sec:galois} states termination
  as an assumption. (The family of stopping sets of
  Section~\ref{sec:blocker} is finite under A1 alone, as a family of
  subsets of a finite set. A3 lets an oracle find its members.)
  \emph{Costs:} injective agreement is outside our scope. Its statement
  requires a conclusion that equates the partner strands matched to two
  runs of the observing role, and we do not treat such conclusions as goals
  of Definition~\ref{def:goals}. The exclusion concerns what the goals
  express. The length of a violation plays no part in it.
\item[A4 (Monotonicity).] \emph{Assume:} assuming more only removes runs.
  \emph{Gives:} Lemma~\ref{lem:monotone}, which every later result uses.
  \emph{Costs:} none we are aware of. An assumption that added runs would
  not be a trust assumption in any usual sense.
\item[A5 (No equational theory).] \emph{Assume:} perfect cryptography over
  a free term algebra, so two terms are equal
  only when they are built the same way. \emph{Gives:} whether an
  assumption rules out a given run is well defined and decidable, so the
  stopping sets of Definition~\ref{def:attacks} are well defined.
  \emph{Costs:} Diffie--Hellman and other protocols that require algebraic
  laws are outside our scope, as they are for the analyses that would instantiate the oracle of
  Section~\ref{sec:computing}.
\item[A6 (Separable stopping).] \emph{Assume:} a set of assumptions rules
  out a run exactly when one of its members rules out that run alone. No
  pair of assumptions rules out a run that neither rules out by itself.
  \emph{Gives:} a trust suffices exactly when it contains at least one
  stopping assumption of every attack. This is a hitting condition over
  sets computed one atom at a time (Proposition~\ref{prop:hitting}).
  \emph{Costs:} a run must be one fixed bundle with one fixed assignment of
  its nodes and values to the variables of $\sigma$, so an adversary with
  two routes to one message gives two runs. Lemma~\ref{lem:separable} proves
  A6 for the atoms of $U$ below, and a richer assumption language would
  have to prove it again.
\end{description}

\subsection{Channels}
\label{sec:channels}

Designers trust channels as well as keys and nonces, assuming that a
channel is authentic or confidential. This subsection describes how those
assumptions enter, so that $U$ can contain them. The origination atoms
$\Non(v)$ and $\Unq(v)$ of Definition~\ref{def:origination} are hypothesis
atoms of the goal language. Channel assumptions are not, and cannot be
expressed as such within the fragment of Definition~\ref{def:goals}. ``Every message
received on $c$ was sent on $c$ by its claimed sender'' is a universal
statement with an existential inside, and ``nothing sent on $c$ reaches the
adversary'' is a negative one. Neither is a conjunction of atoms. Thus,
channel assumptions enter as restrictions on the adversary, which remove
bundles from the model, while origination atoms enter as conjuncts of the
hypothesis, which remove bundles from consideration. The proofs below treat
the two kinds separately.

\paragraph{The levels.}
We take three levels from M\"odersheim and Vigan\`o \cite{modersheim},
who build on Dilloway and Lowe \cite{dilloway}. On \emph{authentic}
channels, the receiver can be sure who sent the message. On
\emph{confidential} channels, only the intended receiver can read the
messages. \emph{Secure} channels are both. Kamil and Lowe
\cite{kamillowe} give channel specifications in strand spaces. They extend
the adversary with network capabilities beyond the Dolev--Yao ones,
including \emph{faking} (sending a message it composed as if from an
honest sender) and \emph{hijacking} (changing the sender or receiver
identity of a message on a channel). They specify each channel property as
a restriction on what the adversary may do. Following them, we read each
level as a restriction on adversary transitions. The two
readings in the definition below, the inject/read reading and the hijack
reading, are ours.

\begin{definition}[Channel bundles]
\label{def:channels}
Fix a finite set $C$ of channels and label every node of every regular
strand of $\Pi$ with a channel in $C$. Classify each edge $n_1 \rightarrow
n_2$ of a bundle by the \emph{capabilities} it uses:
\begin{itemize}
\item both nodes adversary nodes: none;
\item both nodes regular and on the same channel: none (\emph{delivery});
\item $n_1$ an adversary node, $n_2$ a regular reception on $c$:
  $\mathrm{inject}(c)$;
\item $n_1$ a regular transmission on $c$, $n_2$ an adversary node:
  $\mathrm{read}(c)$;
\item $n_1$ a regular transmission on $c'$, $n_2$ a regular reception on
  $c \neq c'$: in the \emph{inject/read reading}, both $\mathrm{read}(c')$
  and $\mathrm{inject}(c)$; in the \emph{hijack reading}, the single
  capability $\mathrm{hijack}(c)$.
\end{itemize}
The channel assumptions forbid capabilities:
$\auth(c)$ forbids $\mathrm{inject}(c)$, and also $\mathrm{hijack}(c)$ in
the hijack reading; $\conf(c)$ forbids $\mathrm{read}(c)$; and $\secure(c)$
forbids every capability that the adversary uses to place a foreign message
on $c$ or to learn a message sent on $c$. For
a set $R$ of channel assumptions, $\mathcal{B}(R)$ is the set of bundles of
$\Pi$ none of whose edges uses a capability forbidden by a member of $R$.
\end{definition}

In the inject/read reading, the adversary has two channel capabilities. It
may inject on $c$, by delivering to a reception on $c$ a message from one
of its own strands. It may also read from $c$, by receiving a message that an
honest principal sent on $c$. A message moved from one channel to another
passes through the adversary's knowledge, so it uses both. Under
$\auth(c)$, every reception on $c$ is delivered from a transmission on $c$,
possibly one that was delivered before. Thus, authenticity alone does not
prevent replay of a genuine message. The goals of A3 are non-injective,
and freshness, when a goal requires it, is the separate atom $\Unq$. Under
$\conf(c)$, honest messages on $c$ still reach their receivers on $c$, but
never the adversary. This reading is used throughout the paper, except in the hijack rows of
Section~\ref{sec:eval} and in the proof of Lemma~\ref{lem:generators},
which covers both readings.

\paragraph{Hijacking.}
The hijack reading adds a third capability of our own, which lets the
adversary move a message sent on $c'$ to a reception on $c$ without reading
it. This capability is a variant of Kamil and Lowe's hijacking, which keeps
the channel and changes the sender or receiver identity. The two are close
when, as in the evaluation, a channel is named by its endpoints, so that
changing an identity changes the channel. The hijack reading changes what
the levels mean. $\auth(c)$ still forbids every way of placing a foreign
message on $c$. However, $\conf(c')$ now forbids less. It keeps the
adversary from learning what is sent on $c'$, but no longer keeps that
message on $c'$. Lemma~\ref{lem:generators} still holds,
with the same proof, provided $\secure(c)$ is read as in
Definition~\ref{def:channels}. Some formulations give the secure level a
further restriction that neither $\auth(c)$ nor $\conf(c)$ imposes, for
example that a message may not be redirected away from $c$, or a binding to
a session or an endpoint. Such a restriction requires its own candidate
assumption per channel, which enlarges $U$ and changes nothing else in the
paper. Lemmas~\ref{lem:monotone} and~\ref{lem:separable} hold in both
readings, since in both, each channel assumption forbids a fixed set of
capabilities.

The choice of reading changes answers. In Lowe's attack on
Needham--Schroeder, the adversary injects on the channel $A{\to}B$ both
messages the responder receives. It also reads the responder's reply on
$B{\to}A$ to pass it to the initiator as its own. In the inject/read
reading, either $\auth(A{\to}B)$ or $\conf(B{\to}A)$ removes the run. In
the hijack reading, the reply can be moved to the initiator unread, so
$\conf(B{\to}A)$ no longer removes it. With one instance per role, only the
authentic branch survives.
Section~\ref{sec:eval} reports the channel rows in both readings.

\paragraph{The instantiation in the evaluation.}
In Section~\ref{sec:eval}, a channel is a directed pair of principals,
and a message travels on the channel from its sender to its intended
receiver as the role sees them. Candidate assumptions exist only for
channels between honest principals. A channel with the adversary at one end
carries no assumption, because the adversary is a legitimate endpoint
there. It reads what is sent to it and sends what it likes in its own
name.

\subsection{The Candidate Assumptions}

The trusts should be plain sets of atoms, so that the combinatorics of
Section~\ref{sec:blocker} applies. The three channel levels prevent this.
Ordered by strength, they give each channel the order $\text{none} <
\auth(c), \conf(c) < \secure(c)$, and the combinatorics would have to be
redone over partial orders. The next lemma removes this obstacle: the top
level is the conjunction of the other two, so it need not be a candidate
assumption.

Say that an assumption $a$ is \emph{at least as strong as} $b$ if, in
every strand space with channel-labelled regular nodes, every bundle
admitted by $a$ is admitted by $b$. Here, a bundle admits an origination
atom when it satisfies it under Definition~\ref{def:origination}, and it
admits a channel assumption $r$ when it lies in $\mathcal{B}(\{r\})$. A set of assumptions
is admitted when each member is, and two sets are \emph{equivalent} when
they admit the same bundles.

\begin{lemma}[Candidate assumptions]
\label{lem:generators}
For each channel $c$, assuming $\secure(c)$ is equivalent to assuming
$\auth(c)$ and $\conf(c)$ together. Consequently, writing
\[
U \;=\; \{\Non(v)\}_{v} \,\cup\, \{\Unq(v)\}_{v} \,\cup\,
        \{\auth(c), \conf(c)\}_{c},
\]
every trust is equivalent to a subset of $U$, and the elements of $U$ are
pairwise incomparable in strength.
\end{lemma}

\begin{proof}
The secure level forbids exactly the union of what the other two forbid.
For incomparability we give, for each ordered pair of atoms, a small bundle
that the first admits and the second does not.

\emph{The secure level.} A bundle lies in $\mathcal{B}(R)$ exactly when it
uses no capability in the union of the sets forbidden by the members of $R$.
By Definition~\ref{def:channels}, the set forbidden by $\secure(c)$ is the
union of the sets forbidden by $\auth(c)$ and by $\conf(c)$: in the
inject/read reading, $\{\mathrm{inject}(c), \mathrm{read}(c)\}$, and in the
hijack reading, $\{\mathrm{inject}(c), \mathrm{hijack}(c),
\mathrm{read}(c)\}$. This matches the secure level of
\cite{modersheim}, a channel that is both authentic and confidential, on
which the adversary can neither place a foreign message nor learn what is
sent. Hence $\mathcal{B}(R \cup
\{\secure(c)\}) = \mathcal{B}(R \cup \{\auth(c), \conf(c)\})$ for every $R$,
and the origination part of a trust is untouched, so the two trusts admit the
same bundles. Replacing each $\secure(c)$ in a trust by $\auth(c)$ and
$\conf(c)$ gives an equivalent subset of $U$.

\emph{Incomparability.} For distinct $a, b \in U$ we exhibit a bundle that
admits $a$ and not $b$. Let $u$ be an atom different from every value
mentioned in $U$. Three kinds of bundle suffice: bundles made of adversary
create strands $\langle +w\rangle$ only, which have no regular nodes and
therefore admit every channel assumption; the bundle $I_c$ consisting of a
create strand $\langle +u\rangle$ and a regular one-node strand $\langle
-u\rangle$ on $c$, joined by an edge that uses $\mathrm{inject}(c)$ and
nothing else; and the bundle $R_c$ consisting of a regular one-node strand
$\langle +u\rangle$ on $c$ and an adversary flush strand $\langle -u\rangle$,
joined by an edge that uses $\mathrm{read}(c)$ and nothing else. $I_c$ and
$R_c$ originate only $u$, so they admit $\Non(v)$ for every $v$ in $U$.
\begin{itemize}
\item $a$ a channel assumption, $b = \Non(v)$ or $b = \Unq(v)$: two create
  strands $\langle +v\rangle$ admit $a$ and originate $v$ twice, violating
  both.
\item $a$ a channel assumption, $b = \auth(c)$ with $a \neq b$: $I_c$ uses
  only $\mathrm{inject}(c)$, which $a$ does not forbid, and $b$ does.
\item $a$ a channel assumption, $b = \conf(c)$ with $a \neq b$: $R_c$ uses
  only $\mathrm{read}(c)$, which $a$ does not forbid, and $b$ does.
\item $a$ an origination atom, $b$ a channel assumption: take $I_c$ if $b =
  \auth(c)$ and $R_c$ if $b = \conf(c)$, and if $a = \Unq(v)$ add one create
  strand $\langle +v\rangle$. The bundle admits $a$ and uses a capability $b$
  forbids.
\item $a = \Non(v)$, $b = \Non(w)$ with $w \neq v$: one create strand
  $\langle +w\rangle$.
\item $a = \Non(v)$, $b = \Unq(w)$: one create strand $\langle +u\rangle$,
  in which $w$ originates nowhere.
\item $a = \Unq(v)$, $b = \Non(w)$: one create strand $\langle +v\rangle$,
  with a second create strand $\langle +w\rangle$ if $w \neq v$.
\item $a = \Unq(v)$, $b = \Unq(w)$ with $w \neq v$: one create strand
  $\langle +v\rangle$, in which $w$ originates nowhere.
\end{itemize}
Every ordered pair of distinct elements of $U$ falls under one case, so no
element of $U$ is at least as strong as another.
\end{proof}

Incomparability is over all strand spaces, as the definition of strength
states. Within one protocol and situation, two atoms of $U$ may be
comparable, and $\sigma$ may make an atom unsatisfiable, for instance
when $\sigma$ itself has a value originate twice and $U$ offers $\Unq$ of
that value.

$\Non(v)$ and $\Unq(v)$ are incomparable but admit no bundle in common.
Thus, a trust containing both makes the hypothesis unsatisfiable and
achieves every goal vacuously. A choice of $U$ that offers both atoms for
one value invites such answers. In practice, $U$ offers $\Non$ for keys and $\Unq$ for
nonces, as in the running example, where $U = \{\Non(sk_A), \Non(sk_B),
\Unq(N)\}$.

\subsection{Trusts, Goals, and Achievement}

With $U$ fixed, a trust is a subset of $U$, a goal is a conclusion, and
$\ach$ relates them. The two properties of $\ach$ that later sections use
are Lemma~\ref{lem:monotone} and Lemma~\ref{lem:polarity}.

\begin{definition}[Trust lattice]
\label{def:trustlat}
The \emph{trust lattice} is $\trustlat = 2^{U}$ ordered by inclusion. A
trust $\tau \in \trustlat$ splits into its origination part $\tau_O$, the
atoms $\Non(v)$ and $\Unq(v)$, read as the conjunction $\bigwedge \tau_O$ in
the hypothesis, and its channel part $\tau_C$, the assumptions $\auth(c)$ and
$\conf(c)$, read as the restriction to $\mathcal{B}(\tau_C)$.
\end{definition}

We read $\tau \le \tau'$ as ``$\tau'$ assumes at least as much''.
$\trustlat$ is a finite Boolean lattice: meets are intersections, joins are
unions, and every subset has both. By Lemma~\ref{lem:generators}, nothing
is lost by taking subsets of $U$ in place of arbitrary combinations of
levels. By its incomparability clause, no candidate assumption implies
another across all strand spaces. In the running example, $\trustlat$ has
eight elements, from $\emptyset$ to all of $U$.

\begin{definition}[Achievement under a trust]
\label{def:ach}
Fix $\Pi$ and $\sigma$, with free variables $\bar{x}$. For $\tau \in
\trustlat$ and a conclusion $\chi$ whose free variables are among $\bar{x}$,
write $\tau \ach \chi$ if for every bundle $\mathcal{B} \in
\mathcal{B}(\tau_C)$ and every assignment $\eta$,
\[
\mathcal{B}, \eta \models \sigma \wedge \textstyle\bigwedge \tau_O
\quad\text{implies}\quad
\mathcal{B}, \eta \models \chi .
\]
That is, $\Pi$ achieves $\forall\bar{x}.\,(\sigma \wedge \tau \Rightarrow
\chi)$ in the model restricted by $\tau$'s channel assumptions. With no
channel assumptions this is achievement in the sense of
Definition~\ref{def:achieve}.
\end{definition}

In the running example, with $\chi$ recent agreement,
$\{\Non(sk_A), \Unq(N)\} \ach \chi$ holds by the argument of
Section~\ref{sec:example}, while $\{\Non(sk_A)\} \ach \chi$ fails because
the replay of Fig.~\ref{fig:runs}(b) satisfies $\sigma \wedge
\Non(sk_A)$ and violates $\chi$.

\begin{definition}[The goal order]
\label{def:goalorder}
$\goals$ is the set of conclusions with free variables among $\bar{x}$,
identified up to equivalence over the bundles of $\Pi$, and ordered by
\[
\chi \le \chi' \iff \text{for every bundle } \mathcal{B} \text{ of } \Pi
\text{ and every } \eta,\ \mathcal{B}, \eta \models \chi' \text{ implies }
\mathcal{B}, \eta \models \chi .
\]
\end{definition}

Thus, $\chi \le \chi'$ means that $\chi'$ implies $\chi$, so a stronger
goal is higher. The greatest element is $\mathsf{false}$, the secrecy
conclusion, and the least is the conclusion that holds everywhere. The
conclusion fragment is closed under both connectives up to equivalence. A
disjunction of conclusions is a conclusion, and a conjunction of two
conclusions is one after distributing $\wedge$ over $\vee$ and renaming the
existential variables apart. Hence, $\goals$ is a lattice whose join is conjunction and
whose meet is disjunction.

\paragraph{Why the order on goals is reversed.}
Definition~\ref{def:goalorder} reverses the implication order. With it,
both orders put ``more'' at the top: more assumed in $\trustlat$, more
guaranteed in $\goals$. Three things then follow directly. The strongest
goal of \cite{rowe} is a greatest element, as its name says. The trust
required by a conjunction of goals is the join of the trusts required by
each. Achievement is monotone in the
trust and antitone in the goal (Lemma~\ref{lem:polarity}). Therefore, a
stronger goal has fewer sufficient trusts, and the map $\chi \mapsto
\Suff(\chi)$ is antitone under inclusion. The weakest-trust map
$\weakest$, which picks the least element of $\Suff(\chi)$ where one
exists, and the strongest-goal map $\strongest$ are monotone. Under the implication order, $\weakest$ and
$\strongest$ would be antitone. The content is the
same, and Section~\ref{sec:galois} uses the monotone form.

The next lemma makes A4 explicit. Together with its refinement
Lemma~\ref{lem:separable} (A6), it is the only property of the protocol
semantics that the rest of the paper uses.

\begin{lemma}[Assuming more cannot hurt]
\label{lem:monotone}
If $\tau \subseteq \tau'$ and $\tau \ach \chi$, then $\tau' \ach \chi$.
That is, if $\Pi$ achieves
$\forall\bar{x}.\,(\sigma \wedge \tau \Rightarrow \chi)$ in the adversary
model restricted by $\tau$'s channel assumptions, then it achieves the
corresponding sentence for $\tau'$.
\end{lemma}

\begin{proof}
A larger trust leaves fewer runs that could break the goal, in both its
channel part and its origination part. For a trust $\tau$ let
\[
X(\tau) = \{(\mathcal{B}, \eta) : \mathcal{B} \in \mathcal{B}(\tau_C),\
\mathcal{B}, \eta \models \sigma \wedge \textstyle\bigwedge \tau_O\},
\]
the pairs that could witness a failure of $\chi$. By
Definition~\ref{def:ach}, $\tau \ach \chi$ if and only if every pair in
$X(\tau)$ satisfies $\chi$. It suffices to show $X(\tau') \subseteq
X(\tau)$. Let $(\mathcal{B}, \eta) \in X(\tau')$.

\emph{Channel part.} $\tau_C \subseteq \tau'_C$, so the capabilities
forbidden by $\tau_C$ are among those forbidden by $\tau'_C$. A bundle none
of whose edges uses a capability forbidden by $\tau'_C$ therefore uses none
forbidden by $\tau_C$, and $\mathcal{B}(\tau'_C) \subseteq
\mathcal{B}(\tau_C)$. Hence, $\mathcal{B} \in \mathcal{B}(\tau_C)$.

\emph{Origination part.} $\tau_O \subseteq \tau'_O$, so $\bigwedge \tau'_O$
contains every conjunct of $\bigwedge \tau_O$, and $\mathcal{B}, \eta
\models \sigma \wedge \bigwedge \tau'_O$ gives $\mathcal{B}, \eta \models
\sigma \wedge \bigwedge \tau_O$.

Hence $(\mathcal{B}, \eta) \in X(\tau)$, and if every pair in $X(\tau)$
satisfies $\chi$ then so does every pair in $X(\tau')$.
\end{proof}

The two parts act differently. An origination atom shrinks the set of
runs a goal speaks about, while a channel assumption shrinks the set of
runs that exist. Both act monotonically, and nothing later needs to
distinguish them except Lemma~\ref{lem:separable}, which uses the same
split.

Two sets describe the relation from each side: the trusts that suffice for
a goal, and the goals a trust achieves.

\begin{definition}[Sufficient trusts and achieved goals]
\label{def:suffach}
\[
\Suff(\chi) = \{\tau \in \trustlat : \tau \ach \chi\},
\qquad
\Ach(\tau) = \{\chi \in \goals : \tau \ach \chi\}.
\]
\end{definition}

In the running example, $\Suff(\chi)$ consists of the two trusts that
contain both $\Non(sk_A)$ and $\Unq(N)$, with or without $\Non(sk_B)$.

\begin{lemma}[Polarity]
\label{lem:polarity}
$\Suff(\chi)$ is an up-set of $\trustlat$ and $\Ach(\tau)$ is a down-set of
$\goals$. Thus, $\ach$ is monotone in the trust and antitone in the goal.
\end{lemma}

\begin{proof}
A larger trust achieves at least as much, and a goal implied by an
achieved goal is achieved. The first claim is Lemma~\ref{lem:monotone}.
For the second, let $\chi \le \chi'$ and $\tau \ach \chi'$. For every
bundle $\mathcal{B} \in \mathcal{B}(\tau_C)$ and every $\eta$ satisfying
$\sigma \wedge \bigwedge \tau_O$, $\mathcal{B}, \eta \models \chi'$. Since
$\mathcal{B}$ is a bundle of $\Pi$, Definition~\ref{def:goalorder} gives
$\mathcal{B}, \eta \models \chi$. Hence, $\tau \ach \chi$. This is the form of
\cite[Lemma~1]{rowe} for restricted models.
\end{proof}

Lemma~\ref{lem:polarity} is the only structural input to the next section.

\section{Why the Weakest Trust Needs Either/Or Assumptions}
\label{sec:galois}

In the two-key variant of Section~\ref{sec:example}, the trusts
$\{\Non(k_1), \Unq(N)\}$ and $\{\Non(k_2), \Unq(N)\}$ both achieve recent
agreement, and neither contains the other. However, their intersection
$\{\Unq(N)\}$ does not achieve it. Thus, there is no single weakest trust.
By contrast, the strongest goal of Rowe et al.\ is unique whenever their
analysis terminates.

Rowe et al.\ ask the forward question: the strongest goal a trust
achieves. It has one answer whenever the analysis terminates, because the
goal language allows disjunctions in conclusions. When the analysis finds
several shapes, the strongest goal says that one of them occurred, and
that disjunction is again a goal, by Theorem~\ref{thm:sas}. The backward
question asks for the weakest trust that achieves a goal, and it requires
the same disjunction in the assumptions. For the two-key variant, the
answer is ``the first trust or the second''. Hypotheses, however, are
restricted to conjunctions of atoms and cannot express this. When
either/or assumptions, hypotheses that are disjunctions of trusts, are
allowed, every goal has exactly one weakest trust, as
Theorem~\ref{thm:completion} and Corollary~\ref{cor:diagnosis} show.

A Galois connection between trusts and goals explains both facts. It
consists of two maps: $\weakest$, from goals to trusts, and $\strongest$,
from trusts to goals. When these maps exist, a trust $\tau$ is at least $\weakest(\chi)$
exactly when it achieves $\chi$, or equivalently, when $\chi$ is at most
$\strongest(\tau)$. However, these maps need not always exist. Birkhoff's
polarities \cite{birkhoff}, on which formal concept analysis rests
\cite{ganter}, give every relation between two sets an antitone Galois
connection between the subsets of each. For achievement, since
$\Suff(\chi)$ is an up-set and $\Ach(\tau)$ a down-set, the polarity connects up-sets of trusts with
down-sets of goals. The maps $\weakest$ and $\strongest$ between
individual trusts and goals are a different object. They exist exactly
when every $\Suff(\chi)$ has a least element and every $\Ach(\tau)$ has a
greatest element, as the paragraph after Theorem~\ref{thm:galois}
explains. Theorem~\ref{thm:nofilter} shows that the first condition can
fail. Goals are ordered as in Definition~\ref{def:goalorder}, with
stronger goals higher, up to equivalence. The two questions become two
maps.

\begin{definition}[The two maps]
\label{def:maps}
$\strongest(\tau)$ is the greatest element of $\Ach(\tau)$, the strongest
goal achieved under $\tau$, and $\weakest(\chi)$ is the least element of
$\Suff(\chi)$, the weakest trust achieving $\chi$, whenever these elements
exist.
\end{definition}

The backward map can fail to exist. Because $\trustlat$ is a finite
lattice, the intersection $\bigcap \Suff(\chi)$ of all sufficient trusts
always exists in $\trustlat$. Definition~\ref{def:maps} asks for more: this
intersection must itself achieve $\chi$, so that $\Suff(\chi)$ has a least
element. In the two-key variant, it does not, as
Theorem~\ref{thm:nofilter} shows.

When both maps exist, four facts follow from the order alone. Both maps
are monotone: a stronger goal has a weakest trust at least as large,
and a larger trust has a strongest goal at least as strong. The weakest
trust for a join of goals, the weakest goal that implies each of
them, is the union of their weakest trusts, whenever that join exists in
$\goals$. Dually, the strongest goal achieved under the intersection of
several trusts is the strongest goal implied by each of their strongest
goals. A round trip from a trust to its strongest goal
and back to the weakest trust for that goal can only drop assumptions, and
a second round trip drops nothing more. A round trip from a goal to its
weakest trust and back can only strengthen the goal, and again a second
round trip changes nothing.

Order theory has names for these facts, and the statement below uses them.
The map $\weakest$ is the \emph{left adjoint} and $\strongest$ the
\emph{right adjoint} of the connection, written $\weakest \dashv
\strongest$. A \emph{kernel} (or interior) operator on a poset $P$ is a map
$\kappa : P \to P$ that is monotone, deflationary ($\kappa(p) \le p$), and
idempotent ($\kappa(\kappa(p)) = \kappa(p)$). A \emph{closure} operator is
the dual: monotone, inflationary ($p \le \kappa(p)$), and idempotent. The
round trips are a kernel operator on trusts and a closure operator on
goals.

\begin{theorem}[Achievement is a Galois connection]
\label{thm:galois}
Suppose $\strongest(\tau)$ exists for every $\tau$ and $\weakest(\chi)$
exists for every $\chi$. Then $\strongest$ and $\weakest$ are monotone and,
for all $\tau$ and $\chi$,
\begin{equation}
\weakest(\chi) \le \tau
\iff
\tau \ach \chi
\iff
\chi \le \strongest(\tau).
\label{eq:adjunction}
\end{equation}
That is, $\weakest \dashv \strongest$. Consequently, $\weakest$ preserves
all joins that exist in $\goals$, $\strongest$ preserves all meets that
exist in $\trustlat$, $\weakest \circ \strongest$ is a kernel operator on
trusts, and $\strongest \circ \weakest$ is a closure operator on goals.
\end{theorem}

\begin{proof}
The two equivalences are the definitions of least and greatest element
combined with the up-set and down-set properties of
Lemma~\ref{lem:polarity}, and the rest follows from them by order theory
alone. Since the proof uses nothing about $\ach$ other than
Lemma~\ref{lem:polarity} and the two existence hypotheses,
Theorem~\ref{thm:completion} can reuse it with another lattice in place of
$\trustlat$.

\emph{The equivalences.} Fix $\tau$ and $\chi$. If $\tau \ach \chi$, then
$\tau \in \Suff(\chi)$, and $\weakest(\chi)$, the least element of
$\Suff(\chi)$, satisfies $\weakest(\chi) \le \tau$. Conversely, if
$\weakest(\chi) \le \tau$, then since $\weakest(\chi) \in \Suff(\chi)$ and
$\Suff(\chi)$ is an up-set, $\tau \in \Suff(\chi)$, that is, $\tau \ach
\chi$. This is the left equivalence. For the right one, if $\tau \ach
\chi$, then $\chi \in \Ach(\tau)$, so $\chi \le \strongest(\tau)$, the
greatest element of $\Ach(\tau)$. Conversely, if $\chi \le
\strongest(\tau)$, then since $\strongest(\tau) \in \Ach(\tau)$ and
$\Ach(\tau)$ is a down-set, $\chi \in \Ach(\tau)$, that is, $\tau \ach
\chi$. Chaining the two gives \eqref{eq:adjunction}, which is the defining
property of an adjunction $\weakest \dashv \strongest$ between posets
\cite{birkhoff,ganter}.

\emph{Unit and counit.} Taking $\tau = \weakest(\chi)$ in
\eqref{eq:adjunction}, the left side holds, so
\begin{equation}
\chi \le \strongest(\weakest(\chi)) \quad\text{for every } \chi.
\label{eq:unit}
\end{equation}
Taking $\chi = \strongest(\tau)$, the right side holds, so
\begin{equation}
\weakest(\strongest(\tau)) \le \tau \quad\text{for every } \tau.
\label{eq:counit}
\end{equation}

\emph{Monotonicity.} Let $\chi_1 \le \chi_2$. By \eqref{eq:unit},
$\chi_1 \le \chi_2 \le \strongest(\weakest(\chi_2))$, and
\eqref{eq:adjunction}, read from right to left with $\chi = \chi_1$ and
$\tau = \weakest(\chi_2)$, gives $\weakest(\chi_1) \le \weakest(\chi_2)$.
Dually, let $\tau_1 \le \tau_2$. By \eqref{eq:counit}, $\weakest(\strongest(\tau_1)) \le
\tau_1 \le \tau_2$, and \eqref{eq:adjunction} with $\chi =
\strongest(\tau_1)$ gives $\strongest(\tau_1) \le \strongest(\tau_2)$.

\emph{Joins and meets.} Let $\{\chi_i\}_{i \in I}$ be a family of goals
whose join $\chi = \bigvee_i \chi_i$ exists in $\goals$. For every trust
$\tau$,
\[
\weakest(\chi) \le \tau
\iff \chi \le \strongest(\tau)
\iff \chi_i \le \strongest(\tau) \text{ for all } i
\iff \weakest(\chi_i) \le \tau \text{ for all } i,
\]
using \eqref{eq:adjunction} twice and the definition of a join in between.
Taking $\tau = \weakest(\chi)$ shows that $\weakest(\chi)$ is an upper bound
of the $\weakest(\chi_i)$. Taking $\tau$ an arbitrary upper bound of them
shows $\weakest(\chi) \le \tau$. So $\weakest(\chi)$ is the least upper
bound: $\weakest(\bigvee_i \chi_i) = \bigvee_i \weakest(\chi_i)$, and in
particular the join on the right exists. The argument for $\strongest$ is
the mirror image: if $\tau = \bigwedge_i \tau_i$ exists in $\trustlat$,
then for every $\chi$, $\chi \le \strongest(\tau)$ iff $\weakest(\chi) \le
\tau$ iff $\weakest(\chi) \le \tau_i$ for all $i$ iff $\chi \le
\strongest(\tau_i)$ for all $i$, so $\strongest(\tau)$ is the greatest
lower bound of the $\strongest(\tau_i)$.

\emph{The composites.} Apply $\strongest$, which is monotone, to
\eqref{eq:counit}: $\strongest(\weakest(\strongest(\tau))) \le
\strongest(\tau)$. Apply \eqref{eq:unit} with $\chi = \strongest(\tau)$:
$\strongest(\tau) \le \strongest(\weakest(\strongest(\tau)))$. Hence
\begin{equation}
\strongest \circ \weakest \circ \strongest = \strongest,
\qquad\text{and dually}\qquad
\weakest \circ \strongest \circ \weakest = \weakest,
\label{eq:triangle}
\end{equation}
the second because $\weakest$, applied to \eqref{eq:unit}, gives
$\weakest(\chi) \le \weakest(\strongest(\weakest(\chi)))$, and
\eqref{eq:counit} at $\tau = \weakest(\chi)$ gives the reverse. Now
$\weakest \circ \strongest$ is monotone as a composite of monotone maps, deflationary by \eqref{eq:counit}, and
idempotent because $\weakest \strongest \weakest \strongest = \weakest
\strongest$ by \eqref{eq:triangle}, so it is a kernel operator. Likewise
$\strongest \circ \weakest$ is monotone, inflationary by \eqref{eq:unit}
and idempotent by \eqref{eq:triangle}, so it is a closure operator.
\end{proof}

Theorem~\ref{thm:galois} assumes that both maps exist, and whether they do
is the point at issue. Each side has an exact answer. The map
$\strongest(\tau)$ exists whenever enrich-by-need terminates from $\sigma
\wedge \tau$, and we assume this.

\begin{description}
\item[Termination assumption.] For every trust $\tau \in \trustlat$,
  enrich-by-need started from the skeleton of $\sigma \wedge \tau$
  terminates, in the model restricted by $\tau$'s channel assumptions.
\end{description}

\noindent The assumption covers every atom of $U$, origination and channel
atoms alike, and it is not guaranteed in general. For trusts of
origination atoms, \cite[Theorem~2]{rowe}, restated as
Theorem~\ref{thm:sas}, then characterises
$\strongest(\tau)$ as the shape analysis sentence. That theorem
presupposes termination and does not supply it. A channel atom restricts
the adversary and leaves the hypothesis unchanged. For channel atoms, we
also assume that the characterisation carries over to the restricted
model, a case that \cite{rowe} does not cover and that
Section~\ref{sec:lim-channels} discusses. The left adjoint need not exist. An
up-set of a lattice has a least element exactly when it is a
\emph{principal filter}, the set of all elements above a single element.
Hence, $\weakest(\chi)$ exists in $\trustlat$ if and only if $\Suff(\chi)$
is a principal filter. By Theorem~\ref{thm:unique}, this happens exactly
when every minimal attack has a single stopping assumption. It fails for
the two-key variant of Section~\ref{sec:example}, but $\weakest(\chi)$
always exists in the disjunctive completion, as
Theorem~\ref{thm:completion} shows. Where the
hypothesis of Theorem~\ref{thm:galois} fails, the completion supplies the
left adjoint.

The two-key variant witnesses that a single weakest trust can fail to
exist.

\begin{theorem}[The sufficient trusts do not form a filter]
\label{thm:nofilter}
There are a protocol, a situation, and a goal $\chi$ with $\tau_1, \tau_2
\in \Suff(\chi)$ and $\tau_1 \cap \tau_2 \notin \Suff(\chi)$. Hence
$\Suff(\chi)$ has no least element in $\trustlat$ and $\weakest(\chi)$ is
undefined there.
\end{theorem}

\begin{proof}
Two trusts that each protect a different key both stop the forgery, and
their intersection protects neither. Take the two-key variant of
Section~\ref{sec:example}, in which the response is nested under $k_1$ and
$k_2$. The situation is $B$'s completed challenger run, and $\chi$ is
$B$'s recent agreement with $A$, over $U = \{\Non(k_1), \Non(k_2),
\Unq(N)\}$. By Section~\ref{sec:example} (Table~\ref{tab:example}), there
are two minimal attacks. In the forgery, the adversary builds the response
and must therefore hold both keys, so protecting either key removes it. The replay
returns a response from an earlier session that used $N$, and only the
freshness of $N$ removes it. Put
\[
\tau_1 = \{\Non(k_1), \Unq(N)\},
\qquad
\tau_2 = \{\Non(k_2), \Unq(N)\}.
\]
Both are sufficient. Each stops the forgery run by one of its keys and
the replay run by freshness (formally, each meets both stopping sets
$\{\Non(k_1),\Non(k_2)\}$ and $\{\Unq(N)\}$, and
Proposition~\ref{prop:hitting} applies). Their intersection $\tau_1 \cap
\tau_2 = \{\Unq(N)\}$ leaves both keys available to the adversary, and the
forgery run survives it, so the intersection is not sufficient. If
$\Suff(\chi)$ had a least element $m$, it would satisfy $m \subseteq \tau_1
\cap \tau_2$, and the up-set property would force $\tau_1 \cap \tau_2 \in
\Suff(\chi)$, a contradiction. Section~\ref{sec:eval} recovers the same
minimal attacks mechanically.
\end{proof}

Thus, sufficient trusts are not closed under intersection. The same
property separates predicate transformers that have a weakest
precondition from those that do not \cite{dijkstra}, as
Section~\ref{sec:related} discusses. We address this by allowing either/or
trusts. They form the disjunctive completion of $\trustlat$, whose
elements are \emph{antichains}: families of trusts in which no member
contains another, each read as ``one of these trusts holds''.

\begin{definition}[Disjunctive completion]
\label{def:completion}
$\completion(\trustlat)$ is the lattice of antichains of subsets of $U$,
an antichain $M$ denoting the up-set it generates, ordered by
\[
M \le M' \iff \text{every } T' \in M' \text{ contains some } T \in M.
\]
Equivalently, $\completion(\trustlat)$ is the lattice of monotone Boolean
functions over $U$, with $M$ denoting
$\bigvee_{T \in M} \bigwedge_{a \in T} a$. Equivalently again, it is the
free bounded distributive lattice on $U$ \cite{birkhoff}. The embedding $\trustlat
\hookrightarrow \completion(\trustlat)$ sends $\tau$ to $\{\tau\}$.
\end{definition}

For an antichain $M$ write $\mathop{\uparrow} M = \{\tau \in \trustlat :
\tau \supseteq T \text{ for some } T \in M\}$ for the up-set it generates.
The order of Definition~\ref{def:completion} is reverse inclusion of these
up-sets: $M \le M'$ iff $\mathop{\uparrow} M' \subseteq \mathop{\uparrow}
M$. Every up-set $F$ of the finite lattice $\trustlat$ is generated by its
set of minimal elements, $F = \mathop{\uparrow} \min F$, and an antichain
is the set of minimal elements of the up-set it generates. Hence, $M
\mapsto \mathop{\uparrow} M$ is a bijection between antichains and up-sets, and
$\completion(\trustlat)$ is order isomorphic to the up-sets of $\trustlat$
under reverse inclusion.

In particular, $\completion(\trustlat)$ is a finite lattice. The join of
$M$ and $M'$ is the antichain generating $\mathop{\uparrow} M \cap
\mathop{\uparrow} M'$, namely the minimal elements of $\{T \cup T' : T \in
M, T' \in M'\}$, and their meet is the antichain generating
$\mathop{\uparrow} M \cup \mathop{\uparrow} M'$, namely the minimal
elements of $M \cup M'$. The least element is $\{\emptyset\}$, which
generates all of $\trustlat$ and denotes $\mathsf{true}$. The greatest is
the empty antichain, which generates the empty up-set and denotes
$\mathsf{false}$. The map $\tau \mapsto \{\tau\}$ is an order embedding,
since $\{\tau\} \le \{\tau'\}$ iff $\tau \subseteq \tau'$. It preserves
joins, since $\{\tau\} \vee \{\tau'\} = \{\tau \cup \tau'\}$, but not
meets. When $\tau$ and $\tau'$ are incomparable, $\{\tau\} \wedge \{\tau'\}
= \{\tau, \tau'\}$, whereas their meet in $\trustlat$ is $\tau \cap
\tau'$. The added meets are exactly what $\trustlat$ lacked in
Theorem~\ref{thm:nofilter}.

In the assumes-more order, an either/or trust assumes less than each of
its disjuncts. Thus, the antichain $\{\tau_1, \tau_2\}$ lies below
$\{\tau_1\}$ and $\{\tau_2\}$, and strictly above $\{\tau_1 \cap \tau_2\}$
when the two are incomparable. The added elements are meets in the order
of $\trustlat$ and joins in the implication order of formulas. We call
them disjunctions to avoid committing to either order. The completion
$\completion(\trustlat)$ is finite, but its number of elements is a
Dedekind number, the number of monotone Boolean functions of $|U|$
variables, which grows doubly exponentially in $|U|$. Finiteness makes the
adjoint computable but not cheap, and Section~\ref{sec:cost} states what is
known about the cost.

The disjunctive completion is the standard way to add the missing
disjunctions and nothing else. Abstract interpretation uses it to make a
domain precise enough for disjunctive analysis \cite{cousot}. In the goal
language, an element of $\completion(\trustlat)$ is a hypothesis in
monotone disjunctive normal form.

Stating the adjunction over $\completion(\trustlat)$ requires extending
$\ach$ to either/or trusts. For $M \in \completion(\trustlat)$ put
\begin{equation}
M \ach \chi \iff T \ach \chi \text{ for every } T \in M.
\label{eq:disjach}
\end{equation}
This is what achieving $\forall\bar{x}.\,(\sigma \wedge \bigvee_{T \in M}
\bigwedge T \Rightarrow \chi)$ means. A run satisfies the disjunctive hypothesis exactly when it
satisfies the hypothesis of some disjunct. Hence, every run satisfying it
satisfies $\chi$ exactly when this holds for each disjunct separately. For
channel assumptions, which restrict the model rather than the hypothesis,
\eqref{eq:disjach} means that either deployment is acceptable: the goal
must hold whichever disjunct is in force. On the embedded trusts,
$\{\tau\} \ach \chi$ exactly when $\tau \ach \chi$, so \eqref{eq:disjach}
extends $\ach$. It is monotone in the trust. If $M \le M'$ and $M \ach
\chi$, each $T' \in M'$ contains some $T \in M$ with $T \ach \chi$, so $T'
\ach \chi$ by Lemma~\ref{lem:monotone}, and $M' \ach \chi$. It is
antitone in the goal, since each $T \ach \chi$ is.

Over either/or trusts the weakest trust always exists and is unique.

\begin{theorem}[The left adjoint exists in the completion]
\label{thm:completion}
For every goal $\chi$, let $\weakest(\chi) \in \completion(\trustlat)$ be
the antichain of minimal elements of $\Suff(\chi)$. Then for every $\tau
\in \trustlat$,
\[
\tau \ach \chi \iff \weakest(\chi) \le \{\tau\},
\]
$\weakest(\chi)$ is the unique element of $\completion(\trustlat)$ with
this property, and $\weakest \dashv \strongest$ holds as in
Theorem~\ref{thm:galois} with $\weakest : \goals \to
\completion(\trustlat)$. Written as a formula, $\weakest(\chi)$ is the
monotone disjunctive normal form
$\bigvee_{T \in \weakest(\chi)} \bigwedge_{a \in T} a$, whose prime
implicants are exactly the minimal sufficient trusts.
\end{theorem}

\begin{proof}
An up-set of a finite lattice is determined by its minimal elements, and
$\completion(\trustlat)$ is built so that this set of minimal elements is
one of its elements. We write $M = \weakest(\chi)$, the set of minimal
elements of $\Suff(\chi)$. It is an antichain, so $M \in
\completion(\trustlat)$, and since $\Suff(\chi)$ is an up-set of the finite
lattice $\trustlat$ (Lemma~\ref{lem:monotone}), $\Suff(\chi) =
\mathop{\uparrow} M$.

\emph{The displayed equivalence.} By Definition~\ref{def:completion},
$M \le \{\tau\}$ iff $\tau$ contains some $T \in M$, iff $\tau \in
\mathop{\uparrow} M = \Suff(\chi)$, iff $\tau \ach \chi$.

\emph{Uniqueness.} Suppose $M' \in \completion(\trustlat)$ also satisfies
$\tau \ach \chi \iff M' \le \{\tau\}$ for every $\tau \in \trustlat$. As
computed above, $M' \le \{\tau\}$ iff $\tau \in \mathop{\uparrow} M'$, so
$\mathop{\uparrow} M' = \Suff(\chi) = \mathop{\uparrow} M$. An antichain
is the set of minimal elements of the up-set it generates, so $M' = M$.

\emph{The adjunction.} Extend $\ach$ to $\completion(\trustlat)$ by
\eqref{eq:disjach}. For the left adjoint: for every $M' \in
\completion(\trustlat)$,
\[
M' \ach \chi
\iff M' \subseteq \Suff(\chi)
\iff \mathop{\uparrow} M' \subseteq \Suff(\chi) = \mathop{\uparrow} M
\iff M \le M',
\]
where the middle step uses that $\Suff(\chi)$ is an up-set. Therefore, the trusts
in $\completion(\trustlat)$ achieving $\chi$ form the up-set of
$\completion(\trustlat)$ with least element $M$, which is the existence
hypothesis of Theorem~\ref{thm:galois} on the $\weakest$ side. For the right
adjoint, assume as in Theorem~\ref{thm:galois} that $\strongest(T)$ exists
for every $T \in \trustlat$ (under the termination assumption above, by
Theorem~\ref{thm:sas}).
By \eqref{eq:disjach}, $M' \ach \chi$ iff $\chi \le \strongest(T)$ for every
$T \in M'$, iff $\chi$ is a lower bound of $\{\strongest(T) : T \in M'\}$ in
$\goals$. The greatest such lower bound is the disjunction of the
$\strongest(T)$: it is implied by each $\strongest(T)$, so it lies below
each in the strength order, and any goal implied by every $\strongest(T)$
is implied by their disjunction. The conclusion fragment of \cite{rowe} is
closed under disjunction, so this goal is in $\goals$, and it is achieved
under $M'$, as each disjunct $T$ achieves its own $\strongest(T)$ and hence
the disjunction. When $M'$ is the empty antichain, the greatest element of
$\completion(\trustlat)$ (not to be confused with $\{\emptyset\}$, the
embedded empty trust), the disjunction is empty. In this case,
$\strongest$ of the empty antichain is $\mathsf{false}$, the greatest goal,
which the empty disjunction of trusts achieves vacuously. Thus,
$\strongest$ is defined
on all of $\completion(\trustlat)$, with
$\strongest(M') = \bigvee_{T \in M'} \strongest(T)$ read as a disjunction
of formulas, which is the meet of the $\strongest(T)$ in $\goals$. Both
existence hypotheses of Theorem~\ref{thm:galois} now hold for $\ach$ on
$\completion(\trustlat) \times \goals$, and its proof, which used only these
and the monotonicity checked above, gives $\weakest \dashv \strongest$ with all
the stated consequences.

\emph{The formula.} The monotone Boolean function denoted by $M$ takes the
value $1$ on $\tau$ exactly when $\tau \in \mathop{\uparrow} M$, so it is
the characteristic function of $\Suff(\chi)$. An implicant of a monotone
function that is a conjunction of positive atoms $\bigwedge_{a \in T} a$ is
a set $T$ on which the function is $1$, and it is prime when no proper
subset is one. The prime implicants are therefore the minimal elements of
$\Suff(\chi)$, which is $M$.
\end{proof}

\paragraph{The completion on the two-key variant.} Abbreviate the atoms of
the two-key variant as $U = \{k_1, k_2, N\}$, for $\Non(k_1)$, $\Non(k_2)$,
and $\Unq(N)$. Then $\trustlat$ has $8$ elements and $\completion(\trustlat)$
has $20$, the Dedekind number for three atoms. By
Theorem~\ref{thm:nofilter} and Proposition~\ref{prop:hitting},
$\Suff(\chi)$ consists of the trusts that contain $N$ and at least one key:
\[
\Suff(\chi) = \{\{k_1, N\}, \{k_2, N\}, \{k_1, k_2, N\}\},
\qquad
\weakest(\chi) = \{\{k_1, N\}, \{k_2, N\}\}.
\]
As a formula, $\weakest(\chi)$ is $(k_1 \wedge N) \vee (k_2 \wedge N)$,
which is also $N \wedge (k_1 \vee k_2)$. The second form is the
conjunction, over the two minimal attacks, of the disjunction of what stops
each. Section~\ref{sec:blocker} turns this duality into the computation. In $\completion(\trustlat)$, $\weakest(\chi) = \{\{k_1,N\}\}
\wedge \{\{k_2,N\}\}$ is the meet of the two embedded sufficient
trusts. It lies strictly above $\{\{N\}\}$, the embedding of their meet
in $\trustlat$, which is not sufficient. The elements of
$\completion(\trustlat)$ above it are exactly the antichains all of whose
members contain $N$ and a key. These are the either/or trusts that achieve
$\chi$. By contrast, for the signed protocol of Section~\ref{sec:example},
$\weakest(\chi) = \{\{\Non(sk_A), \Unq(N)\}\}$ is a single trust, the image
of an element of $\trustlat$, and Theorem~\ref{thm:galois} applies without
the completion.

Read in the goal language, Theorem~\ref{thm:completion} says that the
missing disjunction is the only obstruction.

\begin{corollary}[Diagnosis]
\label{cor:diagnosis}
In the goal language of \cite{rowe} extended so that hypotheses may be
disjunctions of conjunctions of trust atoms, every goal has a unique
weakest sufficient hypothesis. For hypotheses built from trust atoms
(A2), the non-uniqueness reported in \cite[Section~6.4]{rowe} comes from
the conjunctive form of hypotheses: the same $\Suff(\chi)$ has a unique
weakest element when disjunctions are allowed.
\end{corollary}

\begin{proof}
An either/or hypothesis is an element of $\completion(\trustlat)$, so
Theorem~\ref{thm:completion} applies to it directly. A disjunction of
conjunctions of trust atoms is, up to logical equivalence and after
deleting disjuncts that contain another disjunct, an antichain $M \in
\completion(\trustlat)$, and it is sufficient for $\chi$ exactly when $M
\ach \chi$ in the sense of \eqref{eq:disjach}. By
Theorem~\ref{thm:completion}, these are the $M$ with $\weakest(\chi) \le
M$, so $\weakest(\chi)$ is the unique weakest of them. In the conjunctive
fragment, the candidates are the embedded trusts $\{\tau\}$. The minimal
sufficient ones among them are the members of $\weakest(\chi)$. Hence, a
weakest conjunctive hypothesis exists exactly when $\weakest(\chi)$ has
exactly one member. Otherwise, there are none or several. The protocol enters only through
$\Suff(\chi)$, which is the same in both fragments.
\end{proof}

The corollary is an order-theoretic fact: an up-set of a finite lattice
has a unique antichain of minimal elements. Its scope is the trust atoms
of A2. Rowe et al.\ ask about all conjunctions of atoms, including the
role-position and parameter atoms that A2 excludes. The corollary says
nothing about the non-uniqueness they report for those hypotheses. Because
the obstruction over trust atoms is a missing disjunction, it can be
removed. Abduction in general has no such repair.
Section~\ref{sec:lim-completion} discusses what the extension costs.

\section{Weakest Trust from Minimal Attacks}
\label{sec:blocker}

Theorem~\ref{thm:completion} says that $\weakest(\chi)$ exists and is
unique when either/or trusts are allowed, but not how to find it. The
runs that violate the goal give the answer. The weakest trust is the set
of minimal ways to stop every minimal attack: the smallest trusts that
contain a member of every minimal attack's stopping set. It is a single
trust exactly when every minimal attack has one remedy, a single
assumption that stops it. We show that $\weakest(\chi) =
\blocker(\attacks(\chi))$, where
$\attacks(\chi)$ is the family of minimal attacks,
defined below, and the \emph{blocker}
$\blocker(\mathcal{F})$ of a family of sets is the family of minimal sets
that meet every member. Theorem~\ref{thm:unique} gives the uniqueness
criterion, and Section~\ref{sec:example} computed both by hand.

A goal speaks about a bundle through an assignment $\eta$ of its nodes and
values to the variables of $\sigma$ and $\chi$, and the atoms of $U$ name
values through the same variables. Thus, whether an atom rules a run out
depends on $\eta$ as well as on the bundle. For example, one bundle may
contain two challenger strands of $B$: one whose challenge originates once
and one whose challenge was used before. Then $\Unq(N)$ holds when $\eta$
sends $N$ to the first challenge and fails when it sends $N$ to the
second. From here on, a \emph{run} is a pair $x = (\mathcal{B}, \eta)$ of
a bundle $\mathcal{B}$ of the unrestricted model and such an assignment.
Say that a trust $\tau$ \emph{rules out} $x$ when $\mathcal{B} \notin
\mathcal{B}(\tau_C)$ or $\mathcal{B}, \eta$ violates one of $\tau$'s
origination atoms. For a run that satisfies $\sigma$, the second
disjunct is failure of the hypothesis $\sigma \wedge \tau$ at
$(\mathcal{B}, \eta)$.

\begin{definition}[Attacks and what stops them]
\label{def:attacks}
An \emph{attack} is a pair $x = (\mathcal{B}, \eta)$ of a bundle
$\mathcal{B}$ of $\Pi$ and an assignment $\eta$ such that $\mathcal{B},
\eta \models \sigma$, with nothing assumed beyond it, and $\mathcal{B},
\eta \not\models \chi$. For an attack $x$ put
\[
\stops(x) = \{a \in U : \text{assuming } a \text{ rules out } x\}.
\]
This set contains the origination atoms the run violates and the channel
assumptions whose restriction removes it. Let $\attacks(\chi)$ be the
family of members of $\{\stops(x) : x \text{ an attack}\}$ that are
minimal under $\subseteq$. By construction, $\attacks(\chi)$ is
a \emph{clutter}, a family of sets none of which contains another. We call
its members the \emph{minimal attacks}, identifying an attack with its
stopping set.
\end{definition}

We use the same formula for $\stops(x)$ when $x$ is any run. Because
$\{\stops(x)\}$ is a family of subsets of the finite set $U$, it is finite
even when there are infinitely many attacks, and each of its members
contains a minimal one.

Definition~\ref{def:attacks} asks of each atom separately whether it rules
out $x$. Assumption A6 states that a set of atoms rules out $x$ exactly
when one of them does, and the atoms of $U$ satisfy it.

\begin{lemma}[Separability]
\label{lem:separable}
For every run $x$ and trust $\tau$, $\tau$ rules out $x$ if and only if
$\tau \cap \stops(x) \neq \emptyset$.
\end{lemma}

\begin{proof}
A run is a fixed bundle with a fixed assignment, so each atom either
removes it or does not, independently of the other atoms assumed. Let $x =
(\mathcal{B}, \eta)$, and split $\tau = \tau_O \cup \tau_C$ into its
origination part and its channel part (Definition~\ref{def:trustlat}). By
definition, $\tau$ rules out $x$ iff (i) $\mathcal{B}, \eta$ violates some
atom of $\tau_O$, or (ii) $\mathcal{B} \notin \mathcal{B}(\tau_C)$. These
are the two parts of the proof of Lemma~\ref{lem:monotone}.

(i) Whether $\mathcal{B}, \eta$ satisfies an atom $\Non(v)$ or $\Unq(v)$
is decided by where the value $\eta$ gives $v$ originates in
$\mathcal{B}$, and does not depend on which other atoms are assumed. The
origination part of the hypothesis is the conjunction of the atoms of
$\tau_O$, so $\mathcal{B}, \eta$ violates it iff it violates some $a \in
\tau_O$, iff some $a \in \tau_O$ rules out $x$ by itself. This step uses
that a run is one fixed bundle with one fixed assignment, as A6 requires.

(ii) A channel assumption $a$ forbids a set $F(a)$ of adversary
transitions, given in Definition~\ref{def:channels}: in the inject/read
reading, injection on $c$ for $\auth(c)$ and reading from $c$ for
$\conf(c)$. The model restricted by $\tau_C$
forbids every transition forbidden by some member, that is, the union
$\bigcup_{a \in \tau_C} F(a)$, and its runs are the runs of the unrestricted
model that use no forbidden transition (A4). A bundle uses a fixed set of
adversary transitions, which does not depend on $\eta$, so $\mathcal{B}
\notin \mathcal{B}(\tau_C)$ iff $\mathcal{B}$ uses a transition in
$\bigcup_{a \in \tau_C} F(a)$, iff it uses one in $F(a)$ for some $a \in
\tau_C$, iff some $a \in \tau_C$ removes $x$ by itself.

Combining, $\tau$ rules out $x$ iff some $a \in \tau$ rules out $x$ by
itself, iff some $a \in \tau$ lies in $\stops(x)$.
\end{proof}

Without separability, a run that only a
pair of atoms removes would have $\stops(x) =
\emptyset$, and Remark~\ref{rem:unrealisable} would wrongly report that no
trust suffices.

With separability, a trust suffices exactly when it stops every minimal
attack.

\begin{proposition}[Sufficiency is hitting]
\label{prop:hitting}
$\tau \ach \chi$ if and only if $\tau \cap E \neq \emptyset$ for every $E
\in \attacks(\chi)$.
\end{proposition}

\begin{proof}
A trust fails exactly when some attack survives it, and by separability an
attack survives exactly when the trust misses its stopping set. First,
$\tau \ach \chi$ iff $\tau$ rules out every attack. Indeed, $\tau
\ach \chi$ fails iff some run $x = (\mathcal{B}, \eta)$ with $\mathcal{B}
\in \mathcal{B}(\tau_C)$ satisfies $\sigma \wedge \tau$ and violates $\chi$
(Definition~\ref{def:ach}). Such an $x$ is a run of the unrestricted model
(restriction only removes bundles, A4) that meets $\sigma$ and violates
$\chi$, so it is an attack, and it is not ruled out by $\tau$. Conversely,
an attack $(\mathcal{B}, \eta)$ that $\tau$ does not rule out has
$\mathcal{B} \in \mathcal{B}(\tau_C)$, satisfies $\sigma \wedge \tau$ and
violates $\chi$, so $\tau \ach \chi$ fails.

By Lemma~\ref{lem:separable}, $\tau$ rules out the attack $x$ iff $\tau \cap
\stops(x) \neq \emptyset$. So $\tau \ach \chi$ iff $\tau$ meets $\stops(x)$
for every attack $x$. It remains to replace the stopping sets by the
minimal ones. Every member of $\attacks(\chi)$ is some $\stops(x)$, so
meeting all stopping sets implies meeting all members of $\attacks(\chi)$.
Conversely, every $\stops(x)$ contains some $E \in \attacks(\chi)$, and a
set that meets $E$ meets every superset of $E$.
\end{proof}

The hitting form by itself holds for any up-set. What depends on the
protocol is which sets are hit.

\begin{remark}[The proposition is generic]
\label{rem:generic}
Let $F \subseteq 2^{U}$ be an up-set whose complement has maximal
elements $M_1, \dots, M_k$. Then $\tau \in F$ iff $\tau \not\subseteq M_i$
for all $i$, iff $\tau$ meets $U \setminus M_i$ for all $i$. Thus, every up-set is a hitting-set family, and
Proposition~\ref{prop:hitting} follows from Lemma~\ref{lem:monotone}
alone and has no protocol content by itself. Its content is that the clauses may be taken to be the sets
$\stops(x)$ indexed by attacks, which a protocol analysis produces. That
step is Lemma~\ref{lem:separable}. In terms of
Theorem~\ref{thm:completion}, the attacks are the conjunctive normal form
of the same monotone function whose disjunctive normal form is the weakest
trust.
\end{remark}

For a family $\mathcal{F}$ of subsets of $U$, write $\mathrm{Tr}(\mathcal{F})$
for the set of its \emph{transversals}, the sets $S \subseteq U$ that meet
every member of $\mathcal{F}$, and $\blocker(\mathcal{F})$ for its
blocker, the clutter of minimal transversals. Every transversal contains a
minimal one and supersets of transversals are transversals, so
$\mathrm{Tr}(\mathcal{F}) = \mathop{\uparrow} \blocker(\mathcal{F})$.
Proposition~\ref{prop:hitting} says that $\Suff(\chi) =
\mathrm{Tr}(\attacks(\chi))$. Taking minimal elements, and recalling from
Theorem~\ref{thm:completion} that $\weakest(\chi)$ is the set of minimal
elements of $\Suff(\chi)$,
\[
\weakest(\chi) = \blocker(\attacks(\chi)).
\]
We use two facts about transversals below. Lemma~\ref{lem:refinement}
compares two families, and Lemma~\ref{lem:involution} is the involution of
Edmonds and Fulkerson and of Isbell \cite{edmonds,isbell,berge}.

\begin{lemma}[Refinement]
\label{lem:refinement}
For families $\mathcal{A}, \mathcal{B}$ of subsets of $U$,
$\mathrm{Tr}(\mathcal{A}) \subseteq \mathrm{Tr}(\mathcal{B})$ if and only
if every member of $\mathcal{B}$ contains a member of $\mathcal{A}$.
\end{lemma}

\begin{proof}
A set that meets a smaller set meets every larger one, and the complement
of a set is the largest set that misses it. ($\Leftarrow$) Let $S \in
\mathrm{Tr}(\mathcal{A})$ and $E \in \mathcal{B}$. Some $E' \in
\mathcal{A}$ has $E' \subseteq E$. Since $S$ meets $E'$, it meets $E$.
($\Rightarrow$) Let $E \in \mathcal{B}$. The set $U \setminus E$ misses
$E$, so it is not in $\mathrm{Tr}(\mathcal{B})$, hence not in
$\mathrm{Tr}(\mathcal{A})$: some $E' \in \mathcal{A}$ is disjoint from $U
\setminus E$, that is, $E' \subseteq E$.
\end{proof}

\begin{lemma}[Blocker involution]
\label{lem:involution}
For every clutter $\attacks$ over $U$,
\begin{equation}
\blocker(\blocker(\attacks)) = \attacks.
\label{eq:involution}
\end{equation}
\end{lemma}

\begin{proof}
The sets that meet every minimal transversal of $\attacks$ are exactly the
sets that contain a member of $\attacks$, and the minimal ones are the
members of $\attacks$. First, $\mathrm{Tr}(\blocker(\attacks))$ is the set
of $S \subseteq U$ that contain a member of $\attacks$. Indeed, $S$ meets
every member of $\blocker(\attacks)$ iff $S$ meets every member of
$\mathrm{Tr}(\attacks)$, because each transversal contains a minimal one.
Moreover, $S$ meets every member of $\mathrm{Tr}(\attacks)$ iff
$\mathrm{Tr}(\attacks) \subseteq \mathrm{Tr}(\{S\})$, iff, by
Lemma~\ref{lem:refinement} with $\mathcal{A} = \attacks$ and $\mathcal{B} =
\{S\}$, $S$ contains a member of $\attacks$. Hence,
$\mathrm{Tr}(\blocker(\attacks)) = \mathop{\uparrow} \attacks$, and
$\blocker(\blocker(\attacks))$, its set of minimal elements, is the set of
minimal elements of $\mathop{\uparrow} \attacks$. Because $\attacks$ is a
clutter, these are exactly the members of $\attacks$: each member of
$\attacks$ is minimal in $\mathop{\uparrow} \attacks$, as a smaller element
would contain another member, and each minimal element of
$\mathop{\uparrow} \attacks$ contains a member of $\attacks$ and so equals
it.
\end{proof}

In the degenerate cases, $\blocker(\emptyset) = \{\emptyset\}$, since
every set, including $\emptyset$, meets every member of the empty family,
and $\blocker(\{\emptyset\}) = \emptyset$, since no set meets $\emptyset$.
By the involution, the minimal attacks can be recovered from the weakest
trusts. Because $\blocker$ is its own inverse on clutters, it is injective
there: two clutters with the same blocker are equal. This settles the
uniqueness question in both directions.

\begin{theorem}[When the weakest trust is unique]
\label{thm:unique}
$|\weakest(\chi)| = 1$ if and only if every member of $\attacks(\chi)$ is
a singleton, in which case the unique weakest trust is
$\bigcup \attacks(\chi)$, one assumption per minimal attack. Equivalently,
two incomparable weakest trusts exist if and only if some minimal attack
has two or more stopping assumptions.
\end{theorem}

\begin{proof}
If every minimal attack has one remedy, every transversal must contain all
the remedies. Conversely, by the involution, the minimal attacks are the
blocker of the single weakest trust, whose members are singletons.

($\Leftarrow$) If every member of $\attacks(\chi)$ is a singleton, write
$T = \bigcup \attacks(\chi)$. Every transversal must contain each
singleton's element and hence contain $T$. Moreover, $T$ itself is a
transversal. Thus, $T$ is the unique minimal transversal and $\weakest(\chi) = \{T\}$. This
covers the case $\attacks(\chi) = \emptyset$, where $T = \emptyset$.

($\Rightarrow$) Suppose $\weakest(\chi) = \blocker(\attacks(\chi)) =
\{T\}$ for a single set $T$. Apply \eqref{eq:involution}:
$\attacks(\chi) = \blocker(\{T\})$. The minimal transversals of the
one-member clutter $\{T\}$ are exactly the singletons $\{a\}$ for $a \in
T$, since a set meets $T$ iff it contains some element of $T$, and the
minimal such sets are the singletons. Hence every member of
$\attacks(\chi)$ is a singleton.

(The equivalent form.) Distinct members of the antichain $\weakest(\chi)$
are incomparable, so the second sentence says $|\weakest(\chi)| \ge 2$ iff
some member of $\attacks(\chi)$ has at least two elements. The first
sentence excludes $|\weakest(\chi)| = 1$, so the remaining case is
$|\weakest(\chi)| = 0$, which happens exactly when no transversal exists,
that is, when $\emptyset \in \attacks(\chi)$, and then $\attacks(\chi) =
\{\emptyset\}$ since a clutter containing $\emptyset$ has no other member.
If some member of $\attacks(\chi)$ has two or more elements, then
$\emptyset \notin \attacks(\chi)$, so $\weakest(\chi) \neq \emptyset$, and not
every member is a singleton, so $|\weakest(\chi)| \neq 1$. Hence,
$|\weakest(\chi)| \ge 2$. Conversely, if $|\weakest(\chi)| \ge 2$, then
$\emptyset \notin \attacks(\chi)$ and, by the first sentence, some member is
not a singleton; being nonempty, it has two or more elements.
\end{proof}

The informal criterion, which leaves out the word \emph{minimal}, is
false, and Table~\ref{tab:example} is the counterexample. Its
double-forgery row has two entries but is not minimal, and the weakest
trust is unique. Pruning to minimal stopping sets is required for
correctness. Without it, \eqref{eq:involution} does not apply, and the
converse direction of
Theorem~\ref{thm:unique} is false, as the following computation exhibits.

\paragraph{The blocker on the running example.} Abbreviate the atoms of
the signed protocol as $U = \{sk_A, sk_B, N\}$, for $\Non(sk_A)$,
$\Non(sk_B)$, and $\Unq(N)$. The stopping sets of Table~\ref{tab:example} form
the family
\[
\mathcal{F} = \{\{sk_A\},\ \{N\},\ \{sk_A, sk_B\}\},
\]
which is not a clutter, since $\{sk_A\} \subseteq \{sk_A, sk_B\}$. Its
minimal members give the clutter $\attacks(\chi) = \{\{sk_A\},
\{N\}\}$. A transversal must contain $sk_A$ and $N$, and $\{sk_A, N\}$ is
one. Thus,
\[
\weakest(\chi) = \blocker(\attacks(\chi)) = \{\{sk_A, N\}\},
\]
in agreement with Theorem~\ref{thm:unique}. The unminimised family has the
same transversals, so $\blocker(\mathcal{F}) = \{\{sk_A, N\}\}$ as well,
but $\blocker(\blocker(\mathcal{F})) = \{\{sk_A\}, \{N\}\} \neq
\mathcal{F}$. Hence, the involution fails for $\mathcal{F}$, and
$\mathcal{F}$ has a two-element member although its blocker is a single
set. For the
two-key variant, with $U = \{k_1, k_2, N\}$, the minimal attacks are
forgery and replay,
\[
\attacks(\chi) = \{\{k_1, k_2\},\ \{N\}\},
\qquad
\blocker(\attacks(\chi)) = \{\{k_1, N\},\ \{k_2, N\}\},
\]
since a transversal contains $N$ and one of the keys. Applying $\blocker$
again, a set meets both $\{k_1, N\}$ and $\{k_2, N\}$ exactly when it
contains $N$ or contains both keys. Thus, the minimal such sets are $\{N\}$ and $\{k_1,
k_2\}$, which is $\attacks(\chi)$ again, as \eqref{eq:involution} requires.
The member $\{k_1, k_2\}$ has two elements, and Theorem~\ref{thm:unique}
predicts the two incomparable weakest trusts found in
Theorem~\ref{thm:nofilter}.

When some attack has no remedy at all, the answer is that no trust over
$U$ suffices.

\begin{remark}[Unrealisability]
\label{rem:unrealisable}
If an attack $x$ has $\stops(x) = \emptyset$, then $\attacks(\chi) =
\{\emptyset\}$. No set is a transversal, so $\weakest(\chi)$ is the empty
antichain. This represents the empty disjunction $\mathsf{false}$, the
greatest element of $\completion(\trustlat)$ in the assumes-more order.
Thus, no trust achieves $\chi$. This negative answer distinguishes ``you
must assume more'' from ``no assumption will save this protocol''. This
occurs for the flawed \textsc{pkinit} of Section~\ref{sec:intro} when $U$
contains only origination atoms: the attack there compromises no honest
key and replays no stale value \cite{cervesato}. Unrealisability is always
relative to $U$. Enlarging $U$ with a channel assumption that removes the
run turns the answer from $\mathsf{false}$ into a trust that names the
channel.
\end{remark}

\section{What Follows from the Galois Connection}
\label{sec:consequences}

Four facts about trust follow from the Galois connection of
Section~\ref{sec:galois} and from the identity $\weakest(\chi) =
\blocker(\attacks(\chi))$ of Section~\ref{sec:blocker}, with no
protocol-specific argument. Each answers
a question a protocol designer asks: what several goals require together,
whether separately analysed protocols can be combined, whether a stated
assumption is unneeded, and when one protocol is at least as secure as
another.

\subsection{Trust Distributes over Goal Conjunction}

A protocol usually has several goals. The trust required for all of them
is the combination of the trusts required for each, so a designer can
compute the trust per goal and combine the answers.

\begin{corollary}[Conjunction]
\label{cor:conjunction}
$\weakest(\chi_1 \wedge \chi_2) = \weakest(\chi_1) \vee \weakest(\chi_2)$,
where the join is taken in $\completion(\trustlat)$ and computed as the
antichain of minimal elements of
$\{T_1 \cup T_2 : T_i \in \weakest(\chi_i)\}$.
\end{corollary}

\begin{proof}
The identity holds because left adjoints preserve joins. A direct
computation gives it again, together with the formula.

\emph{From the adjunction.} In the strength order on $\goals$ the upper
bounds of $\chi_1$ and $\chi_2$ are the goals that imply both, that is,
the goals that imply $\chi_1 \wedge \chi_2$. Hence, $\chi_1 \wedge \chi_2$ is
the join $\chi_1 \vee \chi_2$ in $\goals$, and it lies in $\goals$ because
the conclusion fragment of \cite{rowe} is closed under conjunction. By
Theorem~\ref{thm:completion}, $\weakest : \goals \to \completion(\trustlat)$
is a left adjoint, and left adjoints preserve all existing joins
(Theorem~\ref{thm:galois}). Hence, $\weakest(\chi_1 \wedge \chi_2) =
\weakest(\chi_1) \vee \weakest(\chi_2)$ in $\completion(\trustlat)$.

\emph{Directly.} A run satisfies $\chi_1 \wedge \chi_2$ iff it satisfies
both, so a trust achieves the conjunction iff it achieves each conjunct:
$\Suff(\chi_1 \wedge \chi_2) = \Suff(\chi_1) \cap \Suff(\chi_2)$. Writing
$M_i = \weakest(\chi_i)$, we have $\Suff(\chi_i) = \mathop{\uparrow} M_i$,
and
\[
\mathop{\uparrow} M_1 \cap \mathop{\uparrow} M_2
= \{\tau : \tau \supseteq T_1 \text{ and } \tau \supseteq T_2
    \text{ for some } T_i \in M_i\}
= \mathop{\uparrow} \{T_1 \cup T_2 : T_i \in M_i\}.
\]
The minimal elements of this up-set are the minimal elements of
$\{T_1 \cup T_2 : T_i \in M_i\}$, which is the stated antichain, and the
antichain generating $\mathop{\uparrow} M_1 \cap \mathop{\uparrow} M_2$ is
by definition the join $M_1 \vee M_2$ in $\completion(\trustlat)$
(Section~\ref{sec:galois}). In terms of formulas this is distribution,
$(\bigvee_i \bigwedge T_i) \wedge (\bigvee_j \bigwedge T'_j)
= \bigvee_{i,j} \bigwedge (T_i \cup T'_j)$,
followed by deleting the implicants that are not prime. (The join in
$\completion(\trustlat)$'s order corresponds to conjunction of the
denoted monotone functions, since assuming more is higher in the trust
order.)
\end{proof}

For example, suppose the signed challenge-response of
Section~\ref{sec:example} must achieve both $B$'s authentication of $A$,
with $\weakest = \{\{\Non(sk_A), \Unq(N)\}\}$, and a hypothetical second
goal whose weakest trust is $\{\{\Non(sk_A)\}, \{\Non(sk_B)\}\}$. The
formula gives the minimal elements of
\[
\{\{\Non(sk_A), \Unq(N)\},\ \{\Non(sk_A), \Non(sk_B), \Unq(N)\}\},
\]
which is $\{\{\Non(sk_A), \Unq(N)\}\}$. Thus, the second goal costs
nothing beyond the first.

The dual half is a prediction about existing tools. The right adjoint
preserves meets, so $\strongest$ of the either/or trust ``$\tau_1$ or
$\tau_2$'', the meet $\{\tau_1\} \wedge \{\tau_2\}$ in
$\completion(\trustlat)$, is the meet of $\strongest(\tau_1)$ and
$\strongest(\tau_2)$ in $\goals$, which is the disjunction of the two
guarantees,
by Theorem~\ref{thm:completion}. An enrich-by-need analysis started from two
skeletons computes this when it returns the union of the shapes found from
each.

\subsection{Composition, and When It Is a Join}

Designers build systems from protocols that were analysed separately.
They need to know when the trust each component requires is all that the
combined system requires. In general, composition is not a join of
trusts. Running two protocols in parallel can create attacks that belong
to neither, which is the subject of protocol independence
\cite{guttmanthayer,guttmanestablishing}. Theorem~\ref{thm:composition}
states exactly when the join is the answer.

\begin{theorem}[Composition]
\label{thm:composition}
Let $\Pi_1$ and $\Pi_2$ be protocols over a common $U$, with situations
$\sigma_1$ and $\sigma_2$, and let the composition $\Pi_1 \parallel \Pi_2$
have situation $\sigma = \sigma_1 \cup \sigma_2$ and goal $\chi = \chi_1
\wedge \chi_2$, with $\chi_i$ a goal of $\Pi_i$. Suppose that every attack
$x$ of $\Pi_i$ on $\chi_i$, for $i = 1, 2$, extends to an attack $x'$ of
the composition, a run of $\Pi_1 \parallel \Pi_2$ that meets $\sigma$ and
still violates $\chi$, with $\stops(x') \subseteq \stops(x)$. Then
\[
\weakest_{\Pi_1 \parallel \Pi_2}(\chi) \;\ge\;
\weakest_{\Pi_1}(\chi_1) \vee \weakest_{\Pi_2}(\chi_2),
\]
with equality if and only if every member of
$\attacks_{\Pi_1 \parallel \Pi_2}(\chi)$ contains a member of
$\attacks_{\Pi_1}(\chi_1) \cup \attacks_{\Pi_2}(\chi_2)$: every minimal
attack of the composition is stopped by no less than what some minimal
attack of a component is stopped by.
\end{theorem}

\begin{proof}
Component attacks are attacks of the composition, so the composition needs
at least the join, and Lemma~\ref{lem:refinement} turns each inclusion of
transversal families into a containment condition on attacks. Write
$\attacks_c = \attacks_{\Pi_1 \parallel \Pi_2}(\chi)$, $\attacks_i =
\attacks_{\Pi_i}(\chi_i)$, and $\mathcal{K} = \attacks_1 \cup \attacks_2$,
and write $\weakest_c$, $\weakest_i$ for the corresponding weakest trusts.

\emph{Step 1: every component stopping set contains a composition
stopping set.} Let $E \in \mathcal{K}$, say $E = \stops(x)$ for an attack
$x$ of $\Pi_i$ on $\chi_i$. By the hypothesis, $x$ extends to an attack
$x'$ of the composition with $\stops(x') \subseteq E$. The set
$\stops(x')$ contains a member of $\attacks_c$, so $E$ does too.

\emph{Step 2: the inequality.} By Step 1, every member $E$ of
$\mathcal{K}$ contains a member of $\attacks_c$. By Lemma~\ref{lem:refinement}, with
$\mathcal{A} = \attacks_c$ and $\mathcal{B} = \mathcal{K}$,
\[
\mathrm{Tr}(\attacks_c) \subseteq \mathrm{Tr}(\mathcal{K})
= \mathrm{Tr}(\attacks_1) \cap \mathrm{Tr}(\attacks_2).
\]
By Proposition~\ref{prop:hitting}, applied to each protocol, the left side
is $\Suff_{\Pi_1 \parallel \Pi_2}(\chi) = \mathop{\uparrow} \weakest_c$ and
the right side is $\Suff_{\Pi_1}(\chi_1) \cap \Suff_{\Pi_2}(\chi_2) =
\mathop{\uparrow} \weakest_1 \cap \mathop{\uparrow} \weakest_2$, which is
$\mathop{\uparrow} (\weakest_1 \vee \weakest_2)$ by the description of joins
in $\completion(\trustlat)$ (Section~\ref{sec:galois}). Since $M \le M'$ in
$\completion(\trustlat)$ iff $\mathop{\uparrow} M' \subseteq
\mathop{\uparrow} M$, this is $\weakest_1 \vee \weakest_2 \le \weakest_c$.

\emph{Step 3: the equality condition.} An antichain is determined by the
up-set it generates, so equality holds iff $\mathrm{Tr}(\attacks_c) =
\mathrm{Tr}(\mathcal{K})$. One inclusion is Step 2, so equality holds iff
$\mathrm{Tr}(\mathcal{K}) \subseteq \mathrm{Tr}(\attacks_c)$, which by
Lemma~\ref{lem:refinement}, now with $\mathcal{A} = \mathcal{K}$ and
$\mathcal{B} = \attacks_c$, holds iff every member of $\attacks_c$ contains
a member of $\mathcal{K}$. That is the displayed condition.
\end{proof}

The hypothesis is used once, in Step 1, and Step 1 is the only protocol
content. Steps 2 and 3 hold for any three clutters related as in Step 1.

When $\sigma_1 = \sigma_2 = \sigma$, the hypothesis holds with $x' = x$.
Every role of $\Pi_i$ is a role of $\Pi_1 \parallel \Pi_2$, so an attack
$x = (\mathcal{B}, \eta)$ of $\Pi_i$ is a run of the composition that
satisfies $\sigma$ and violates $\chi_i$, and hence $\chi$. Its stopping
set is the same in both protocols. When the situations differ, a run of
$\Pi_1$ alone does not satisfy $\sigma_2$, so the extension must add a run
$y$ of $\Pi_2$ that satisfies $\sigma_2$. Adding $y$ alongside $x$ is not
sufficient by itself, since it need not keep the stopping set within
$\stops(x)$ or keep $\chi$ false. The $\Unq$ assumption is not additive:
if $y$ originates a value that originates once in $x$, then $\Unq$ of that
value fails in the union although it held in $x$. Moreover, the strands of
$y$ might supply the witness that $\chi_i$ requires. One sufficient
condition has four parts. First, $y$ is a bundle of honest strands of
$\Pi_2$, disjoint from $x$ and with no edge to or from it, whose every edge
is a delivery on one channel. Second, $y$ satisfies $\sigma_2$ under an
extension $\eta'$ of $\eta$ and violates no atom of $U$ under $\eta'$.
Third, every value $y$ originates occurs nowhere in $x$. Fourth, no strand
of $y$ is an instance of a role named in $\chi_i$. Then an origination atom
of $U$ fails in $x \cup y$ only if it fails in $x$, and no edge of $y$ uses
a capability. Moreover, when every node that $\chi_i$ requires carries a
role atom, as in agreement goals, the witnesses $\chi_i$ could use are
nodes of $x$, whose order and parameters $y$ does not change. Thus,
$\chi_i$ still fails.

In the pairs of Table~\ref{tab:composition}, the situation of the
composition is the union of the components' situations, each naming a
completed run of its own protocol, so the hypothesis requires such a
$y$. For an attack on $\Pi_1$, we take for $y$ an honest run of the
attestation service described below, using its own fresh value $M$,
chosen to occur nowhere in the attack. It originates only $M$, which no atom of $U$ names,
emits no key, and has no strand of a role of $\Pi_1$. For an attack on
$\Pi_2$, we take for $y$ an honest run of the signed challenge-response on a
challenge chosen to occur nowhere in the attack, with $\eta'$ sending $N$
to it. It originates only that challenge, once, so it satisfies $\Unq(N)$
under $\eta'$, emits no key, and has no strand of a role of $\Pi_2$. Without
some such run, the inequality is not guaranteed, since the extra conjuncts
of $\sigma$ could exclude a component's attacks from the composition.

\begin{remark}[Disjoint encryption, as a conjecture]
\label{rem:disjoint}
We conjecture that disjoint encryption in the sense of
\cite{guttmanthayer} is a sufficient condition for the equality, and we
record why without claiming it as an instance. Their independence result transforms every run of
the composition into behaviour of the components, preserving the regular
strands and granting the adversary no capability it lacked. An attack on
$\chi_i$ would then map to a component attack whose stopping set is
contained in the original's, as the equality condition requires.
Verifying this requires tracking stopping sets through the transformation
of \cite{guttmanthayer}, which we have not yet done. We therefore state it
as a conjecture.
\end{remark}

\begin{table}[t]
\caption{Composition (Theorem~\ref{thm:composition}) of the signed challenge-response
($\chi_1$) with an attestation service ($\chi_2$). Last column: minimal attacks of the
composition that contain no component's minimal attack.}
\label{tab:composition}
\centering
\setlength{\tabcolsep}{3pt}
\footnotesize
\begin{tabular}{p{2.0cm} p{3.1cm} p{3.1cm} p{2.6cm}}
\toprule
pair & join of components & composition & cross-protocol\\
\midrule
tagged, shared $sk_A$ & $\{\Non(sk_A), \Unq(N)\}$ & $\{\Non(sk_A), \Unq(N)\}$ & none\\
untagged, shared $sk_A$ & $\{\Non(sk_A), \Unq(N)\}$ & $\mathsf{false}$ & $\{\Non(sk_C)\}$ on $\chi_1$, $\emptyset$ on $\chi_2$\\
separate keys & $\{\Non(sk_A), \Non(sk'_A),$ $\Unq(N)\}$ & $\{\Non(sk_A), \Non(sk'_A),$ $\Unq(N)\}$ & none\\
\bottomrule
\end{tabular}
\end{table}

Table~\ref{tab:composition} shows three pairs, computed by
Algorithm~\ref{alg:loop} over the analyser of Section~\ref{sec:eval}, each
confirmed by evaluating every trust. Each pair composes the signed
challenge-response of Section~\ref{sec:example} ($\Pi_1$, with $\chi_1$
recent agreement) with an attestation service ($\Pi_2$). In response to a
request signed with $sk_C$, $A$ signs $\langle M, A, B\rangle$ for $B$, and
$\chi_2$ is $B$'s agreement with $A$ on $M$. The situation is the union of
the two components' situations, with each goal tied to its own completed
$B$ run. Alone, the components require $\{\Non(sk_A), \Unq(N)\}$ and
$\{\Non(sk_A)\}$. When the two signed formats carry distinct tags, the
composition's minimal attacks are forgery and replay, both component
attacks. Hence, the condition of Theorem~\ref{thm:composition} holds
although the key is shared, and the weakest trust is the join. Without
tags, the analyser finds two multi-protocol attacks \cite{cremersares}, in
which messages of one protocol are accepted as messages of the other. In
the first, the adversary forges a request on $B$'s challenge and returns
$A$'s attestation as the response. Only $\Non(sk_C)$ removes this run, and
neither component needs this assumption. In the second, the adversary
submits its own value to $A$ as a challenge and delivers the signature to
$B$ as an attestation. No atom of $U$ removes this run. The composition's
weakest trust is $\mathsf{false}$, strictly above the join, and the gap is
these two runs. With a separate key $sk'_A$ for the attestation service,
the join is exact again, since no message of one protocol is accepted by
the other.

Trust is compositional exactly to the extent that attacks are. A designer
may compute the trust required by each component separately and join the
answers, and the result is correct provided the components do not
interfere. Where they do interfere, and component attacks still carry
over to the composition, the join is a lower bound on the trust required,
and the gap is the cross-protocol attacks.

\subsection{Tight Trusts and Attainable Goals}

A Galois connection restricts to an isomorphism between the fixed points
of its two composites, the elements that the round trip through the other
side leaves unchanged. For trusts and goals, the fixed points are the
tight trusts and the attainable goals below. The theorem holds as stated, but
the goal language limits what it tells a designer, as the paragraph after
the proof explains.

\begin{definition}[Tight and attainable]
\label{def:tight}
A trust $\tau \in \completion(\trustlat)$ is \emph{tight} if $\tau =
\weakest(\strongest(\tau))$: every assumption in it is required by the
strongest guarantee it yields. A goal $\chi$ is \emph{attainable} if
$\chi = \strongest(\weakest(\chi))$: it is the strongest goal achieved by
the weakest trust achieving it.
\end{definition}

\begin{theorem}[Correspondence]
\label{thm:correspondence}
$\weakest$ and $\strongest$ restrict to mutually inverse order
isomorphisms between the attainable goals and the tight trusts, both of
which are complete lattices. Every trust $\tau$ has a tight core
$\weakest(\strongest(\tau)) \le \tau$ achieving the same goals, and every
goal $\chi$ has an attainable strengthening $\strongest(\weakest(\chi))
\ge \chi$ requiring the same trust.
\end{theorem}

\begin{proof}
Every claim follows from the unit, the counit and the identities
\eqref{eq:triangle}, which say that a second round trip changes nothing.
Throughout, $\weakest : \goals \to \completion(\trustlat)$ and $\strongest :
\completion(\trustlat) \to \goals$ are the adjoints of
Theorem~\ref{thm:completion}, so the unit \eqref{eq:unit}, the counit
\eqref{eq:counit} and the identities \eqref{eq:triangle} hold.

\emph{The maps restrict.} If $\tau$ is tight then $\strongest(\tau)$ is
attainable: $\strongest(\weakest(\strongest(\tau))) = \strongest(\tau)$ by
\eqref{eq:triangle}. If $\chi$ is attainable then $\weakest(\chi)$ is tight:
$\weakest(\strongest(\weakest(\chi))) = \weakest(\chi)$, again by
\eqref{eq:triangle}. (In fact \eqref{eq:triangle} says every value of
$\strongest$ is attainable and every value of $\weakest$ is tight, so the
attainable goals are exactly the image of $\strongest$ and the tight trusts
exactly the image of $\weakest$.)

\emph{They are mutually inverse order isomorphisms.} For tight $\tau$,
$\weakest(\strongest(\tau)) = \tau$ by definition, and for attainable
$\chi$, $\strongest(\weakest(\chi)) = \chi$ by definition. Hence, the two
restrictions are inverse bijections. Both are monotone
(Theorem~\ref{thm:galois}), and a monotone bijection whose inverse is
monotone is an order isomorphism.

\emph{The tight trusts form a complete lattice.} Write $\kappa = \weakest
\circ \strongest$, a kernel operator on $\completion(\trustlat)$
(Theorem~\ref{thm:galois}), whose fixed points are the tight trusts. Let
$\{\tau_j\}_{j \in J}$ be any family of fixed points, possibly empty, and
$\tau = \bigvee_j \tau_j$ its join in $\completion(\trustlat)$, which
exists because $\completion(\trustlat)$ is a finite lattice. For each $j$,
$\tau_j = \kappa(\tau_j) \le \kappa(\tau)$ by monotonicity, so $\tau \le
\kappa(\tau)$. Also $\kappa(\tau) \le \tau$ because $\kappa$ is deflationary.
Therefore, $\tau$ is a fixed point, and the fixed points are closed under all
joins of $\completion(\trustlat)$, including the empty join, its least
element. A
subset of a complete lattice closed under all joins is a complete lattice
in the induced order, with the same joins and with the meet of a family
given by the join of its lower bounds within the subset
\cite{birkhoff}. The attainable goals, being order isomorphic to the tight
trusts, form a complete lattice too, whose joins and meets are the images
under $\strongest$ of those of the tight trusts. The attainable goals are
the fixed points of the closure operator $\strongest \circ \weakest$, so
they are closed under every meet that exists in $\goals$: if each
$\chi_j$ is attainable and $\chi = \bigwedge_j \chi_j$ exists in
$\goals$, then $\strongest(\weakest(\chi)) \le
\strongest(\weakest(\chi_j)) = \chi_j$ for every $j$, so
$\strongest(\weakest(\chi)) \le \chi$, and $\chi \le
\strongest(\weakest(\chi))$ by \eqref{eq:unit}. In particular, the meet of
finitely many attainable goals is their disjunction. Joins, however, may
differ. The join of attainable goals among the attainable goals is
$\strongest \circ \weakest$ of their conjunction, which may be strictly
stronger than the conjunction.

\emph{The tight core.} By \eqref{eq:counit}, $\weakest(\strongest(\tau))
\le \tau$, and it is tight since $\kappa$ is idempotent. It achieves the
same goals as $\tau$: for every $\chi$, $\chi \le
\strongest(\weakest(\strongest(\tau)))$ iff $\chi \le \strongest(\tau)$, by
\eqref{eq:triangle}, and by \eqref{eq:adjunction} these say
$\weakest(\strongest(\tau)) \ach \chi$ and $\tau \ach \chi$.

\emph{The attainable strengthening.} By \eqref{eq:unit},
$\strongest(\weakest(\chi)) \ge \chi$, and it is attainable since
$\strongest \circ \weakest$ is idempotent. It requires the same trust:
$\weakest(\strongest(\weakest(\chi))) = \weakest(\chi)$ by
\eqref{eq:triangle}, so by \eqref{eq:adjunction} a trust achieves
$\strongest(\weakest(\chi))$ iff it achieves $\chi$.
\end{proof}

The correspondence has a limitation. Conclusions of the goal language may
contain origination atoms, and $\strongest(\tau)$ does. It is the shape
analysis sentence of Theorem~\ref{thm:sas}, and the shapes carry the
origination assumptions of $\tau$, so $\delta$ of each shape lists them. A
trust that omits an origination atom $a$ of $\tau$ admits, as a rule, a
run in which $a$ fails, and that run violates every disjunct of
$\strongest(\tau)$. As a result, $\weakest(\strongest(\tau))$ retains the
origination atoms of $\tau$. The tight core of a trust of origination
atoms is in general the trust itself, making nearly every such trust
tight. Thus, the tight core does not identify the assumptions a stated
guarantee can do without. Computing $\weakest(\chi)$ for the stated goal
$\chi$ does, as in the rest of the paper. A tight core that discards
unneeded assumptions would need conclusions free of trust atoms, for instance
$\strongest(\tau)$ with its origination atoms removed. Whether the
correspondence survives that change is open.
The attainable strengthening $\strongest(\weakest(\chi))$ is affected in
the same way, since it also contains the origination atoms of
$\weakest(\chi)$.

\subsection{Attack Containment and the Tenet}

A designer comparing two protocols, or a protocol and its repair, wants a
criterion stated in terms of attacks. Rowe, Guttman, and Liskov open with
a tenet: a system is at least as secure as another exactly when every
attack that succeeds against the first also succeeds against the second.
They present it as a starting point, define an attack as a failing goal,
and then define protocol comparison so as to respect the tenet.
Theorem~\ref{thm:tenet} restates the tenet by comparing attacks through
their stopping sets. It is a restatement of the tenet in our terms rather
than a derivation of it as Rowe et al.\ state it.

\begin{theorem}[Attack containment]
\label{thm:tenet}
Let $\Pi_1$ and $\Pi_2$ share a situation $\sigma$, a goal language (so
that every goal $\chi$ is a goal of both, and $\sigma$ and $\chi$ name
roles of both), and the candidate assumptions $U$. Write $\Pi_1 \preceq
\Pi_2$ when every trust sufficient for $\Pi_1$ is
sufficient for $\Pi_2$, that is, $\Suff_{\Pi_1}(\chi) \subseteq
\Suff_{\Pi_2}(\chi)$ for every goal $\chi$. Then $\Pi_1 \preceq \Pi_2$ if
and only if, for every $\chi$ and every $E \in \attacks_{\Pi_2}(\chi)$,
there is $E' \in \attacks_{\Pi_1}(\chi)$ with $E' \subseteq E$.
\end{theorem}

\begin{proof}
Sufficient trusts are the transversals of the minimal attacks, and
Lemma~\ref{lem:refinement} compares transversal families. Fix $\chi$ and
write $\attacks_i = \attacks_{\Pi_i}(\chi)$. By
Proposition~\ref{prop:hitting}, applied to each protocol,
$\Suff_{\Pi_i}(\chi) = \mathrm{Tr}(\attacks_i)$. By
Lemma~\ref{lem:refinement} with $\mathcal{A} = \attacks_1$ and
$\mathcal{B} = \attacks_2$, $\mathrm{Tr}(\attacks_1) \subseteq
\mathrm{Tr}(\attacks_2)$ iff every member of $\attacks_2$ contains a member
of $\attacks_1$. Quantifying over $\chi$ gives the theorem. Only the
refinement lemma is needed, not the involution. The involution adds that
the condition, for clutters, is a partial order on clutters, since mutual
refinement of two clutters forces them to be equal (if $E \in \attacks_2$
contains $E' \in \attacks_1$, which contains $E'' \in \attacks_2$, then $E
= E''$ because $\attacks_2$ is a clutter, so $E = E' \in \attacks_1$).
\end{proof}

The right-hand side reads as follows. Every minimal attack on $\Pi_2$
contains, among its stopping assumptions, the whole stopping set of some
minimal attack on $\Pi_1$. Thus, an attack on $\Pi_2$ is, up to what stops
it, an attack on $\Pi_1$. In $\Pi_1 \preceq \Pi_2$, the more secure
protocol is on the right, whereas the tenet names the more secure system
first. Hence, $\Pi_2$ plays the tenet's first system and $\Pi_1$ its
second. With that reading, and with attacks identified by their stopping
sets, the tenet holds as Theorem~\ref{thm:tenet} for protocols that share
$\sigma$, the goal language, and $U$.

\section{Computing the Weakest Trust}
\label{sec:computing}

Section~\ref{sec:blocker} described the weakest trust in terms of the
minimal attacks, which are not known in advance. A loop that uses a
protocol verifier as its only oracle finds both. It keeps the stopping sets
of the attacks found so far and, as candidates, the minimal ways to stop
all of them. At the start, no attack is known, and the only candidate is
the empty trust. The loop asks the verifier whether a candidate achieves
the goal. If the verifier finds an attack, the loop records the attack's
stopping set, recomputes the candidates, and asks again. It stops when
every candidate passes. Theorem~\ref{thm:correctness} shows that the
candidates are then $\weakest(\chi)$ and the recorded stopping sets are
the minimal attacks. Two steps require care. Attacks must be found one
candidate at a time, because a single run of the verifier from the empty
trust misses some. Moreover, a verifier reports an attack but not its
stopping set, which requires further queries. Sections~\ref{sec:obvious}
and~\ref{sec:backend} treat these two steps. None of the parts is new, as
Section~\ref{sec:loop} records.

\subsection{The One-Shot Method, and Why It Fails}
\label{sec:obvious}

The one-shot method runs the verifier once. An attack is a run that violates
the goal, and an enrich-by-need analysis returns the minimal runs
compatible with a situation, its shapes. Thus, one could run the analysis
once with nothing assumed, keep the shapes that violate the goal, and read
off $\attacks(\chi)$.

The running example shows that this fails. Start the analysis from $B$'s
completed challenger strand with nothing assumed uncompromised or fresh.
The reception of $\sig{sk_A}{N,A,B}$ is explained at once. With $sk_A$
compromised, the adversary builds the signature itself, and no regular
strand has to be added. The starting skeleton is already realised, the
analysis ends with a single shape, and that shape's stopping set is
$\{\Non(sk_A)\}$. The replay attack never appears. Replay obtains the
signature from an earlier genuine session, and the analysis adds that
session only after forgery is barred, that is, after $\Non(sk_A)$ is
assumed. Thus, the one-shot method returns $\{\{\Non(sk_A)\}\}$ and drops
$\Unq(N)$ without warning.

The reason is general. Enrich-by-need adds regular behaviour only to
explain receptions the adversary cannot produce, and assumptions are what
make a reception unexplained. With nothing assumed, it has the least to
explain, and it finds the fewest runs. Attacks must therefore be found relative
to the trust being tested, which is how enrich-by-need is designed to be
run. CPSA fails in the same way on the flawed \textsc{pkinit} of
Section~\ref{sec:trace}. With nothing assumed, the only failing shape is
ruled out by $\Non(sk(S))$ and by no other atom. A one-shot reading would
report $\{\Non(sk(S))\}$ as the weakest trust, whereas the answer is
$\mathsf{false}$.

\subsection{A Counterexample-Guided Loop}
\label{sec:loop}

Algorithm~\ref{alg:loop} states the loop. It is Reiter's hitting-set tree
\cite{reiter,greiner}, which computes the minimal hitting sets of a family
of sets (the family's minimal transversals, in the terms of
Section~\ref{sec:blocker}), with a protocol verifier in place of the
theorem prover that supplies the sets. Equivalently, it is
counterexample-guided inductive synthesis \cite{cegis} with achievement as
the specification.

\begin{algorithm}[t]
\caption{The weakest trust by counterexample-guided search.}
\label{alg:loop}
\begin{algorithmic}[1]
\Statex \textbf{Input:} situation $\sigma$, goal $\chi$, candidate assumptions $U$,
  and an oracle $O$ that decides $T \ach \chi$ and returns a witness run on
  failure.
\Statex \textbf{Output:} $\weakest(\chi)$, a clutter over $U$.
\State $H \gets \emptyset$
\Loop
  \State $C \gets \blocker(H)$
  \State \textbf{if} $O(T) = \textit{achieves}$ for every $T \in C$
    \textbf{then return} $C$
  \State pick $T \in C$ with $O(T) = \textit{fails}$, with witness run $x$
  \State $H \gets \min_{\subseteq}(H \cup \{\stops(x)\})$
\EndLoop
\end{algorithmic}
\end{algorithm}

$H$ holds the stopping sets found so far, with non-minimal ones removed, so
$H$ is a clutter. $C$ holds the candidates, the minimal sets that meet
every member of $H$. At the first iteration, $\blocker(\emptyset) =
\{\emptyset\}$, so the first candidate is the empty trust. If $H$ ever
contains the empty set, no set meets every member, and $C$ is the empty
clutter. The test of line~4 then holds vacuously, and the loop returns the
empty clutter, which is $\mathsf{false}$.
Lines~4 and~5 are implemented together: the candidates are queried in a
fixed order until one fails, and that one is picked. With a memo table on
$O$, each distinct trust costs one oracle call (Listing~\ref{lst:loop}).

Any tool that decides achievement on the fragment and returns a
counterexample can serve as the oracle. An enrich-by-need analysis
started from $\sigma \wedge T$ does so whenever it terminates, which the
termination assumption of Section~\ref{sec:galois} supplies. The CPSA
driver below reports an error, not a verdict, when a search is cut off.
Line~6 requires more than a counterexample. It requires the
counterexample's stopping set, which Section~\ref{sec:backend} computes
from a second query to the backend. The correctness proof uses only one
property of the set that line~6 adds. The set must be a \emph{conflict} for the failing trust $T$: a set of
atoms that contains no atom of $T$ and that every sufficient trust meets.

\begin{definition}[Conflict]
\label{def:conflict}
A set $K \subseteq U$ is a \emph{conflict} for a trust $T$ if
$K \cap T = \emptyset$ and every sufficient trust meets $K$.
\end{definition}

\begin{lemma}[Stopping sets are conflicts]
\label{lem:conflict}
If $x = (\mathcal{B}, \eta)$ is a run with $\mathcal{B} \in
\mathcal{B}(T_C)$ that meets $\sigma \wedge T$ and violates $\chi$, then
$\stops(x)$ is a conflict for $T$.
\end{lemma}

\begin{proof}
The run $x$ is an attack that $T$ does not rule out, and by separability
any trust that rules it out must meet $\stops(x)$.
In detail, $x$ satisfies $\sigma$ and violates $\chi$, so it is an attack
in the sense of Definition~\ref{def:attacks}. It satisfies $T$, so $T$ does
not rule it out, and by Lemma~\ref{lem:separable}, $T \cap \stops(x) =
\emptyset$.
If a trust $\tau$ misses $\stops(x)$, then by Lemma~\ref{lem:separable}
$\tau$ does not rule out $x$, so $x$ is a run of the model restricted by
$\tau$ that meets $\sigma \wedge \tau$ and violates $\chi$, and $\tau$ is
not sufficient. Hence every sufficient trust meets $\stops(x)$.
\end{proof}

Recall from Section~\ref{sec:blocker} that $\mathrm{Tr}(H)$ is the set of
transversals of $H$, the sets that meet every member, and that
$\mathrm{Tr}(H)$ is the up-set generated by $\blocker(H)$. Moreover,
$\mathrm{Tr}(\min_{\subseteq}(H \cup \{K\})) = \mathrm{Tr}(H) \cap
\mathrm{Tr}(\{K\})$, since a set meets every member of $H \cup \{K\}$
exactly when it meets every minimal one.

\begin{theorem}[Correctness]
\label{thm:correctness}
If $O$ is sound and complete for achievement on the fragment, then
Algorithm~\ref{alg:loop} terminates after at most $2^{|U|}$ failing
oracle calls, returns $\weakest(\chi)$, and ends with $H =
\attacks(\chi)$. The same holds if line~6 adds, in place of $\stops(x)$,
any conflict for $T$.
\end{theorem}

\begin{proof}
Each failure removes the failing candidate from $\mathrm{Tr}(H)$, which can
happen only finitely often, and at return $\mathrm{Tr}(H)$ is exactly the
set of sufficient trusts, so its minimal elements are both $\blocker(H)$
and $\weakest(\chi)$.

By Lemma~\ref{lem:conflict}, the set $\stops(x)$ added at line~6 is a
conflict for $T$, since soundness and completeness of $O$ make the witness
$x$ a run that satisfies $\sigma \wedge T$ and violates $\chi$. It therefore
suffices to prove the claims when every set added at line~6 is a conflict
$K$ for the trust $T$ picked at line~5. Two invariants hold throughout.
First, $H$ is a clutter, because line~6 removes non-minimal sets. Second,
every member of $H$ is a conflict for some trust, so every sufficient trust meets every member
of $H$, that is, $\Suff(\chi) \subseteq \mathrm{Tr}(H)$.

\emph{Termination.} At line~6, $T$ is a member of $\blocker(H)$ and hence
a transversal of $H$, while $T \cap K = \emptyset$. The new clutter $H' =
\min_{\subseteq}(H \cup \{K\})$ has $\mathrm{Tr}(H') = \mathrm{Tr}(H) \cap
\mathrm{Tr}(\{K\})$, which is contained in $\mathrm{Tr}(H)$ and does not
contain $T$. Hence, $\mathrm{Tr}(H)$ strictly decreases at every failure.
It is a family of subsets of the finite set $U$, so there are at most
$2^{|U|}$ failures, and each iteration without a failure returns.

\emph{Correctness.} Consider the state at return, with final clutter $H$
and answer $C = \blocker(H)$. Every member of $C$ achieves $\chi$, since
line~4 checked it and $O$ is sound. Hence, $C \subseteq \Suff(\chi)$. Since
$\Suff(\chi)$ is an up-set (Lemma~\ref{lem:monotone}), it contains the
up-set generated by $C$, which is $\mathrm{Tr}(H)$. With the invariant,
$\Suff(\chi) = \mathrm{Tr}(H)$. The minimal elements of $\Suff(\chi)$ are
$\weakest(\chi)$ (Theorem~\ref{thm:completion}), and those of
$\mathrm{Tr}(H)$ are $\blocker(H) = C$, so $C = \weakest(\chi)$. Finally,
$H$ is a clutter with $\blocker(H) = \weakest(\chi) =
\blocker(\attacks(\chi))$ (Proposition~\ref{prop:hitting}), and applying
$\blocker$ to both sides and using the involution \eqref{eq:involution}
gives $H = \attacks(\chi)$. Thus, the loop finds the minimal attacks
without listing them in advance.
\end{proof}

The proof uses neither that the sets added are minimal nor that they are
stopping sets of single runs. It uses only that they are conflicts. Real
backends require both relaxations. The first affects the cost
bound of Theorem~\ref{thm:cost}, and the second is required for the CPSA
backend of Section~\ref{sec:backend}.

\begin{remark}[A hand analysis that follows the loop]
\label{rem:byhand}
Rowe et al.\ \cite[Section~6.1]{rowe} analyse \textsc{pkinit} as follows.
They start from $\Non(sk(S))$, rather than from the empty trust, and
observe that the skeleton is realised, as the adversary may recover the
session key. They restart with $\Non(sk(C))$ and $\Unq(k)$ added and obtain
a new unexplained reception. Finally, they add $\Unq(n_2)$ to force two
client strands to coincide. Their analysis follows the pattern of
Algorithm~\ref{alg:loop}: query a trust, inspect the counterexample, and
add assumptions. It differs in that each assumption is picked by
inspection, whereas the loop records a stopping set and recomputes the
candidates. They write that they know of no tool designed for this kind of
reasoning \cite[Section~6.4]{rowe}. Algorithm~\ref{alg:loop} carries it out
with a verifier as the oracle.
\end{remark}

\subsection{Computing \texorpdfstring{$\stops(x)$}{stops(x)} from a Backend}
\label{sec:backend}

A verdict and a counterexample are not enough for line~6 of
Algorithm~\ref{alg:loop}. A backend must answer two kinds of queries.
\begin{description}
\item[Verify$(T)$.] Decide whether $T \ach \chi$. On failure, return one
  or more counterexamples $x$, runs $(\mathcal{B}, \eta)$ or symbolic
  descriptions of runs that meet $\sigma \wedge T$ and violate $\chi$. The
  assignment is part of the counterexample, since the atoms of $U$ name
  values through it.
\item[Survives$(x, T')$.] For a counterexample $x$ returned by
  Verify$(T)$ and a trust $T' \supseteq T$, decide whether $x$ remains a
  counterexample under the additional assumptions in $T'$.
\end{description}
Verify$(T)$ is the oracle $O$. Survives is needed because a counterexample
does not, in general, show its stopping set. A backend reports the regular
strands and the bindings of a run. Whether
an assumption rules the run out may depend on what the adversary had to do
to produce a reception, which the backend need not print. No other access
to the backend is needed, neither the adversary's derivations nor the way
the backend represents assumptions.

\begin{algorithm}[t]
\caption{The clause of a counterexample $x$ found at a failing trust $T$.}
\label{alg:stops}
\begin{algorithmic}[1]
\State $M \gets \emptyset$
\For{$a \in U \setminus T$, in a fixed order}
  \If{Survives$(x, T \cup M \cup \{a\})$}
    \State $M \gets M \cup \{a\}$
  \EndIf
\EndFor
\State \Return $U \setminus (T \cup M)$
\end{algorithmic}
\end{algorithm}

Algorithm~\ref{alg:stops} computes what we call the \emph{clause} of $x$,
the set of atoms that line~6 adds. It accumulates assumptions under which
the attack survives. Starting with $M = \emptyset$, it adds each atom $a$
to $M$ whenever $x$ survives $T \cup M \cup \{a\}$, and finally returns $U
\setminus (T \cup M)$. This requires $|U \setminus T|$ Survives queries per
counterexample. When Verify returns several counterexamples, each yields a
clause, and each clause may be added to $H$. Our implementation adds the
smallest.

\begin{proposition}[The clause is a conflict]
\label{prop:clause}
Suppose that Survives$(x, T')$ answers yes only if $T'$ is not
sufficient. Then the clause $K$ returned by Algorithm~\ref{alg:stops} is a
conflict for $T$. If, moreover, $x$ is a single run $(\mathcal{B}, \eta)$
and Survives$(x, T')$
answers yes exactly when $T'$ does not rule out $x$, then $K = \stops(x)$,
whatever the order of the loop.
\end{proposition}

\begin{proof}
The trust $T \cup M$ stays insufficient throughout, and every trust that
misses the returned clause is contained in it.

$K \cap T = \emptyset$ by construction. For the second condition of
Definition~\ref{def:conflict}, we show that the trust $T \cup M$ is not
sufficient at every point of the loop. Initially $M = \emptyset$, and $x$ witnesses that
$T$ fails. Afterwards $M$ grows only to a set $T \cup M \cup \{a\}$ for
which Survives answered yes, which by hypothesis is not sufficient. At
return, a trust $\tau$ with $\tau \cap K = \emptyset$ satisfies $\tau
\subseteq T \cup M$. Since $\Suff(\chi)$ is an up-set
(Lemma~\ref{lem:monotone}) and $T \cup M \notin \Suff(\chi)$, $\tau \notin
\Suff(\chi)$. Therefore, every sufficient trust meets $K$.

For the second claim, Lemma~\ref{lem:separable} says that $T'$ rules out
$x$ iff $T' \cap \stops(x) \neq \emptyset$. Since $x$ meets $T$, $T \cap
\stops(x) = \emptyset$, and inductively $(T \cup M) \cap \stops(x) =
\emptyset$. The test at line~3 therefore succeeds iff $a \notin \stops(x)$. Hence,
$M = (U \setminus T) \setminus \stops(x)$ and $K = (U \setminus T) \cap
\stops(x) = \stops(x)$, independently of the order.
\end{proof}

Theorem~\ref{thm:correctness} requires only conflicts. Thus,
Algorithm~\ref{alg:loop} with line~6 computed by Algorithm~\ref{alg:stops}
returns $\weakest(\chi)$ under the first hypothesis of
Proposition~\ref{prop:clause} alone. That hypothesis asks only that every
positive Survives result be correct. A wrong negative result keeps an atom
in the clause and makes the clause larger, which costs oracle calls but
not correctness.

\paragraph{Families of runs.} A counterexample need not be a single run.
A CPSA shape is a skeleton, and a realised skeleton represents every
bundle that realises it. Survives$(x, T')$ then reports whether some member
of the family is a counterexample under $T'$, and different members may
survive different atoms. Survival of a family need not decompose atom by
atom, so testing each atom independently, taking the clause $\{a \in U \setminus
T : \text{not Survives}(x, T \cup \{a\})\}$, is unsound. Suppose $U =
\{a, b\}$ and the only attacks are two runs $x_1$, $x_2$ with $\stops(x_1)
= \{b\}$ and $\stops(x_2) = \{a\}$, reported together as one symbolic
counterexample $x$ at $T = \emptyset$. Then $x$ survives $\{a\}$ through
$x_1$ and $\{b\}$ through $x_2$, so the per-atom clause is empty. The loop
would then return $\mathsf{false}$, although $\{a, b\}$ is sufficient.
Algorithm~\ref{alg:stops} keeps $a$ in $M$, then finds that $x$ does not
survive $\{a, b\}$, and returns $\{b\}$, which is $\stops(x_1)$. Adding
atoms to $M$ one at a time selects a member of the family that survives
everything in $M$, and the returned clause is one that this member
satisfies. The loop then queries $\{b\}$, obtains $\{a\}$, and returns
$\{\{a, b\}\}$. When $x$ is a single run, the two procedures agree, by the second half of
Proposition~\ref{prop:clause}.

\paragraph{The CPSA instantiation.} For CPSA~4 \cite{cpsaspec}, the two
queries are runs of the tool.
\begin{itemize}
\item \emph{Verify$(T)$.} Each atom of $U$ is a goal fact, \texttt{(non
  $t$)} or \texttt{(uniq $t$)}, over the goal's universally quantified
  variables. The facts of $T$ are added to the hypothesis of the goal, and
  CPSA is run once. A run counts only if the search ends with
  \texttt{Nothing left to do} and reports no step-limit or strand-bound
  abort. Otherwise, the driver stops with an error rather than a verdict. Every
  shape annotated \texttt{(satisfies (no \dots))} is a counterexample, and
  the annotation lists the bindings of the goal's variables to terms of
  the shape.
\item \emph{Survives$(x, T')$.} The shape $x$ is printed back as an input
  skeleton, keeping its strands, its ordering, and its origination
  declarations. The atoms of $T$ are already among these, since CPSA
  carries the hypothesis facts into every skeleton it derives. Each atom
  of $T' \setminus T$ is instantiated by the bindings of the verdict and
  appended to the skeleton's \texttt{non-orig} or \texttt{uniq-orig} list,
  and CPSA is run on the skeleton alone. The shape survives exactly when
  the first skeleton CPSA prints, the input itself, is marked realised and is
  marked neither as a preskeleton nor as dead. Output on the error stream,
  which signals an ill-formed input, counts as not surviving.
\end{itemize}
A realised input skeleton satisfies the hypothesis of
Proposition~\ref{prop:clause}. Being realised, it has a realisation that
adds only adversary strands, so the regular strands and the bindings are
those of $x$. The realisation satisfies $\sigma \wedge T'$, since the
added declarations hold in it, and it violates $\chi$, since the conclusion
fails on $x$'s regular strands with these bindings. Hence, $T'$ is not
sufficient.

The preskeleton and dead marks matter for the following reason. When an
added \texttt{uniq} atom names a value that the shape originates at two
points, the input is not a skeleton, and $\Unq(v)$ rules out $x$ by
Definition~\ref{def:attacks}. CPSA does not reject such an input. Instead,
it marks the first skeleton realised and a preskeleton and then merges the
two strands into a different skeleton, or it marks the input realised and
dead. In both cases, the realised mark does not describe a run with $x$'s
strands. Reading it as survival removes from the clause an atom that rules
out $x$, so the clause need not be a conflict.
Section~\ref{sec:eval-impl} gives a case where this changed the answer. Reading the
two marks as not surviving is safe by the remark after
Proposition~\ref{prop:clause}, since a negative result only keeps the atom
in the clause.

This case occurs on the evaluated rows. On the fixed Blanchet exchange,
the first failing shape is replayed with $\Unq(d)$, and CPSA answers with
a preskeleton. Thus, $\Unq(d)$ stays in the clause, and exhaustive
evaluation lists $\{\Unq(d)\}$ among the minimal attacks.
Listing~\ref{lst:oracle} is the driver, abridged.

Attack traces from ProVerif or Tamarin can be checked in the same way, by
asking whether the trace still satisfies the strengthened hypothesis. A
bounded model checker answers Survives by replaying the run under the
added restrictions.

\subsection{A Worked Trace: The Flawed \textnormal{\textsc{pkinit}}}
\label{sec:trace}

The flawed \textsc{pkinit} row of Table~\ref{tab:results} has $U =
\{\Non(sk(S)), \Non(sk(C))\}$, written in the artifact as \texttt{non\_as}
for \texttt{(non (privk as))} and \texttt{non\_c} for \texttt{(non (privk
c))}, and has the client's authentication of the server as its goal.
Listing~\ref{lst:pkinit} is the part of the artifact's CPSA log that
records the loop on this row, with witnesses as CPSA returns them (file
\texttt{cpsa\_log.txt} in \texttt{results/final}); Table~\ref{tab:pkinit}
reads it.

\begin{lstlisting}[caption={The run of Algorithm~\ref{alg:loop} on the flawed \textsc{pkinit}, verbatim from the CPSA log.},label={lst:pkinit}]
    O({}) fails: shape 0: stops = {non_as}  [non_as: unrealized; non_c: realized]
    O({non_as}) fails: shape 1: stops = {}  [non_c: realized]
  weakest[as-found]: W = [], H = {}, distinct 2 (2g+3r), rejected 0
\end{lstlisting}

\begin{table}[t]
\caption{The trace of Listing~\ref{lst:pkinit}. Each row is one oracle
call: one CPSA goal run at the candidate $T$, followed by one replay per atom of
$U \setminus T$ in the order of $U$ (Algorithm~\ref{alg:stops}), each with
$M \cup \{a\}$ added to the shape.}
\label{tab:pkinit}
\centering
\small
\begin{tabular}{llllll}
\toprule
call & candidate $T$ & failing shape & replays & clause & $H$ after\\
\midrule
1 & $\emptyset$ & shape 0 & $\Non(sk(S))$: unrealised & $\{\Non(sk(S))\}$ & $\{\{\Non(sk(S))\}\}$\\
  & & & $\Non(sk(C))$: realised & & \\
2 & $\{\Non(sk(S))\}$ & shape 1 & $\Non(sk(C))$: realised & $\emptyset$ & $\{\emptyset\}$\\
\bottomrule
\end{tabular}
\end{table}

At the first iteration, $H = \emptyset$ and the only candidate is the
empty trust. CPSA reports one failing shape. Replaying it with $\Non(sk(S))$
added leaves it unrealised, so the atom is not kept. Replaying it with
$\Non(sk(C))$ added (with $M$ still empty) leaves it realised, so $M =
\{\Non(sk(C))\}$ and the clause is $\{\Non(sk(S))\}$. The adversary in
this shape therefore needs the server's private key but not the client's. At the
second iteration, $\blocker(\{\{\Non(sk(S))\}\}) = \{\{\Non(sk(S))\}\}$,
and CPSA reports one failing shape under $\Non(sk(S))$. Its one replay,
with $\Non(sk(C))$ added, is realised, so the clause is empty. No atom of
$U$ rules the shape out, which is the situation
Remark~\ref{rem:unrealisable} describes for the attack of
\cite{cervesato}. Now $H = \{\emptyset\}$,
whose blocker is the empty clutter. Line~4 holds vacuously, and the loop
returns the empty clutter, $\mathsf{false}$. The log prints a clutter as
its members separated by semicolons. In this notation, \texttt{H = \{\}} is the clutter
whose one member is the empty set, and \texttt{W = []} is the empty
clutter. The run made two goal runs and three replay runs, the counts of
Table~\ref{tab:calls}, and no replay was marked a preskeleton or dead.
The exhaustive evaluation in the same log finds all four trusts failing.

The first witness is not minimal: its clause $\{\Non(sk(S))\}$ strictly
contains the only minimal attack, $\emptyset$. Thus, the loop makes two
calls where $|\attacks(\chi)| + |\weakest(\chi)| = 1$, a case where the
bound of Theorem~\ref{thm:cost} fails without its hypothesis. The first
iteration alone is the one-shot method of Section~\ref{sec:obvious}, which
would have answered $\{\Non(sk(S))\}$.

\subsection{Cost}
\label{sec:cost}

An oracle call is a goal query, a call of Verify. The replays of
Algorithm~\ref{alg:stops} are counted separately. A \emph{distinct} oracle
call is a call on a trust that has not been queried before. The memo table of
Listing~\ref{lst:loop} counts these. Throughout, lines~4
and~5 query candidates in a fixed order until the first failure.
\emph{Monotone dualisation} is
the problem of deciding, given two clutters, whether each is the blocker of
the other. An enumeration is \emph{output-polynomial} if it runs in time
polynomial in the combined size of its input and its output.

\begin{figure}[t]
\centering
\begin{tikzpicture}[
    x=1cm, y=1cm, >=Stealth,
    dot/.style={circle, fill, inner sep=1.4pt},
    msg/.style={font=\small},
    lay/.style={font=\small, inner sep=2pt}]
  \node[font=\small\bfseries] at (0,0.6) {$B_j$};
  \node[font=\small\bfseries] at (11,0.6) {$A_j$};
  \node[dot] (b1) at (0,0) {};
  \node[dot] (b2) at (0,-2.6) {};
  \node[dot] (a1) at (11,-0.3) {};
  \node[dot] (a2) at (11,-1.9) {};
  \draw[double, -{Implies}, double distance=1.2pt] (b1) -- (b2);
  \draw[double, -{Implies}, double distance=1.2pt] (a1) -- (a2);
  \draw[->] (b1) -- node[msg, above, sloped] {$N_j$} (a1);
  \draw[->] (a2) -- node[msg, above, sloped, pos=0.55]
    {$t_j(N_j)$: regular route, the goal holds} (b2);
  \node[lay] (c0) at (10,-4.3) {$\langle N_j, \mathsf{tag}_j\rangle$};
  \node[lay] (c1) at (6.4,-4.3)
    {$\{\!|\langle N_j, \mathsf{tag}_j\rangle|\!\}_{k_{i_1}}$};
  \node[lay] (c2) at (3.6,-4.3) {$\cdots$};
  \node[lay] (c3) at (1.3,-4.3) {$t_j(N_j)$};
  \draw[->] (c0) -- node[msg, above] {$k_{i_1}$} (c1);
  \draw[->] (c1) -- node[msg, above] {$k_{i_2}$} (c2);
  \draw[->] (c2) -- node[msg, above] {$k_{i_r}$} (c3);
  \draw[->, dashed] (c3) -- node[msg, left, pos=0.5]
    {adversary route} (b2);
  \draw[dashed, rounded corners] (0.6,-3.75) rectangle (11.0,-4.85);
  \node[font=\small] at (5.8,-5.25)
    {adversary: originates $k_{i_1}, \dots, k_{i_r}$ and encrypts layer by layer};
\end{tikzpicture}
\caption{The gadget of the hardness reduction for one set $S_j = \{i_1 <
\dots < i_r\}$, with $t_j(N) = \{\!|\cdots\{\!|\langle N,
\mathsf{tag}_j\rangle|\!\}_{k_{i_1}}\cdots|\!\}_{k_{i_r}}$. Double
arrows are strand succession. The reception of $B_j$ is explained either
by a regular $A_j$ strand with parameter $N_j$, which satisfies the goal,
or by the adversary building the term, which requires every key of $S_j$.
Hence, the attack on this component has stopping set $\{\Non(k_i) : i \in
S_j\}$, and the tag keeps the components apart.}
\label{fig:gadget}
\end{figure}
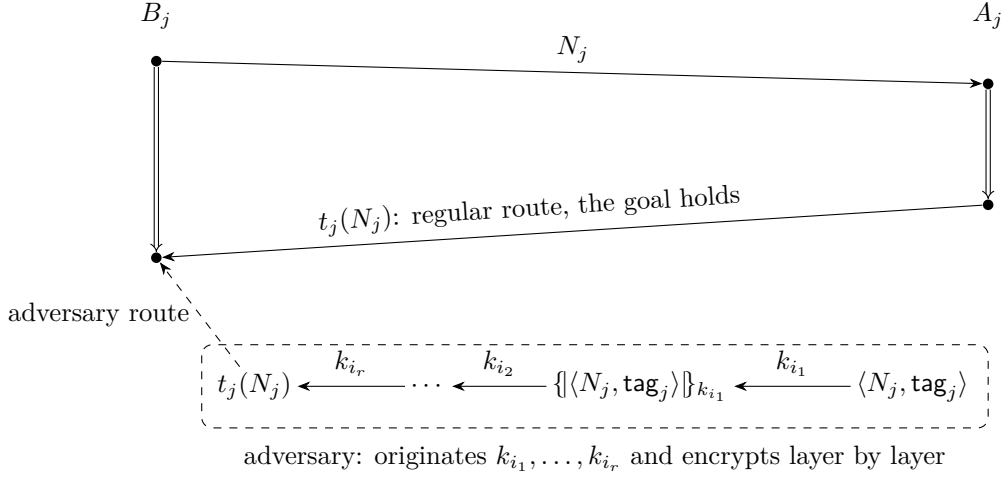

\begin{theorem}[Cost]
\label{thm:cost}
Let $N = |\attacks(\chi)| + |\weakest(\chi)|$. Given $\attacks(\chi)$
and a candidate clutter $C$, deciding whether $C =
\blocker(\attacks(\chi))$ is an instance of monotone dualisation and can be
solved in time $N^{o(\log N)}$ \cite{fredman}. Given only the oracle,
checking $C$ requires computing $\blocker(C)$ and making the $|C| +
|\blocker(C)|$ goal queries $T$ for $T \in C$ and $U \setminus E$ for $E
\in \blocker(C)$. Then $C = \weakest(\chi)$ exactly when every $T$ passes
and every $U \setminus E$ fails. Whether the minimal transversals can be
enumerated in output-polynomial time is a long-standing open question
\cite{eiter}. If every witness returned by the oracle
has a $\subseteq$-minimal stopping set, Algorithm~\ref{alg:loop} makes at
most $N$ distinct goal queries. With the backend of
Section~\ref{sec:backend}, it also makes at most $|U|$ replays for each
counterexample of the at most $|\attacks(\chi)|$ failing queries. With
arbitrary witnesses, it still terminates by
Theorem~\ref{thm:correctness}, every distinct passing
candidate is still a member of $\weakest(\chi)$, and the calls number at
most $|\weakest(\chi)| + |S_O|$, where $S_O$ is the set of stopping sets
the oracle can return. Minimalising each witness costs at most $|U|$
further goal queries and gives the bound $|\weakest(\chi)| +
|\attacks(\chi)|(1+|U|)$. Deciding whether some trust of size at most $k$
achieves $\chi$ is NP-hard.
\end{theorem}

\begin{proof}
Every passing candidate is a weakest trust and every failing candidate adds
a clause not seen before, so the calls are bounded by the number of weakest
trusts plus the number of distinct clauses. Hardness follows by encoding
Hitting Set as a protocol with nested encryptions.

All clauses below are conflicts, by Lemma~\ref{lem:conflict} or
Proposition~\ref{prop:clause}.

\emph{Dualisation and enumeration.} Given $\attacks(\chi)$ and a candidate
clutter $C$, deciding $C = \blocker(\attacks(\chi))$ is the monotone
dualisation problem, for which Fredman and Khachiyan give an algorithm
running in time $N^{o(\log N)}$ in the combined size \cite{fredman}.
Whether $\blocker(\attacks)$ can be enumerated in time polynomial in the
sizes of $\attacks$ and $\blocker(\attacks)$ is open \cite{eiter}.
Fredman and Khachiyan's algorithm requires both clutters as input. With
only the oracle, $\attacks(\chi)$ is not given. Suppose $C$ is an antichain.
If $C = \weakest(\chi)$, then every $T \in C$ passes, and by the
involution $\blocker(C) = \attacks(\chi)$, so for each $E \in \blocker(C)$
the trust $U \setminus E$ misses $E$ and fails. Conversely, if every $T
\in C$ passes, then $\mathop{\uparrow} C \subseteq \Suff(\chi)$. If every
$U \setminus E$ fails, then for $E \in \blocker(C)$ some member of
$\attacks(\chi)$ is disjoint from $U \setminus E$
(Proposition~\ref{prop:hitting}), that is, contained in $E$, and
Lemma~\ref{lem:refinement} gives $\Suff(\chi) = \mathrm{Tr}(\attacks(\chi))
\subseteq \mathrm{Tr}(\blocker(C)) = \mathop{\uparrow} C$. Therefore,
$\Suff(\chi) = \mathop{\uparrow} C$ and $C = \weakest(\chi)$.

\emph{Passing candidates.} Let $T$ be a candidate for which $O$ answers
achieves. $T$ is a minimal transversal of the current $H$ and is
sufficient, so it contains some $T^* \in \weakest(\chi)$. Being sufficient,
$T^*$ meets every member of $H$, since each is a conflict. Hence, $T^*$ is
a transversal of $H$, and minimality of $T$ forces $T = T^* \in
\weakest(\chi)$. Distinct passing calls therefore number at most
$|\weakest(\chi)|$, with or without minimal witnesses.

\emph{Failing candidates.} A failing candidate $T$ is picked at line~5,
because candidates are queried until the first failure, so every distinct
failing call adds a clause $K$ at line~6. Afterwards $T \notin
\mathrm{Tr}(H)$, so $T$ is never a candidate again and its memoised
answer is never reused. Each such $K$ is new: $T$ meets every member of the
current $H$ and misses $K$, so $K$ contains no member of $H$. A set
added earlier and later removed by minimisation was removed because a
subset of it entered $H$, and $H$ thereafter always contains a subset of
it, so $K$ is not one of those either. Hence the failing calls number at
most the number of distinct clauses the oracle can return, $|S_O|$, and
the total is at most $|\weakest(\chi)| + |S_O|$.

\emph{Minimal witnesses.} If every clause is a member of
$\attacks(\chi)$, the distinct clauses added are distinct members of
$\attacks(\chi)$, so the failing calls number at most $|\attacks(\chi)|$
and all calls at most $N$. The hypothesis is a property of the oracle: a
shape is minimal in CPSA's homomorphism order on skeletons, which says
nothing about inclusion of stopping sets, and Section~\ref{sec:trace}
gives a CPSA run that exceeds the bound.

\emph{Minimalising a witness.} By Proposition~\ref{prop:hitting}, a trust
fails exactly when it misses a member of $\attacks(\chi)$. Thus, for $S
\subseteq U$, the trust $U \setminus S$ fails exactly when $S$ contains a
member of $\attacks(\chi)$. Given a failing clause $s$, set $S := s$ and, for each
$a \in s$ in turn that is still in $S$, query $U \setminus (S \setminus
\{a\})$. If it fails with clause $s'$, then $s'$ is disjoint from that
trust, so $s' \subseteq S \setminus \{a\}$, and $S := s'$. If it passes,
$a$ stays. Each element of $s$ is tested at most once, so this costs at
most $|s| \le |U|$ calls. $S$ is always a clause the oracle returned, hence
a conflict. At the end, $U \setminus S$ misses $S$ and is therefore not
sufficient, so $S$ contains some $E \in \attacks(\chi)$. If $E \neq S$,
take $a \in S \setminus E$: when $a$ was tested, the current set $S_0
\supseteq S$ had $S_0 \setminus \{a\} \supseteq E$, so the query failed
and replaced $S_0$ by a subset of $S_0 \setminus \{a\}$, and $a$, once
removed, never returns, contradicting $a \in S$. Therefore, $S = E \in
\attacks(\chi)$, and $S \subseteq s$ is disjoint from the failing
candidate, so the argument for minimal witnesses applies to the clauses
added. Each failure then costs at most $1 + |U|$ distinct calls, and the
total is at most $|\weakest(\chi)| + |\attacks(\chi)|(1 + |U|)$.

\emph{Hardness.} By reduction from Hitting Set \cite{karp}: given sets
$S_1, \dots, S_m \subseteq \{1, \dots, n\}$ and a bound $k$, decide
whether some $I \subseteq \{1, \dots, n\}$ with $|I| \le k$ meets every
$S_j$, that is, whether the family has a hitting set of size at most $k$.
The construction generalises the two-key variant of
Section~\ref{sec:example} (Fig.~\ref{fig:gadget}). Take symmetric keys
$k_1, \dots, k_n$, distinct constants $\mathsf{tag}_1, \dots,
\mathsf{tag}_m$, and nonces $N_1, \dots, N_m$. For each $j$ with $S_j =
\{i_1 < \dots < i_r\}$, let
\[
t_j(N) \;=\; \{\!|\cdots\{\!|\{\!|\langle N,
\mathsf{tag}_j\rangle|\!\}_{k_{i_1}}|\!\}_{k_{i_2}}\cdots|\!\}_{k_{i_r}},
\]
the tagged nonce encrypted successively under the keys of $S_j$. For each
$j$ the protocol has a role $B_j$ that transmits a nonce $N$ and then
receives $t_j(N)$, and a role $A_j$, played by the agent $A$, that
receives a nonce $N$ and transmits $t_j(N)$. The role $A_j$ is executable since
$A$ holds every key. The variable $N$ has the sort of nonces, so no
regular strand transmits a term containing an encryption except an $A_j$
strand, and only an $A_j$ strand transmits an encryption carrying
$\mathsf{tag}_j$. Keys are used only as encryption keys, so no regular strand
originates any $k_i$. Let $\sigma$ consist of one completed $B_j$ strand
with parameter $N_j$ for each $j$, together with $\Unq(N_j)$, so that $U
= \{\Non(k_1), \dots, \Non(k_n)\}$. Let $\chi$ be the conjunction over $j$
of non-injective agreement of $B_j$ with $A$ on $N_j$: an $A_j$ strand with
parameter $N_j$ exists.

Consider a run meeting $\sigma$ that violates the conjunct for $j$, so no
$A_j$ strand has parameter $N_j$. The $B_j$ strand receives $t_j(N_j)$.
Every encryption carrying $\langle N_j, \mathsf{tag}_j\rangle$ that a
regular strand transmits, and every layer obtained from one by
decryption, comes from an $A_j$ strand with parameter $N_j$, and there is
none. Thus, the adversary constructs every layer of $t_j(N_j)$, from the
innermost outwards, and each layer under $k_i$ requires $k_i$. The
adversary holds $k_i$ only if it originates it, since no regular strand
does. Hence, the run originates every $k_i$ with $i \in S_j$, and each $\Non(k_i)$ with
$i \in S_j$ rules it out. Hence every attack $x$ has $\stops(x) \supseteq
K_j = \{\Non(k_i) : i \in S_j\}$ for some violated conjunct $j$.
Conversely, for each $j$ there is an attack $x_j$: the adversary
originates the keys of $S_j$, builds $t_j(N_j)$ and delivers it to $B_j$,
while every other $B_{j'}$ is answered by a regular $A_{j'}$ strand with
parameter $N_{j'}$. In $x_j$ the only keys that originate are those of
$S_j$ and every nonce originates once, so $\stops(x_j) = K_j$. The tags
make the $t_j$ pairwise non-unifiable, so no message of one component is
accepted in another. Therefore $\attacks(\chi)$ is the family of
$\subseteq$-minimal members of $\{K_1, \dots, K_m\}$, and a set meets
every $K_j$ exactly when it meets every minimal one. By
Proposition~\ref{prop:hitting}, a trust $\tau$ achieves $\chi$ exactly when
$\{i : \Non(k_i) \in \tau\}$ meets every $S_j$, and $|\tau|$ is the size of
that set. Therefore, a trust of size at most $k$ achieves $\chi$ exactly when
the instance has a hitting set of size at most $k$. If some $S_j$ is empty, $t_j(N_j)$ is a
plaintext the adversary builds without keys, $\emptyset$ is a minimal
attack, and both sides are false. The protocol has $2m$ roles whose
messages have total size $O(\sum_j |S_j|)$, so the construction is
polynomial.
\end{proof}

The bound of $N$ goal queries is an instance of known results on
learning monotone Boolean functions with membership queries and on
dualisation, which Section~\ref{sec:related} cites.

The number of assumptions is a poor measure of what a trust costs. A
secure channel to a hardware root of trust and the freshness of a nonce
are one assumption each, but their deployment costs differ widely.
A better formulation gives the atoms a weight $w : U \to \mathbb{R}_{\ge
0}$ and asks for a transversal of minimum total weight. That problem is
still NP-hard, since unit weights give back the cardinality version.
Since $\weakest(\chi)$ does not depend on the objective, we compute all of
it instead of optimising over it. It is computed once, and the weights can
be chosen later, at deployment.

\section{Implementation and Evaluation}
\label{sec:eval}

Theorems~\ref{thm:correctness} and~\ref{thm:cost} prove that
Algorithm~\ref{alg:loop} computes the weakest trust and bound the number of
verifier calls it makes. The experiments test both results on real
protocols and a real verifier, and interpret what the answers say about
the protocols. The questions were fixed with the experiment design.
\begin{itemize}
\item Does the loop return the weakest trust on real protocols, with no
  assumption chosen by hand, and which assumptions does each goal require?
  (Section~\ref{sec:eval-rq1})
\item How many verifier calls does it make, against the bound
  $|\attacks(\chi)| + |\weakest(\chi)|$ of Theorem~\ref{thm:cost}, and what
  do they cost with CPSA? (Section~\ref{sec:eval-rq2})
\item Does the implementation keep to the call bound, and how does its time
  grow, as the numbers of assumptions, minimal attacks, and weakest trusts
  grow? (Section~\ref{sec:eval-scale})
\item Which answers depend on the bounds of our own analyser?
  (Section~\ref{sec:eval-bound})
\item Which answers depend on the model of channels, and does trust compose
  as Theorem~\ref{thm:composition} predicts?
  (Section~\ref{sec:eval-channels})
\end{itemize}

\subsection{Setup}
\label{sec:eval-impl}
\label{sec:eval-protocols}

\paragraph{Implementation.} We implemented Algorithm~\ref{alg:loop} as an
OCaml library of \numLLibLines{} lines (implementation files, without
interfaces), with no dependency beyond the standard library and
\texttt{unix}. The loop and the blocker computation are independent of the
protocol. Only the \emph{oracle}, the verifier the loop queries, knows about
protocols, and we instantiated it in two ways. Listing~\ref{lst:loop} shows
the main loop. Its memo table makes each distinct candidate trust cost one
oracle call, as Theorem~\ref{thm:cost} counts. Besides the loop, the
library contains operations on families of sets (removing non-minimal members,
and Berge's algorithm for the blocker), a Dolev--Yao term algebra with a
check of what the adversary can derive, a bounded analyser with the protocol
models of its rows, and a driver for CPSA. A Python prototype written
earlier serves as a reference for differential testing and is not part of
the artifact. The artifact, described at the end of this subsection,
contains the code, the protocol models, and the scripts that regenerate every table
and every number in this section from the raw results.

\begin{lstlisting}[language={[Objective]Caml},caption={The loop of Algorithm~\ref{alg:loop}.},label={lst:loop}]
let weakest ~u oracle =
  let memo = ref Amap.empty in             (* trust -> answer *)
  let ask t = match Amap.find_opt t !memo with
    | Some r -> r                           (* one call per trust *)
    | None -> let r = oracle t in
              memo := Amap.add t r !memo; r in
  let failure t =
    match ask t with Fails w -> Some w | Achieves -> None in
  let rec loop h =                          (* minimal attacks so far *)
    let c = Clutter.blocker h in            (* blocker [] = [empty] *)
    match List.find_map failure c with
    | None -> c                             (* every candidate achieves *)
    | Some w -> loop (Clutter.minimalize (w.stops :: h)) in
  loop []
\end{lstlisting}

\paragraph{Correctness checks.} No experiment was reported before the
following checks passed, and the final run repeated them. Unit and property
tests (\numLNTestsPassed{} checks, \numLNTestsFailed{} failures) cover the
blocker involution $\blocker(\blocker(H)) = \min_\subseteq H$, agreement of
the blocker with brute-force minimal transversals on random families, and
agreement of the loop with the brute-force minimal sufficient trusts on
random monotone oracles in three witness modes (minimal, padded, and padded
then shrunk). Each case also checks $H = \attacks(\chi)$ and the call
bounds of Theorem~\ref{thm:cost}. A differential test runs \numLNDiffCore{}
random families through the OCaml library and the Python prototype and
compares the answers and the call counts exactly (\numLNDiffCoreMismatch{}
mismatches). The two implementations of the bounded analyser print identical
results on all analyser rows (\numLNDiffBoundedLines{} lines), and the CPSA
driver agrees with the corrected Python CPSA oracle on all
\numLNDiffCpsaRows{} CPSA rows in each of
\numLNDiffCpsaFields{} fields (exhaustive $\weakest$ and $\attacks$, and,
per witness mode, $\weakest$, $H$, distinct calls, goal runs, and replay
runs).

\paragraph{The bounded analyser.} One oracle is a bounded Dolev--Yao
analyser of our own. It searches the runs that a fixed \emph{pool} of
strands can produce, with the adversary acting as the network. A strand is
one principal's part in one session. The search covers all interleavings,
depth first, and skips states it has already seen. Variables are
instantiated by atoms only, so received messages are no deeper than the
role templates. Whether the adversary can derive a message is decided
exactly over the subterms of its knowledge and of the target. Thus,
type-flaw attacks, which require a variable bound to a compound message,
are outside the model. The trust enters in
three places: the adversary's initial knowledge (it knows every key not
assumed uncompromised), the channels on which it may not inject or read,
and whether the test nonce may be a stale one. Appendix~\ref{app:protocols}
lists the pools. For the challenge-response rows, the pool contains one
instance of each role for each ordered pair of $A$ and $B$, a prover $A$
for the adversary, and one earlier completed session whose transcript the
adversary holds. Without $\Unq(N)$, the test nonce may be that session's.
The analyser's answers about freshness therefore depend on this earlier session
being in the pool. For Needham--Schroeder
and its fix, the pool contains the initiator $A$ with peers $B$ and $I$,
and the test responder $B$. The stopping set of a run is computed
by replaying the run under each further atom.

A \emph{witness} is the attack the oracle returns when a candidate trust
fails. The analyser offers three witness modes. \emph{As found} is the first
attack state the search reaches. \emph{Pruned} is that run cut down to a
minimal run by dropping trailing events of strands other than the test
strand, as long as it remains an attack with no larger stopping set.
\emph{Least} is a run whose stopping set is smallest among all attack
states the search reaches, and is therefore minimal under inclusion. A
minimal run need not have a minimal stopping set.
Least witnesses satisfy the hypothesis of the call bound by construction,
at the cost of exploring every attack state. For the bound study, the search uses two reductions that
preserve verdicts and minimal stopping sets. First, a strand's receives
are executed together with its next send. Second, copies of a role without
fresh values of their own are treated as interchangeable.

\paragraph{CPSA.} The second oracle is CPSA~\numLCpsaVersion{}
\cite{cpsaspec} with its default bounds (step limit \numLCpsaStepLimit{},
strand bound \numLCpsaStrandBound{}) and a limit of \numLCpsaTimeout\,s per
run. The trust is added as \texttt{non} and \texttt{uniq} facts to the
hypothesis of a CPSA goal, and the verdict is read from CPSA's
\texttt{(satisfies)} annotation (Listing~\ref{lst:oracle}). For each failing
shape, the replay procedure of Section~\ref{sec:backend}
(Algorithm~\ref{alg:stops}) computes its clause, the atoms whose assumption
removes it, and the oracle returns the shape with the smallest clause. A
replay counts as surviving only when CPSA reports the input realised and
marks it neither a preskeleton nor dead, for the reason given in
Section~\ref{sec:backend}. An earlier version of our driver read both
marks as survival. On Otway--Rees with $A \neq B$, it then returned
$\mathsf{false}$, whereas evaluating every trust gives $\{\Non(k_{AS}),
\Non(k_{BS}), \Unq(n_B)\}$. In the final runs, the two marks decided
\numLNCpsaRejected{} replays on \numLNCpsaRejectedRows{} rows. Every CPSA
run in this section exhausted its search without reaching the step or
strand bound, and none wrote to standard error.

\paragraph{Protocols.} Besides the \numLNAnalyserRows{} rows on our
analyser (the running example, its two-key variant, Needham--Schroeder with
and without channel assumptions, and Lowe's fix), we took protocols shipped
with CPSA in its \texttt{tst/} directory, using a selection rule fixed
before any run. The rule requires the following.
\begin{itemize}
\item The protocol text is shipped and used unchanged.
\item The goal is a safety goal over the fragment, either agreement of one
  role with another or secrecy. It is the file's \texttt{defgoal} if it has
  one. Otherwise, it is the standard agreement goal for the role of the
  file's analysis (marked in the template as written by us), or secrecy
  where the file's analysis is a secrecy query. No file the rule selected
  has a \texttt{defgoal} for its protocol, so all \numLNCpsaRows{} CPSA
  goals are ours.
\item $U$ consists of all \texttt{non} facts on long-term keys of the
  principals the observed role names, and all \texttt{uniq} facts on its
  nonces, timestamps, and session keys that no role declares
  \texttt{uniq-orig}. Names, lifetimes, and free-text fields are not in
  $U$, and role-level \texttt{non-orig} and \texttt{uniq-orig} declarations
  stay in the situation.
\item The candidate assumptions satisfy $2 \le |U| \le 8$.
\item There is no equational theory (A5).
\item CPSA terminates under its defaults at every point of the trust
  lattice, within \numLSelectTimeout\,s per run during selection (the final
  runs allowed \numLCpsaTimeout\,s).
\item There are at most \numLRowCap{} rows in total, analyser rows included.
\end{itemize}
Protocols on a list fixed in advance took precedence. Blanchet's flawed and
fixed exchange was admitted from outside the list because it supplies the
only secrecy goal, and with it the cap was reached. The remaining shipped
protocols outside the list were not evaluated. The adopted \textsc{pkinit}
fix is the only row whose protocol text is not shipped. We wrote it in the
style of the shipped flawed version (Appendix~\ref{app:protocols-cpsa}).
Wide-Mouthed Frog, Woo--Lam, the two-pass unilateral ISO 9798-3 exchange,
and the Diffie--Hellman and station-to-station files fail the rule, and
Appendix~\ref{app:protocols-cpsa} gives the reason for each. The analyser
rows cover the CPSA versions of Needham--Schroeder and Lowe's fix.

Two rows were added after the first results, outside the rule and the cap,
because the rule's goal failed on them for a reason we traced to the model
or the goal. Under the full trust, every failing run of Otway--Rees has $A
= B$, with the server accepting $B$'s own block in the initiator's
position. We report the protocol as modelled and add a row with $A \neq B$
in the situation, which adds a rule to the protocol text. CPSA's simplified
Kerberos issues two tickets of the same form, so the initiator's agreement
on the names, the goal the rule produced, fails at every trust. We report
that row, with answer $\mathsf{false}$, and add a row for the initiator's
agreement on the session key alone. Both added rows are second readings of
a rule row, and Table~\ref{tab:results} shows both readings.
Appendix~\ref{app:protocols} gives each CPSA row's file, goal (verbatim),
$U$ with its CPSA facts, and the declarations that stay in the situation.

\paragraph{Reference answers and baselines.} Every row was run in every
witness mode its oracle offers. As a reference, we evaluated achievement
at all $2^{|U|}$ trusts using the same oracle and took the minimal
sufficient trusts from that evaluation. The minimal attacks are their
blocker. This reference shares the oracle with the loop, so agreement with
it tests the loop and the stopping-set computation, but not the oracle.
Calls are distinct oracle calls, the goal
queries that Theorem~\ref{thm:cost} counts. They do not include the replay
runs of Algorithm~\ref{alg:stops}. On CPSA, we also count the runs of CPSA,
split into goal runs (one per call) and replay runs (the runs that compute
stopping sets). We compare against three baselines. B1 evaluates all
$2^{|U|}$ trusts. B2 enumerates trusts by size and skips supersets of
trusts found sufficient. B3 is the conjunct-dropping method of
\cite{rowe}. Starting from a sufficient trust, it drops atoms one at a time
as long as the goal still holds, and returns one minimal sufficient trust.

\paragraph{Machine and timing.} All runs except those on CPSA were made in
one session from \numLRunStart{} to \numLRunEnd{} on one machine
(\numLMachine, \numLOS), with most of the time spent in the scale and
bound studies.
The CPSA step (all CPSA rows, their call and run counts, and their timing)
was rerun alone from \numLCpsaRunStart{} to \numLCpsaRunEnd{}, on the same
machine with the same tool versions, after the Kerberos name-agreement row
was added. Each CPSA row was timed
sequentially, \numLCpsaReps{} times after one warm-up, for the loop with
witnesses as found and with shrinking and for the three baselines. We
also recorded the share of wall time spent inside CPSA. Each run of the
scale study had a limit of \numLScaleTimeout\,s and \numLScaleHeapCap\,MB
of heap. Its call counts come from a pass with \numLScaleWorkers{} runs in
parallel, and its times come from a separate sequential pass, on seed 0 of every instance
and on all seeds of a few rows. The bound study allowed
\numLBoundTimeout\,s per run of Algorithm~\ref{alg:loop}. Times are
reported as medians with interquartile ranges.

\paragraph{Reproduction.}
\label{sec:eval-repro}
The artifact, including the library, the experiment drivers, the CPSA
templates, and the raw results of the runs reported here, is available at
\url{https://github.com/aalok-thakkar/minimal-trust-tool}. It requires
OCaml~\numLOcamlVersion{} and dune~\numLDuneVersion{}, Python~3 for the
table scripts, and
CPSA~\numLCpsaVersion{}, found as \texttt{cpsa4} through the variable
\texttt{CPSA4}, the \texttt{PATH}, or \verb|~/.local/bin|. The commands are:
\begin{lstlisting}[language={}]
dune build --profile release
dune test --profile release            # unit and property tests
scripts/run_all.sh                     # all experiments -> results/final
python3 scripts/make_tables.py    results/final
python3 scripts/scale_table.py    results/final
python3 scripts/make_tables_long.py results/final OUT_DIR
\end{lstlisting}
\texttt{run\_all.sh} runs, in order, the tests, the differential checks, the
analyser rows, the channel and composition studies, the CPSA rows with their
timing, the scale study (counts, then timing), the gadget check, and the
bound study. It also logs the load average around each of its
\numLNSteps{} steps. It writes a manifest with the machine, the tool
versions, and a SHA-256 digest of the sources (\texttt{\numLSourcesSha\dots} for the runs
reported here). The table scripts read only \texttt{results/final} and the
templates, and write every table and every number in this section.

\subsection{Answers on Protocols}
\label{sec:eval-rq1}

On \numLNRowsAgree{} of the \numLNRows{} rows, and in every witness mode,
the loop returned the weakest trust found by evaluating every trust with
the same oracle. Table~\ref{tab:results} lists all rows, with $|U|$ from
\numLMinU{} to \numLMaxU{}. It also shows the channel row at two instances
per role, which was not checked exhaustively. The stopping sets the loop
collected were then the minimal attacks, but this is not a separate check.
The loop returns $\blocker(H)$, and $\blocker(H) = \weakest(\chi)$ gives
$\min_\subseteq H = \attacks(\chi)$ by the involution. On all
\numLNCpsaMonotone{} CPSA rows, achievement was monotone over the
\numLNCpsaPoints{} trusts evaluated, as Lemma~\ref{lem:monotone} requires
of a sound and complete oracle. This comparison tests the loop and the
stopping-set computation, but not the oracle, which both sides use. It
also cannot detect stopping sets that are too large, since these cost
extra calls but leave the answer unchanged (Section~\ref{sec:backend}).

On the running example, the loop takes \numLSignedPruned{} calls with
pruned witnesses. The forgery fails the empty trust and forces
$\Non(sk_A)$, the replay fails the next candidate and forces $\Unq(N)$, and
the last candidate achieves the goal. This is the unique weakest trust of
Table~\ref{tab:example}. The two-key variant takes \numLTwoKeyPruned{}
calls and returns the two incomparable weakest trusts, with minimal
attacks $\{\Unq(N)\}$ and $\{\Non(k_1),\Non(k_2)\}$. This is the example of
Theorem~\ref{thm:nofilter}, computed by the loop. The two rows cover both
cases of Theorem~\ref{thm:unique}, one with a single weakest trust and one
with two.

The answer is $\mathsf{false}$, meaning that no trust over $U$ suffices,
on \numLNFalse{} rows: Needham--Schroeder,
the flawed \textsc{pkinit} \cite{cervesato}, the flawed ISO 9798-3 exchange
and Blanchet's example shipped with CPSA, Otway--Rees as modelled, and the
simplified Kerberos with agreement on the names. On the first four, the
answer comes from the known attack, whose stopping set is empty. On
Otway--Rees, it comes from the $A = B$ run. On Kerberos, it comes from a
run in which the two tickets, which have the same form, are swapped, so
that $A$'s own responder strand answers $A$'s initiator
(Appendix~\ref{app:protocols-cpsa}). This is a property of the simplified
model, not a known attack on Kerberos~V5.

An answer names the assumptions a goal requires, for the given situation
$\sigma$, goal, and $U$. A candidate assumption that appears in no weakest
trust is one the goal does not need, although another goal may need it.
Two of our choices fix part of every CPSA
answer in advance. First, values that a role declares \texttt{uniq-orig},
such as the session keys of Yahalom and Kerberos and the nonces of
\textsc{pkinit}, are already fresh in $\sigma$ and are not in $U$. Second,
every agreement goal is non-injective, so no goal asks for recency. Thus,
an answer that needs no freshness atom is an answer for this $\sigma$ and
goal. It does not mean that the protocol needs no freshness. Within these
choices, most answers leave out much of $U$, which contains every
long-term key and every nonce not declared fresh, as a uniform hypothesis such as ``all honest keys
are uncompromised and all nonces are fresh'' does.
\begin{itemize}
\item \textsc{pkinit}, adopted fix: the client's authentication of the
  server requires only the server's key. The client's own key, which the
  deployed hypothesis also includes, is not needed for this goal. The claim
  concerns our model of the fix, which is not shipped with CPSA. The model
  simplifies the standard's keyed hash to an unkeyed hash of the client's
  request, with the freshness of the client's nonces and of the session key
  in $\sigma$ (Appendix~\ref{app:protocols-cpsa}). It is a claim about this
  goal only. The client's key may be needed for other goals, such as
  secrecy of the session key, which we did not evaluate. What each repair
  relies on is the question the debate of Section~\ref{sec:intro} turned
  on.
\item In the signed challenge-response and in Lowe's fix of
  Needham--Schroeder, $B$'s goal requires $\Non(sk_A)$, and the
  challenge-response also requires $\Unq(N)$, given the earlier session in
  the pool. Neither requires $B$'s own key $sk_B$.
\item ISO 9798-3, corrected, and Denning--Sacco: the responder's agreement
  requires only $\Non(sk_A)$. On Denning--Sacco, the omission of the server's
  key is an artefact of CPSA's \texttt{pubk}, which is a function of the
  name, so the server's certificates bind nothing that the key does not
  already bind.
\item Yahalom, Neuman--Stubblebine, and the simplified Kerberos with
  agreement on the session key require both long-term keys and no
  freshness atom of $U$, for the reasons above.
\item Otway--Rees with $A \neq B$ requires both long-term keys and the
  freshness of $B$'s nonce, since without it, $B$ accepts the server's
  reply for another session that reused $n_B$.
\item The fixed Blanchet exchange is the only row that requires every
  candidate assumption.
\end{itemize}
The $\mathsf{false}$ answers separate ``assume more'' from ``no assumption
over these keys and values saves the goal''. Four of these protocols have a
fix in the table, namely Lowe's fix, the adopted \textsc{pkinit} fix, the
corrected ISO exchange, and the fixed Blanchet exchange. Each fix turns
$\mathsf{false}$ into a trust, and that trust is what the fix relies on. On
Otway--Rees and Kerberos, the second reading, rather than a fix, turns
$\mathsf{false}$ into a trust. Channel assumptions added to the $U$ of Needham--Schroeder also turn
$\mathsf{false}$ into trusts that name channels, under the inject/read
reading of Section~\ref{sec:channels} and with one or two instances per
role in the pool, as Section~\ref{sec:eval-channels} reports.

\begin{table}[tbp]
\caption{Every protocol row, with candidate assumptions $U$, the weakest trust $\weakest(\chi)$ returned by Algorithm~\ref{alg:loop}, and the number of minimal attacks $|\attacks|$, using our analyser (A) or CPSA (C). ``agrees'' means that in every witness mode (Table~\ref{tab:calls}), the loop returned the minimal sufficient trusts found by evaluating all $2^{|U|}$ trusts with the same oracle, and n/c means not checked. $k_{XS}$ is $X$'s long-term key shared with the server. The goals are (recent) agreement \cite{lowe} of the responder with the initiator, the client's authentication of the server, secrecy of $d$, and the initiator's agreement on the names and the session key $K$ (Kerberos), or on $K$ alone (agreement on $K$). $\mathsf{false}$ marks an unrealisable goal (Remark~\ref{rem:unrealisable}). The channel row is shown with one instance per role (the pool on which it was checked) and with two (Section~\ref{sec:eval-bound}). The goals and $U$ of the CPSA rows are in Appendix~\ref{app:protocols}.}
\label{tab:results}
\centering
\setlength{\tabcolsep}{2.5pt}
\footnotesize
\begin{tabular}{l c l r l l r c}
\toprule
protocol & & goal & $|U|$ & $U$ & $\weakest(\chi)$ & $|\attacks|$ & agrees\\
\midrule
signed CR & A & rec.\ agr. & 3 & $\Non(sk_A)$, $\Non(sk_B)$ & $\{\Non(sk_A), \Unq(N)\}$ & 2 & yes\\
 & & & & $\Unq(N)$ &  & & \\
\quad two-key variant & A & rec.\ agr. & 3 & $\Non(k_1)$, $\Non(k_2)$ & $\{\Non(k_1), \Unq(N)\}$, & 2 & yes\\
 & & & & $\Unq(N)$ & $\{\Non(k_2), \Unq(N)\}$ & & \\
Needham--Schroeder & A & agr. & 2 & $\Non(sk_A)$, $\Non(sk_B)$ & $\mathsf{false}$ & 1 & yes\\
\quad with channels & A & agr. & 6 & $\Non(sk_A)$, $\Non(sk_B)$ & $\{\auth(A{\to}B)\}$, & 1 & yes\\
 & & & & $\auth(A{\to}B)$, $\auth(B{\to}A)$ & $\{\conf(B{\to}A)\}$ & & \\
 & & & & $\conf(A{\to}B)$, $\conf(B{\to}A)$ &  & & \\
\quad\quad two instances & A & agr. & 6 & & $\{\conf(B{\to}A)\}$, & 2 & n/c\\
 & & & & & $\{\Non(sk_A), \auth(A{\to}B)\}$, & & \\
 & & & & & $\{\auth(A{\to}B), \auth(B{\to}A)\}$ & & \\
\quad Lowe's fix & A & agr. & 2 & $\Non(sk_A)$, $\Non(sk_B)$ & $\{\Non(sk_A)\}$ & 1 & yes\\
\textsc{pkinit} (flawed) & C & auth. & 2 & $\Non(sk(S))$, $\Non(sk(C))$ & $\mathsf{false}$ & 1 & yes\\
\quad adopted fix & C & auth. & 2 & $\Non(sk(S))$, $\Non(sk(C))$ & $\{\Non(sk(S))\}$ & 1 & yes\\
ISO 9798-3 (flawed) & C & agr. & 4 & $\Non(sk_A)$, $\Non(sk_B)$ & $\mathsf{false}$ & 1 & yes\\
 & & & & $\Unq(n_A)$, $\Unq(n_C)$ &  & & \\
\quad corrected & C & agr. & 4 & $\Non(sk_A)$, $\Non(sk_B)$ & $\{\Non(sk_A)\}$ & 1 & yes\\
 & & & & $\Unq(n_A)$, $\Unq(n_C)$ &  & & \\
Blanchet & C & secr.\ $d$ & 4 & $\Non(sk_A)$, $\Non(sk_B)$ & $\mathsf{false}$ & 1 & yes\\
 & & & & $\Unq(s)$, $\Unq(d)$ &  & & \\
\quad fixed & C & secr.\ $d$ & 4 & $\Non(sk_A)$, $\Non(sk_B)$ & $\{\Non(sk_A), \Non(sk_B),$ & 4 & yes\\
 & & & & $\Unq(s)$, $\Unq(d)$ & $\Unq(d), \Unq(s)\}$ & & \\
Otway--Rees & C & agr. & 4 & $\Non(k_{AS})$, $\Non(k_{BS})$ & $\mathsf{false}$ & 1 & yes\\
 & & & & $\Unq(M)$, $\Unq(n_B)$ &  & & \\
\quad with $A \neq B$ & C & agr. & 4 & $\Non(k_{AS})$, $\Non(k_{BS})$ & $\{\Non(k_{AS}), \Non(k_{BS}),$ & 3 & yes\\
 & & & & $\Unq(M)$, $\Unq(n_B)$ & $\Unq(n_B)\}$ & & \\
Yahalom & C & agr. & 4 & $\Non(k_{AS})$, $\Non(k_{BS})$ & $\{\Non(k_{AS}), \Non(k_{BS})\}$ & 2 & yes\\
 & & & & $\Unq(n_A)$, $\Unq(n_B)$ &  & & \\
Denning--Sacco & C & agr. & 5 & $\Non(sk_A)$, $\Non(sk_B)$ & $\{\Non(sk_A)\}$ & 1 & yes\\
 & & & & $\Non(sk_{KS})$, $\Unq(K)$ &  & & \\
 & & & & $\Unq(t_A)$ &  & & \\
Neuman--Stubblebine & C & agr. & 5 & $\Non(k_{AS})$, $\Non(k_{BS})$ & $\{\Non(k_{AS}), \Non(k_{BS})\}$ & 2 & yes\\
 & & & & $\Unq(r_A)$, $\Unq(r_B)$ &  & & \\
 & & & & $\Unq(t_B)$ &  & & \\
Kerberos (simplified) & C & init.\ agr. & 4 & $\Non(k_{AS})$, $\Non(k_{BS})$ & $\mathsf{false}$ & 1 & yes\\
 & & & & $\Unq(t)$, $\Unq(t')$ &  & & \\
\quad agreement on $K$ & C & agr.\ $K$ & 4 & $\Non(k_{AS})$, $\Non(k_{BS})$ & $\{\Non(k_{AS}), \Non(k_{BS})\}$ & 2 & yes\\
 & & & & $\Unq(t)$, $\Unq(t')$ &  & & \\
\bottomrule
\end{tabular}
\end{table}

\subsection{Oracle Calls and Cost}
\label{sec:eval-rq2}
\label{sec:eval-cost}

When the oracle's witnesses have minimal stopping sets, the loop makes
exactly $N = |\attacks(\chi)| + |\weakest(\chi)|$ calls, matching the
bound of Theorem~\ref{thm:cost}. Calls here are goal queries, and the
replay runs that compute stopping sets are counted separately. Otherwise,
the excess is the number of non-minimal stopping sets the oracle returned.
Table~\ref{tab:calls} gives the counts per row and witness mode. On the
analyser rows, least witnesses meet the bound on all
\numLNAnalyserLeastAtBound{} rows, and pruned witnesses meet it on
\numLNAnalyserPrunedAtBound{}. The exceptions are the two
Needham--Schroeder rows, where a pruned run still has a non-minimal
stopping set that a later witness replaces. With witnesses as found, the
signed challenge-response takes \numLSignedAsFound{} calls. The first attack
the search reaches there is the double forgery of Table~\ref{tab:example},
with stopping set $\{\Non(sk_A), \Non(sk_B)\}$, which the forgery's
$\{\Non(sk_A)\}$ later replaces.

CPSA's shapes are minimal runs but need not have minimal stopping sets. On
\numLNCpsaOverBound{} of the \numLNCpsaRows{} CPSA rows, the calls exceed
the bound, as Theorem~\ref{thm:cost} allows. Over all CPSA rows, the loop
made
\numLNCpsaGoalRuns{} goal runs and \numLNCpsaReplayRuns{} replay runs.
Shrinking the witnesses keeps the answer correct on all \numLNShrinkAgree{}
CPSA rows. It lowers the replay runs on \numLNShrinkFewerReplays{} rows
(\numLNCpsaShrinkReplayRuns{} in all), where the first failing goal run
returns many shapes, but raises the distinct calls on
\numLNShrinkMoreCalls{}, since each shrink step is itself a call. B2 makes
fewer CPSA runs than the loop on \numLNLevelwiseFewerCalls{} rows and finds
$\weakest(\chi)$ on all \numLNLevelwiseAgree. B3 is the cheapest method,
and its answer is consistent with $\weakest(\chi)$ on all \numLNGreedyInW{} CPSA rows (a
member of it, or no answer when it is $\mathsf{false}$), as it must be for
an oracle whose verdicts are monotone in the trust.

With a real verifier, the cost is the number of verifier runs, replay runs
included. CPSA accounts for \numLCpsaShareMin--\numLCpsaShareMax\% of the
wall time in every method, and
every method takes at most \numLCpsaMaxAny\,s per row (median), with
interquartile ranges of at most \numLCpsaMaxIqr\,s
(Table~\ref{tab:timing}). The loop is faster than evaluating all
$2^{|U|}$ trusts on \numLNLoopFasterThanExh{} of the \numLNCpsaRows{}
rows, B2 is faster than the loop on \numLNLevelwiseFasterThanLoop, and
shrinking is faster than plain witnesses on \numLNShrinkFasterThanLoop. B3
is the fastest method on \numLNGreedyFastest{} rows. At these sizes, the
choice of method does not matter in practice, and
the loop is not the fastest.

The call bound is a property of the oracle's witnesses, and the counts
follow it. At $|U| \le \numLMaxU$, the replay runs dominate. They are what
the loop pays for computing $\attacks(\chi)$ as well as $\weakest(\chi)$,
and on lattices this small, they cost as much as the goal runs the loop
saves. The loop gains in calls over enumeration only at larger $|U|$,
which we measured only with the hypergraph oracle of
Section~\ref{sec:eval-scale}. In return, it collects the minimal attacks,
which show what each assumption is there to stop, along with the weakest
trust. B3 returns only one trust, so it cannot detect when the weakest
trust is not unique, as in the two-key variant.

\begin{table}[tbp]
\caption{Distinct oracle calls of Algorithm~\ref{alg:loop} per witness mode, compared with $N = |\attacks(\chi)| + |\weakest(\chi)|$, the bound of Theorem~\ref{thm:cost} for witnesses with minimal stopping sets. The as-found witness is the first attack the oracle reports (for CPSA, the failing shape with the smallest clause). The pruned witness is the analyser's run pruned to a minimal run. The least witness is a run whose stopping set is least among all attacks the analyser reaches. Shrink uses as-found witnesses minimalised by further calls, which are counted. CPSA runs are goal runs $+$ replay runs that compute $\stops(x)$. B1 evaluates all $2^{|U|}$ trusts, B2 enumerates trusts by size and skips supersets of sufficient ones, and B3 is the conjunct-dropping method of~\cite{rowe}, which returns one minimal sufficient trust. B2 and B3 were run on the CPSA rows only.}
\label{tab:calls}
\centering
\setlength{\tabcolsep}{3pt}
\footnotesize
\begin{tabular}{l c r r r r r r r r r r}
\toprule
protocol & & $N$ & \multicolumn{4}{c}{calls of Algorithm~\ref{alg:loop}} & \multicolumn{2}{c}{CPSA runs} & B1 & B2 & B3\\
\cmidrule(lr){4-7}\cmidrule(lr){8-9}
 & & & as found & pruned & least & shrink & as found & shrink & & &\\
\midrule
signed CR & A & 3 & 5 & 3 & 3 & 6 & -- & -- & 8 & -- & --\\
\quad two-key variant & A & 4 & 4 & 4 & 4 & 5 & -- & -- & 8 & -- & --\\
Needham--Schroeder & A & 1 & 2 & 2 & 1 & 2 & -- & -- & 4 & -- & --\\
\quad with channels & A & 3 & 4 & 4 & 3 & 6 & -- & -- & 64 & -- & --\\
\quad Lowe's fix & A & 2 & 2 & 2 & 2 & 3 & -- & -- & 4 & -- & --\\
\textsc{pkinit} (flawed) & C & 1 & 2 & -- & -- & 2 & 2{+}3 & 2{+}2 & 4 & 4 & 1\\
\quad adopted fix & C & 2 & 2 & -- & -- & 3 & 2{+}2 & 3{+}2 & 4 & 3 & 3\\
ISO 9798-3 (flawed) & C & 1 & 2 & -- & -- & 2 & 2{+}7 & 2{+}4 & 16 & 16 & 1\\
\quad corrected & C & 2 & 2 & -- & -- & 3 & 2{+}4 & 3{+}4 & 16 & 9 & 5\\
Blanchet & C & 1 & 4 & -- & -- & 3 & 4{+}12 & 3{+}5 & 16 & 16 & 1\\
\quad fixed & C & 5 & 6 & -- & -- & 6 & 6{+}13 & 6{+}11 & 16 & 16 & 5\\
Otway--Rees & C & 1 & 2 & -- & -- & 2 & 2{+}37 & 2{+}4 & 16 & 16 & 1\\
\quad with $A \neq B$ & C & 4 & 4 & -- & -- & 5 & 4{+}47 & 5{+}47 & 16 & 15 & 5\\
Yahalom & C & 3 & 3 & -- & -- & 4 & 3{+}10 & 4{+}10 & 16 & 13 & 5\\
Denning--Sacco & C & 2 & 3 & -- & -- & 4 & 3{+}37 & 4{+}11 & 32 & 17 & 6\\
Neuman--Stubblebine & C & 3 & 3 & -- & -- & 4 & 3{+}13 & 4{+}13 & 32 & 25 & 6\\
Kerberos (simplified) & C & 1 & 4 & -- & -- & 3 & 4{+}20 & 3{+}6 & 16 & 16 & 1\\
\quad agreement on $K$ & C & 3 & 4 & -- & -- & 5 & 4{+}16 & 5{+}12 & 16 & 13 & 5\\
\bottomrule
\end{tabular}
\end{table}

\begin{table}[tbp]
\caption{Wall-clock time in seconds per CPSA row and method, as the median of 5 repetitions after one warm-up, with the interquartile range in parentheses. Methods are as in Table~\ref{tab:calls}. B3 returns one trust, not $\weakest(\chi)$. The last column gives the range over the methods of the median share of the wall time spent in CPSA.}
\label{tab:timing}
\centering
\setlength{\tabcolsep}{3pt}
\footnotesize
\begin{tabular}{l r@{\,}l r@{\,}l r@{\,}l r@{\,}l r@{\,}l r}
\toprule
protocol & \multicolumn{2}{c}{Alg.~1} & \multicolumn{2}{c}{shrink} & \multicolumn{2}{c}{B1} & \multicolumn{2}{c}{B2} & \multicolumn{2}{c}{B3} & CPSA \%\\
\midrule
\textsc{pkinit} (flawed) & 0.091 & (0.005) & 0.072 & (0.003) & 0.073 & (0.000) & 0.071 & (0.001) & 0.017 & (0.001) & 97--98\\
\quad adopted fix & 0.073 & (0.001) & 0.091 & (0.004) & 0.074 & (0.001) & 0.056 & (0.001) & 0.056 & (0.002) & 97--98\\
ISO 9798-3 (flawed) & 0.161 & (0.001) & 0.108 & (0.006) & 0.285 & (0.003) & 0.292 & (0.004) & 0.019 & (0.001) & 97--98\\
\quad corrected & 0.108 & (0.001) & 0.127 & (0.003) & 0.281 & (0.013) & 0.157 & (0.007) & 0.090 & (0.003) & 97--98\\
Blanchet & 0.288 & (0.001) & 0.139 & (0.002) & 0.284 & (0.006) & 0.286 & (0.001) & 0.018 & (0.000) & 97--98\\
\quad fixed & 0.344 & (0.001) & 0.311 & (0.015) & 0.287 & (0.004) & 0.291 & (0.007) & 0.091 & (0.005) & 97--98\\
Otway--Rees & 1.266 & (0.026) & 0.255 & (0.004) & 0.842 & (0.015) & 0.846 & (0.023) & 0.057 & (0.001) & 98--99\\
\quad with $A \neq B$ & 1.419 & (0.034) & 1.423 & (0.015) & 0.652 & (0.019) & 0.623 & (0.012) & 0.193 & (0.005) & 98\\
Yahalom & 0.330 & (0.004) & 0.466 & (0.013) & 0.686 & (0.006) & 0.282 & (0.009) & 0.448 & (0.002) & 97--98\\
Denning--Sacco & 0.911 & (0.042) & 1.303 & (0.002) & 0.860 & (0.010) & 0.361 & (0.011) & 0.173 & (0.015) & 96--99\\
Neuman--Stubblebine & 0.364 & (0.021) & 0.393 & (0.005) & 0.741 & (0.013) & 0.454 & (0.012) & 0.169 & (0.003) & 97--98\\
Kerberos (simplified) & 0.940 & (0.012) & 0.791 & (0.017) & 1.138 & (0.023) & 1.137 & (0.041) & 0.322 & (0.002) & 98--99\\
\quad agreement on $K$ & 0.734 & (0.011) & 1.185 & (0.003) & 1.126 & (0.008) & 0.366 & (0.009) & 0.665 & (0.005) & 98--99\\
\bottomrule
\end{tabular}
\end{table}

\subsection{Scaling}
\label{sec:eval-scale}

This study checks the implementation against the cost bound of
Theorem~\ref{thm:cost} using an oracle on families of sets, and makes no
verifier calls on protocols. In all \numLNAlgRuns{} completed runs with
minimal witnesses, up to $n = \numLMaxN$ candidate assumptions, the loop
made exactly $|\attacks(\chi)| + |\weakest(\chi)|$ calls. Its running time
grows with the numbers of minimal attacks and weakest trusts rather than
with $n$.

Every protocol row has at most \numLMaxU{} candidate assumptions. To run
the loop with more, we used the reduction in the proof of
Theorem~\ref{thm:cost} (Fig.~\ref{fig:gadget}) to generate instances. The
reduction turns a family of sets over $n$ candidate assumptions into a
protocol whose minimal attacks are exactly those sets. Thus, the oracle
can be replaced by the check ``does $T$ meet every set, and if not, which
set does it miss''. This hypergraph oracle is exact and monotone by
construction. For small $n$, an exhaustive bounded search over all trusts
confirms this replacement on the gadget protocol itself. On the
\numLNGadget{} gadget instances, with $n \le \numLGadgetMaxN$ and
\numLNGadgetTrusts{} trusts in all, the gadget protocol and the check on
sets agreed on every trust. The loop over the gadget oracle returned the
blocker of the sets and recovered the sets as $\attacks$ (\numLNGadgetAgree{} of \numLNGadget{}
instances, Table~\ref{tab:gadget}).

The families are random families with $m \in \{n, 2n, 4n\}$ sets of size
one to three, with non-minimal sets removed (\numLNSeeds{} seeds per
point, $n = 5, 10, \dots, \numLMaxN$), disjoint pairs (with $|\weakest| =
2^{n/2}$), the dual of disjoint pairs (with $|\attacks| = 2^{n/2}$), and
the threshold families of all subsets of size $n/2$ and of size $n/2+1$.
Each instance was run through Algorithm~\ref{alg:loop} with minimal
witnesses, with witnesses padded by two atoms outside the trust, and with
padded witnesses shrunk. It was also run through the baselines B1 (for $n
\le 24$), B2 (stopped after \numLBtwoBudget{} calls), and B3. A series (family and method) stopped at the smallest $n$ with
a timeout.

Tables~\ref{tab:scale} and~\ref{tab:scale-struct} give one row per family
and $n$. Of \numLNScaleRuns{} runs, \numLNScaleTimeouts{} hit the time limit,
and the tables record all of them. In all \numLNAlgRuns{} completed runs of
Algorithm~\ref{alg:loop} with minimal witnesses, the distinct calls equal
$|\attacks(\chi)| + |\weakest(\chi)|$ (\numLNAlgAtBound{} runs), up to $n =
\numLMaxN$, $|\weakest| = \numLMaxW$, and $|\attacks| = \numLMaxH$. These
runs also recover $\attacks(\chi)$, which follows from the correct
$\weakest$ by the involution. Padded witnesses leave the loop correct, but
the calls exceed the bound in all \numLNPaddedRuns{} completed runs. On
every random family, the padded series times out from $n =
\numLPaddedTimeoutN$ (``t/o'' in Table~\ref{tab:scale}), where the series with minimal and with
shrunk witnesses still complete. With shrinking, all \numLNShrinkRuns{}
completed runs stay within $|\weakest| + |\attacks|(1+|U|)$. Wherever two
or more methods completed on an instance, they returned the same
$\weakest$ (\numLNDigestInstances{} instances, \numLNDigestMismatch{}
disagreements). B1 is out of reach beyond $n = 24$ by design, and B2
exhausted its budget in \numLNBtwoBudget{} runs. B3's single answer was in
$\weakest$ in all \numLNBthreeRuns{} runs, which is guaranteed for an
exact monotone oracle and so tests only the implementation of B3.

More of the running time is spent in the blocker computation than in the
oracle. Of the
\numLNScaleSlowRuns{} timed runs over one second, the blocker took longer
than the oracle in \numLNScaleBlockerDominates. Random instances at $n =
\numLMaxN$ with few minimal trusts finish in milliseconds, while the
threshold families, whose answers grow binomially, time out at small $n$
(Table~\ref{tab:scale-struct}). Seeds of the same random family spread
widely, because $|\weakest|$ varies by orders of magnitude between instances
of the same size. The longest timed run took \numLScaleMaxTime\,s.

\begin{table}[p]
\caption{Algorithm~\ref{alg:loop} on instances from the reduction in the proof of Theorem~\ref{thm:cost} (Fig.~\ref{fig:gadget}), random families, one row per $n$. Seeds counts the runs of Algorithm~\ref{alg:loop} completed within 600\,s over the runs started (a series stops at the first $n$ with a timeout, and empty cells were not run). $|\attacks|$, $|\weakest|$ (maximum in parentheses when there are several seeds), and calls are medians over completed seeds. Calls are distinct oracle calls with minimal witnesses. Time is the median wall-clock time of Algorithm~\ref{alg:loop} for seed 0 over its repetitions, with their interquartile range. $^{*}$ marks rows timed on every seed, where time and IQR are over the per-seed medians. Padded witnesses are padded with two atoms outside the trust, and shrink minimalises padded witnesses. $^{\dagger}$ marks a median over the seeds that completed. B1, B2, and B3 are as in Table~\ref{tab:calls}. B2 stops after $2^{21}$ calls, and B3 counts the seeds whose single answer is in $\weakest$. Each random instance draws $m$ edges of size one to three over $n$ candidate assumptions and minimalises them.}
\label{tab:scale}
\centering
\setlength{\tabcolsep}{2.6pt}
\scriptsize
\begin{tabular}{r r r r r r r r r r r r r}
\toprule
$n$ & seeds & $|\attacks|$ & $|\weakest|$ & calls & time (s) & IQR & padded & shrink & B1 & B2 & \multicolumn{2}{c}{B3}\\
\cmidrule(lr){12-13}
 & & & & & & & & & & & calls & in $\weakest$\\
\midrule
\multicolumn{13}{l}{random, $m=n$}\\
5 & 10/10 & 3 & 2 (3) & 5 & $2.6{\cdot}10^{-5}$ & $4.0{\cdot}10^{-6}$ & 19 & 12 & 32 & 24 & 6 & 10/10\\
10 & 10/10 & 6 & 4 (11) & 9 & $1.9{\cdot}10^{-5}$ & $4.0{\cdot}10^{-6}$ & 146 & 28 & 1\,024 & 870 & 11 & 10/10\\
15 & 10/10 & 8 & 6 (20) & 13 & $5.7{\cdot}10^{-5}$ & $3.0{\cdot}10^{-6}$ & 500 & 39 & 32\,768 & 31\,880 & 16 & 10/10\\
20 & 10/10 & 11 & 18 (40) & 28 & $1.9{\cdot}10^{-4}$$^{*}$ & $2.8{\cdot}10^{-4}$ & 1\,246 & 64 & 1\,048\,576 & 1\,041\,651 & 21 & 10/10\\
25 & 10/10 & 15 & 64 (105) & 79 & $6.0{\cdot}10^{-4}$ & $2.5{\cdot}10^{-5}$ & t/o & 134 &  & $>2^{21}$ & 26 & 10/10\\
30 & 10/10 & 18 & 91 (1\,089) & 112 & 0.233 & $7.7{\cdot}10^{-4}$ &  & 187 &  & $>2^{21}$ & 31 & 10/10\\
35 & 10/10 & 20 & 56 (296) & 76 & 0.014 & $3.2{\cdot}10^{-4}$ &  & 153 &  & $>2^{21}$ & 36 & 10/10\\
40 & 10/10 & 24 & 81 (415) & 103 & 0.0085 & $1.2{\cdot}10^{-4}$ &  & 178 &  & $>2^{21}$ & 41 & 10/10\\
45 & 10/10 & 26 & 486 (6\,930) & 510 & 0.064 & $5.6{\cdot}10^{-4}$ &  & 596 &  & $>2^{21}$ & 46 & 10/10\\
50 & 10/10 & 27 & 796 (2\,928) & 824 & 0.056 & 0.0020 &  & 919 &  & $>2^{21}$ & 51 & 10/10\\
55 & 10/10 & 33 & 684 (29\,886) & 716 & 0.081$^{*}$ & 14.0 &  & 830 &  & $>2^{21}$ & 56 & 10/10\\
60 & 9/10 & 34 & 588 (7\,216) & 622 & 4.6 & 0.0059 &  & 739$^{\dagger}$ &  & $>2^{21}$ & 61 & 10/10\\
65 &  &  &  &  &  &  &  &  &  & $>2^{21}$ & 66 & 10/10\\
70 &  &  &  &  &  &  &  &  &  & $>2^{21}$ & 71 & 10/10\\
75 &  &  &  &  &  &  &  &  &  & $>2^{21}$ & 76 & 10/10\\
80 &  &  &  &  &  &  &  &  &  & $>2^{21}$ & 81 & 10/10\\
\midrule
\multicolumn{13}{l}{random, $m=2n$}\\
5 & 10/10 & 4 & 1 (3) & 5 & $1.7{\cdot}10^{-5}$ & $1.0{\cdot}10^{-6}$ & 23 & 13 & 32 & 29 & 6 & 10/10\\
10 & 10/10 & 7 & 2 (6) & 8 & $7.8{\cdot}10^{-5}$ & $3.0{\cdot}10^{-6}$ & 168 & 26 & 1\,024 & 1\,012 & 11 & 10/10\\
15 & 10/10 & 11 & 4 (10) & 15 & $2.5{\cdot}10^{-4}$ & $3.9{\cdot}10^{-5}$ & 554 & 49 & 32\,768 & 32\,702 & 16 & 10/10\\
20 & 10/10 & 14 & 4 (10) & 16 & $1.3{\cdot}10^{-4}$$^{*}$ & $7.8{\cdot}10^{-5}$ & 1\,305 & 52 & 1\,048\,576 & 1\,047\,940 & 21 & 10/10\\
25 & 10/10 & 18 & 10 (24) & 28 & $8.8{\cdot}10^{-4}$ & $4.5{\cdot}10^{-4}$ & t/o & 80 &  & $>2^{21}$ & 26 & 10/10\\
30 & 10/10 & 21 & 26 (92) & 48 & $4.7{\cdot}10^{-4}$ & $7.2{\cdot}10^{-5}$ &  & 114 &  & $>2^{21}$ & 31 & 10/10\\
35 & 10/10 & 26 & 30 (120) & 56 & $9.3{\cdot}10^{-4}$ & $3.1{\cdot}10^{-4}$ &  & 131 &  & $>2^{21}$ & 36 & 10/10\\
40 & 10/10 & 29 & 44 (480) & 72 & $6.9{\cdot}10^{-4}$ & $1.9{\cdot}10^{-4}$ &  & 164 &  & $>2^{21}$ & 41 & 10/10\\
45 & 10/10 & 34 & 120 (1\,852) & 156 & 0.0019 & $3.3{\cdot}10^{-4}$ &  & 266 &  & $>2^{21}$ & 46 & 10/10\\
50 & 10/10 & 38 & 150 (1\,652) & 190 & 0.0019 & $2.9{\cdot}10^{-4}$ &  & 316 &  & $>2^{21}$ & 51 & 10/10\\
55 & 10/10 & 40 & 210 (5\,438) & 252 & 0.0026 & $2.2{\cdot}10^{-4}$ &  & 375 &  & $>2^{21}$ & 56 & 10/10\\
60 & 10/10 & 43 & 334 (1\,525) & 380 & 0.024 & $1.6{\cdot}10^{-4}$ &  & 522 &  & $>2^{21}$ & 61 & 10/10\\
65 & 10/10 & 48 & 898 (4\,032) & 950 & 0.0055 & $1.3{\cdot}10^{-4}$ &  & 1\,108 &  & $>2^{21}$ & 66 & 10/10\\
70 & 10/10 & 55 & 2\,070 (10\,056) & 2\,125 & 1.2$^{*}$ & 23.2 &  & 2\,302 &  & $>2^{21}$ & 71 & 10/10\\
75 & 9/10 & 59 & 768 (23\,792) & 819 & 137 & 2.3 &  & 967$^{\dagger}$ &  & $>2^{21}$ & 76 & 10/10\\
80 &  &  &  &  &  &  &  &  &  & $>2^{21}$ & 81 & 10/10\\
\midrule
\multicolumn{13}{l}{random, $m=4n$}\\
5 & 10/10 & 4 & 1 (4) & 5 & $2.1{\cdot}10^{-5}$ & $1.3{\cdot}10^{-5}$ & 24 & 13 & 32 & 30 & 6 & 10/10\\
10 & 10/10 & 8 & 1 (3) & 10 & $4.5{\cdot}10^{-5}$ & $5.0{\cdot}10^{-6}$ & 173 & 28 & 1\,024 & 1\,021 & 11 & 10/10\\
15 & 10/10 & 12 & 2 (8) & 15 & $9.9{\cdot}10^{-5}$ & $4.0{\cdot}10^{-6}$ & 570 & 44 & 32\,768 & 32\,748 & 16 & 10/10\\
20 & 10/10 & 16 & 2 (21) & 18 & $2.1{\cdot}10^{-4}$ & $1.0{\cdot}10^{-5}$ & 1\,337 & 58 & 1\,048\,576 & 1\,048\,556 & 21 & 10/10\\
25 & 10/10 & 21 & 2 (6) & 23 & $3.4{\cdot}10^{-4}$ & $1.4{\cdot}10^{-4}$ & t/o & 71 &  & $>2^{21}$ & 26 & 10/10\\
30 & 10/10 & 26 & 4 (8) & 30 & $6.6{\cdot}10^{-4}$ & $2.6{\cdot}10^{-4}$ &  & 94 &  & $>2^{21}$ & 31 & 10/10\\
35 & 10/10 & 29 & 5 (26) & 34 & $9.2{\cdot}10^{-4}$ & $4.5{\cdot}10^{-4}$ &  & 107 &  & $>2^{21}$ & 36 & 10/10\\
40 & 10/10 & 34 & 8 (22) & 41 & $9.3{\cdot}10^{-4}$ & $3.1{\cdot}10^{-4}$ &  & 128 &  & $>2^{21}$ & 41 & 10/10\\
45 & 10/10 & 38 & 15 (361) & 51 & 0.0013 & $7.8{\cdot}10^{-5}$ &  & 141 &  & $>2^{21}$ & 46 & 10/10\\
50 & 10/10 & 42 & 6 (80) & 48 & 0.0017 & $9.7{\cdot}10^{-5}$ &  & 152 &  & $>2^{21}$ & 51 & 10/10\\
55 & 10/10 & 47 & 10 (36) & 57 & 0.0019 & $1.1{\cdot}10^{-5}$ &  & 171 &  & $>2^{21}$ & 56 & 10/10\\
60 & 10/10 & 52 & 8 (24) & 58 & 0.0022 & $2.0{\cdot}10^{-4}$ &  & 179 &  & $>2^{21}$ & 61 & 10/10\\
65 & 10/10 & 54 & 14 (157) & 68 & 0.0020 & $4.2{\cdot}10^{-5}$ &  & 206 &  & $>2^{21}$ & 66 & 10/10\\
70 & 10/10 & 59 & 22 (291) & 84 & 0.0029 & $4.0{\cdot}10^{-5}$ &  & 236 &  & $>2^{21}$ & 71 & 10/10\\
75 & 10/10 & 64 & 36 (1\,072) & 102 & 0.0023 & $3.9{\cdot}10^{-5}$ &  & 265 &  & $>2^{21}$ & 76 & 10/10\\
80 & 10/10 & 68 & 40 (375) & 108 & 0.0042$^{*}$ & 0.021 &  & 275 &  & $>2^{21}$ & 81 & 10/10\\
\bottomrule
\end{tabular}
\end{table}

\begin{table}[p]
\caption{Algorithm~\ref{alg:loop} on instances from the reduction in the proof of Theorem~\ref{thm:cost} (Fig.~\ref{fig:gadget}), structured families, one row per $n$. Seeds counts the runs of Algorithm~\ref{alg:loop} completed within 600\,s over the runs started (a series stops at the first $n$ with a timeout, and empty cells were not run). $|\attacks|$, $|\weakest|$ (maximum in parentheses when there are several seeds), and calls are medians over completed seeds. Calls are distinct oracle calls with minimal witnesses. Time is the median wall-clock time of Algorithm~\ref{alg:loop} for seed 0 over its repetitions, with their interquartile range. $^{*}$ marks rows timed on every seed, where time and IQR are over the per-seed medians. Padded witnesses are padded with two atoms outside the trust, and shrink minimalises padded witnesses. $^{\dagger}$ marks a median over the seeds that completed. B1, B2, and B3 are as in Table~\ref{tab:calls}. B2 stops after $2^{21}$ calls, and B3 counts the seeds whose single answer is in $\weakest$. Disjoint pairs have $n/2$ edges $\{2i, 2i{+}1\}$, so $|\weakest| = 2^{n/2}$. The dual of pairs is its blocker, so $|\attacks| = 2^{n/2}$. The threshold families consist of all subsets of size $n/2$, and of size $n/2+1$, of $n$ candidate assumptions. There is one instance per $n$.}
\label{tab:scale-struct}
\centering
\setlength{\tabcolsep}{2.6pt}
\scriptsize
\begin{tabular}{r r r r r r r r r r r r r}
\toprule
$n$ & seeds & $|\attacks|$ & $|\weakest|$ & calls & time (s) & IQR & padded & shrink & B1 & B2 & \multicolumn{2}{c}{B3}\\
\cmidrule(lr){12-13}
 & & & & & & & & & & & calls & in $\weakest$\\
\midrule
\multicolumn{13}{l}{disjoint pairs}\\
4 & 1/1 & 2 & 4 & 6 & $1.6{\cdot}10^{-5}$ & $3.0{\cdot}10^{-6}$ & 11 & 11 & 16 & 11 & 5 & 1/1\\
6 & 1/1 & 3 & 8 & 11 & $2.8{\cdot}10^{-5}$ & $2.0{\cdot}10^{-6}$ & 38 & 21 & 64 & 45 & 7 & 1/1\\
8 & 1/1 & 4 & 16 & 20 & $5.1{\cdot}10^{-5}$ & $7.0{\cdot}10^{-6}$ & 98 & 36 & 256 & 191 & 9 & 1/1\\
10 & 1/1 & 5 & 32 & 37 & $9.9{\cdot}10^{-5}$ & $9.0{\cdot}10^{-6}$ & 207 & 57 & 1\,024 & 813 & 11 & 1/1\\
12 & 1/1 & 6 & 64 & 70 & $2.4{\cdot}10^{-4}$ & $6.3{\cdot}10^{-5}$ & 385 & 94 & 4\,096 & 3\,431 & 13 & 1/1\\
14 & 1/1 & 7 & 128 & 135 & $6.6{\cdot}10^{-4}$ & $2.7{\cdot}10^{-4}$ & 660 & 163 & 16\,384 & 14\,325 & 15 & 1/1\\
16 & 1/1 & 8 & 256 & 264 & 0.0019 & $5.4{\cdot}10^{-4}$ & 1\,076 & 296 & 65\,536 & 59\,231 & 17 & 1/1\\
18 & 1/1 & 9 & 512 & 521 & 0.0029 & $3.1{\cdot}10^{-4}$ & 1\,709 & 557 & 262\,144 & 242\,973 & 19 & 1/1\\
20 & 1/1 & 10 & 1\,024 & 1\,034 & 0.0053 & $2.8{\cdot}10^{-4}$ & t/o & 1\,074 & 1\,048\,576 & 990\,551 & 21 & 1/1\\
22 & 1/1 & 11 & 2\,048 & 2\,059 & 0.012 & $1.8{\cdot}10^{-4}$ &  & 2\,103 & 4\,194\,304 & $>2^{21}$ & 23 & 1/1\\
24 & 1/1 & 12 & 4\,096 & 4\,108 & 0.026 & $1.8{\cdot}10^{-4}$ &  & 4\,156 & 16\,777\,216 & $>2^{21}$ & 25 & 1/1\\
\midrule
\multicolumn{13}{l}{dual of pairs}\\
4 & 1/1 & 4 & 2 & 6 & $4.4{\cdot}10^{-5}$ & $6.0{\cdot}10^{-6}$ & 11 & 13 & 16 & 11 & 5 & 1/1\\
6 & 1/1 & 8 & 3 & 11 & $8.0{\cdot}10^{-5}$ & $1.2{\cdot}10^{-5}$ & 29 & 33 & 64 & 30 & 7 & 1/1\\
8 & 1/1 & 16 & 4 & 20 & $3.6{\cdot}10^{-4}$ & $4.2{\cdot}10^{-5}$ & 76 & 80 & 256 & 85 & 9 & 1/1\\
10 & 1/1 & 32 & 5 & 37 & 0.0024 & $5.5{\cdot}10^{-4}$ & 197 & 190 & 1\,024 & 248 & 11 & 1/1\\
12 & 1/1 & 64 & 6 & 70 & 0.0088 & $9.8{\cdot}10^{-5}$ & 502 & 422 & 4\,096 & 735 & 13 & 1/1\\
14 & 1/1 & 128 & 7 & 135 & 0.047 & $4.3{\cdot}10^{-4}$ & 1\,255 & 979 & 16\,384 & 2\,194 & 15 & 1/1\\
16 & 1/1 & 256 & 8 & 264 & 0.286 & 0.0014 & 3\,080 & 2\,198 & 65\,536 & 6\,569 & 17 & 1/1\\
18 & 1/1 & 512 & 9 & 521 & 1.9 & 0.010 & t/o & 4\,864 & 262\,144 & 19\,692 & 19 & 1/1\\
20 & 1/1 & 1\,024 & 10 & 1\,034 & 14.1 & 0.050 &  & 10\,724 & 1\,048\,576 & 59\,059 & 21 & 1/1\\
22 & 1/1 & 2\,048 & 11 & 2\,059 & 109 & 0.376 &  & t/o & t/o & 177\,158 & 23 & 1/1\\
24 & 0/1 & \multicolumn{5}{c}{timeout} &  &  &  & 531\,453 & 25 & 1/1\\
\midrule
\multicolumn{13}{l}{all $\frac{n}{2}$-sets}\\
4 & 1/1 & 6 & 4 & 10 & $4.5{\cdot}10^{-5}$ & $1.1{\cdot}10^{-5}$ & 15 & 14 & 16 & 15 & 5 & 1/1\\
6 & 1/1 & 20 & 15 & 35 & $7.5{\cdot}10^{-4}$ & $2.1{\cdot}10^{-4}$ & 56 & 49 & 64 & 57 & 7 & 1/1\\
8 & 1/1 & 70 & 56 & 126 & 0.017 & $1.2{\cdot}10^{-4}$ & 210 & 176 & 256 & 219 & 9 & 1/1\\
10 & 1/1 & 252 & 210 & 462 & 0.760 & 0.0020 & 792 & 644 & 1\,024 & 848 & 11 & 1/1\\
12 & 1/1 & 924 & 792 & 1\,716 & 43.6 & 0.101 & 3\,003 & 2\,388 & 4\,096 & 3\,302 & 13 & 1/1\\
14 & 0/1 & \multicolumn{5}{c}{timeout} & t/o & t/o & 16\,384 & 12\,911 & 15 & 1/1\\
16 &  &  &  &  &  &  &  &  & 65\,536 & 50\,643 & 17 & 1/1\\
18 &  &  &  &  &  &  &  &  & 262\,144 & 199\,140 & 19 & 1/1\\
20 &  &  &  &  &  &  &  &  & t/o & t/o & 21 & 1/1\\
\midrule
\multicolumn{13}{l}{all $(\frac{n}{2}{+}1)$-sets}\\
4 & 1/1 & 4 & 6 & 10 & $3.1{\cdot}10^{-5}$ & $9.0{\cdot}10^{-6}$ & 11 & 11 & 16 & 11 & 5 & 1/1\\
6 & 1/1 & 15 & 20 & 35 & $4.1{\cdot}10^{-4}$ & $4.2{\cdot}10^{-5}$ & 42 & 41 & 64 & 42 & 7 & 1/1\\
8 & 1/1 & 56 & 70 & 126 & 0.011 & $3.4{\cdot}10^{-4}$ & 162 & 153 & 256 & 163 & 9 & 1/1\\
10 & 1/1 & 210 & 252 & 462 & 0.530 & 0.0013 & 627 & 574 & 1\,024 & 638 & 11 & 1/1\\
12 & 1/1 & 792 & 924 & 1\,716 & 31.2 & 0.122 & 2\,431 & 2\,166 & 4\,096 & 2\,510 & 13 & 1/1\\
14 & 0/1 & \multicolumn{5}{c}{timeout} & t/o & t/o & 16\,384 & 9\,908 & 15 & 1/1\\
16 &  &  &  &  &  &  &  &  & 65\,536 & 39\,203 & 17 & 1/1\\
18 &  &  &  &  &  &  &  &  & 262\,144 & 155\,382 & 19 & 1/1\\
20 &  &  &  &  &  &  &  &  & t/o & t/o & 21 & 1/1\\
\bottomrule
\end{tabular}
\end{table}

\begin{table}[tbp]
\caption{Gadget check. For each hypergraph (edges over candidate assumptions $0, \dots, n-1$), the protocol of the reduction in the proof of Theorem~\ref{thm:cost} is decided by exhaustive bounded Dolev--Yao search and compared with the hypergraph oracle (``does $T$ meet every edge''). Agree counts the trusts on which both give the same verdict and, on failure, the gadget's stopping set is an edge disjoint from $T$. Loop means that Algorithm~\ref{alg:loop} over the gadget oracle returns the blocker of the edges and recovers the edges as $\attacks$, and calls are its distinct calls. Time is in seconds for the whole check.}
\label{tab:gadget}
\centering
\setlength{\tabcolsep}{3pt}
\scriptsize
\begin{tabular}{l r l r r r c r r}
\toprule
instance & $n$ & edges & $|\attacks|$ & $|\weakest|$ & agree & loop & calls & time (s)\\
\midrule
fixed-1 & 4 & $\{3\}, \{0,1\}, \{1,2\}$ & 3 & 2 & 16/16 & yes & 5 & 0.060\\
fixed-2 & 4 & $\{0\}, \{1,2,3\}$ & 2 & 3 & 16/16 & yes & 5 & 0.0060\\
fixed-3 & 4 & $\{0,1\}, \{2,3\}$ & 2 & 4 & 16/16 & yes & 6 & 0.0060\\
fixed-4 & 5 & $\{0,4\}, \{2,3\}, \{0,1,2\}$ & 3 & 4 & 32/32 & yes & 7 & 0.204\\
fixed-5 & 5 & $\{0,3\}, \{2,3\}, \{3,4\}, \{0,2,4\}$ & 4 & 4 & 32/32 & yes & 8 & 5.3\\
rand-n4-s0 & 4 & $\{2\}, \{1,3\}$ & 2 & 2 & 16/16 & yes & 4 & 0.0050\\
rand-n4-s1 & 4 & $\{1,2\}, \{0,2,3\}$ & 2 & 3 & 16/16 & yes & 5 & 0.0070\\
rand-n4-s2 & 4 & $\{3\}$ & 1 & 1 & 16/16 & yes & 2 & 0.0010\\
rand-n4-s3 & 4 & $\{0\}, \{1,2\}$ & 2 & 2 & 16/16 & yes & 4 & 0.0040\\
rand-n4-s4 & 4 & $\{2\}, \{0,1\}$ & 2 & 2 & 16/16 & yes & 4 & 0.0040\\
rand-n4-s5 & 4 & $\{1,3\}, \{0,2,3\}$ & 2 & 3 & 16/16 & yes & 5 & 0.0070\\
rand-n4-s6 & 4 & $\{1\}, \{2\}$ & 2 & 1 & 16/16 & yes & 3 & 0.0030\\
rand-n4-s7 & 4 & $\{0,3\}, \{0,1,2\}$ & 2 & 3 & 16/16 & yes & 5 & 0.0070\\
rand-n5-s0 & 5 & $\{0\}, \{1\}, \{2,4\}$ & 3 & 2 & 32/32 & yes & 5 & 0.069\\
rand-n5-s1 & 5 & $\{0\}, \{1\}, \{2,3,4\}$ & 3 & 3 & 32/32 & yes & 6 & 0.090\\
rand-n5-s2 & 5 & $\{1,3\}, \{0,1,2\}, \{1,2,4\}$ & 3 & 3 & 32/32 & yes & 6 & 0.246\\
rand-n5-s3 & 5 & $\{2\}, \{0,4\}, \{1,3,4\}$ & 3 & 3 & 32/32 & yes & 6 & 0.131\\
rand-n5-s4 & 5 & $\{0\}, \{4\}, \{2,3\}$ & 3 & 2 & 32/32 & yes & 5 & 0.069\\
rand-n5-s5 & 5 & $\{0,3\}, \{0,1,2\}, \{1,2,4\}$ & 3 & 5 & 32/32 & yes & 8 & 0.260\\
rand-n5-s6 & 5 & $\{0,2\}, \{3,4\}, \{1,2,4\}$ & 3 & 4 & 32/32 & yes & 7 & 0.201\\
rand-n5-s7 & 5 & $\{0\}, \{1,3\}, \{2,3,4\}$ & 3 & 3 & 32/32 & yes & 6 & 0.132\\
rand-n6-s0 & 6 & $\{1\}, \{3,4\}$ & 2 & 2 & 64/64 & yes & 4 & 0.015\\
rand-n6-s1 & 6 & $\{0\}, \{1\}, \{3\}, \{5\}$ & 4 & 1 & 64/64 & yes & 5 & 1.1\\
rand-n6-s2 & 6 & $\{3\}, \{1,2,5\}$ & 2 & 3 & 64/64 & yes & 5 & 0.018\\
rand-n6-s3 & 6 & $\{3\}, \{4\}$ & 2 & 1 & 64/64 & yes & 3 & 0.011\\
rand-n6-s4 & 6 & $\{0\}, \{2\}, \{3\}$ & 3 & 1 & 64/64 & yes & 4 & 0.083\\
rand-n6-s5 & 6 & $\{1,5\}, \{3,5\}, \{0,2,3\}$ & 3 & 4 & 64/64 & yes & 7 & 0.390\\
rand-n6-s6 & 6 & $\{0\}, \{1\}$ & 2 & 1 & 64/64 & yes & 3 & 0.012\\
rand-n6-s7 & 6 & $\{4\}$ & 1 & 1 & 64/64 & yes & 2 & 0.0030\\
stress-n6-m6 & 6 & $\{0,1\}, \{0,5\}, \{1,4\}, \{0,3,4\}, \{2,3,5\}$ & 5 & 8 & 64/64 & yes & 13 & 329\\
\bottomrule
\end{tabular}
\end{table}

\subsection{Sensitivity to the Bound}
\label{sec:eval-bound}

Every attack the loop reports is a run of the protocol model, and its
stopping set is computed by replaying that run, so neither depends on
the bound. Therefore, every trust that suffices without a bound meets each
stopping set the loop collected. When the reported weakest trust is
unique, every minimal attack found has a single stopping assumption, and
each of its members is needed regardless of the bound. When
$|\weakest(\chi)| > 1$, a member of one reported weakest trust need not be
needed. What holds is only that every sufficient trust meets each
collected stopping set. Sufficiency, however, holds only as far as the
oracle is complete. Enlarging our analyser's pool changed only answers
that involve channel assumptions. The rows over keys and freshness kept
their answers, and every configuration with channel assumptions changed.

We reran each analyser row with one, two, and three instances per role,
with and without the earlier completed session, using pruned witnesses.
Table~\ref{tab:bound} gives the answers, the states explored, and the
time. With three instances, every analyser row except Needham--Schroeder
exceeds the time limit, and the $\mathsf{false}$ answer for
Needham--Schroeder is unchanged (\numLNBoundTimeouts{} of
\numLNBoundRuns{} runs time out). The state space grows from at most
\numLBoundMaxStatesOne{} states at one instance to \numLBoundMaxStatesTwo{}
at two. Without the earlier session, the challenge-response rows lose
$\Unq(N)$, since the replay attack needs an earlier session to replay
from. The pool must contain the
attacks the goal is about. Attacks outside the analyser's model, such as
type-flaw attacks, are not found at any pool size.

The CPSA rows are theorems about the CPSA model of each protocol, for our
goal, $U$, and $\sigma$, as far as CPSA is complete on the fragment. For
instance, the Denning--Sacco answer reflects CPSA's \texttt{pubk}.
The analyser's $\mathsf{false}$ answers come
from a run with an empty stopping set and hold for its model without a
bound. The other analyser rows give, for each collected stopping set, a
constraint that every sufficient trust meets, as above. They also give
observations, relative to the pool, about which trusts suffice. The
channel row changes at two instances, so its answer at one instance is not
the answer for the protocol model without a bound.

\begin{table}[tbp]
\caption{Sensitivity of $\weakest(\chi)$ to the analyser's pool: instances per role, with or without one earlier completed session whose transcript the adversary holds (pruned witnesses, reduced search). ``same'' means equal to the one-instance answer with the same session setting. States are summed over all calls of one run of Algorithm~\ref{alg:loop}. Time is the wall-clock time of that run in seconds, with a limit of 600\,s.}
\label{tab:bound}
\centering
\setlength{\tabcolsep}{3pt}
\footnotesize
\begin{tabular}{l c c l r r r}
\toprule
row & earlier & inst. & $\weakest(\chi)$ & calls & states & time (s)\\
\midrule
signed CR & yes & 1 & $\{\Non(sk_A), \Unq(N)\}$ & 3 & 3\,202 & 0.011\\
 &  & 2 & same & 3 & 3\,309\,263 & 35.3\\
 &  & 3 & \multicolumn{3}{l}{timeout} & 600\\
\addlinespace[2pt]
signed CR & no & 1 & $\{\Non(sk_A)\}$ & 2 & 976 & 0.0038\\
 &  & 2 & same & 2 & 799\,929 & 7.6\\
 &  & 3 & \multicolumn{3}{l}{timeout} & 600\\
\addlinespace[2pt]
two-key variant & yes & 1 & $\{\Non(k_1), \Unq(N)\}$, & 4 & 3\,247 & 0.014\\
 & & & $\{\Non(k_2), \Unq(N)\}$ & & & \\
 &  & 2 & same & 4 & 2\,327\,293 & 29.7\\
 &  & 3 & \multicolumn{3}{l}{timeout} & 600\\
\addlinespace[2pt]
two-key variant & no & 1 & $\{\Non(k_1)\}$, & 3 & 1\,324 & 0.0051\\
 & & & $\{\Non(k_2)\}$ & & & \\
 &  & 2 & same & 3 & 759\,831 & 8.7\\
 &  & 3 & \multicolumn{3}{l}{timeout} & 600\\
\addlinespace[2pt]
Needham--Schroeder & yes & 1 & $\mathsf{false}$ & 2 & 412 & 0.0010\\
 &  & 2 & same & 2 & 863 & 0.0030\\
 &  & 3 & same & 2 & 1\,546 & 0.0069\\
\addlinespace[2pt]
Needham--Schroeder & no & 1 & $\mathsf{false}$ & 2 & 149 & 0.0005\\
 &  & 2 & same & 2 & 414 & 0.0016\\
 &  & 3 & same & 2 & 863 & 0.0039\\
\addlinespace[2pt]
N--S with channels & yes & 1 & $\{\auth(A{\to}B)\}$, & 4 & 1\,226 & 0.0054\\
 & & & $\{\conf(B{\to}A)\}$ & & & \\
 &  & 2 & $\{\conf(B{\to}A)\}$, & 7 & 1\,570\,351 & 14.0\\
 & & & $\{\Non(sk_A), \auth(A{\to}B)\}$, & & & \\
 & & & $\{\auth(A{\to}B), \auth(B{\to}A)\}$ & & & \\
 &  & 3 & \multicolumn{3}{l}{timeout} & 600\\
\addlinespace[2pt]
N--S with channels & no & 1 & $\{\auth(A{\to}B)\}$, & 4 & 545 & 0.0024\\
 & & & $\{\conf(B{\to}A)\}$ & & & \\
 &  & 2 & $\{\conf(B{\to}A)\}$, & 7 & 482\,837 & 3.5\\
 & & & $\{\Non(sk_A), \auth(A{\to}B)\}$, & & & \\
 & & & $\{\auth(A{\to}B), \auth(B{\to}A)\}$ & & & \\
 &  & 3 & \multicolumn{3}{l}{timeout} & 600\\
\addlinespace[2pt]
Lowe's fix & yes & 1 & $\{\Non(sk_A)\}$ & 2 & 639 & 0.0020\\
 &  & 2 & same & 2 & 1\,406\,661 & 10.9\\
 &  & 3 & \multicolumn{3}{l}{timeout} & 600\\
\addlinespace[2pt]
Lowe's fix & no & 1 & $\{\Non(sk_A)\}$ & 2 & 289 & 0.0009\\
 &  & 2 & same & 2 & 404\,625 & 2.5\\
 &  & 3 & \multicolumn{3}{l}{timeout} & 600\\
\bottomrule
\end{tabular}
\end{table}

\subsection{Channels and Composition}
\label{sec:eval-channels}
\label{sec:eval-composition}

Under the inject/read reading of channels, and in a pool of one or two
instances per role, channel assumptions turn Needham--Schroeder's
$\mathsf{false}$ into trusts that name channels. Which channels appear
depends on both the pool and the channel reading, so these answers are
observations about the bounded model. To the
$U$ of Needham--Schroeder and Lowe's fix, we added the atoms $\auth(c)$ and
$\conf(c)$ for the two channels between the honest principals. We
interpret them under both the inject/read reading and the hijack reading
of Section~\ref{sec:channels}, and ran each configuration with one and two
instances per role. Table~\ref{tab:channels} gives the answers. At one
instance, all \numLNChannelAgree{} configurations agree with evaluating
every trust using the same oracle, and \numLNChannelChanged{} of the \numLNChannelRows{}
answers change at two instances.

On Needham--Schroeder with one instance per role, Lowe's run is the only
attack. Its stopping set is $\{\auth(A{\to}B), \conf(B{\to}A)\}$, so either
channel assumption suffices. Two instances per role bring further attacks,
and the authentic channel alone no longer suffices. It also requires
$\Non(sk_A)$ or $\auth(B{\to}A)$. Under the hijack reading, the
confidential branch disappears. With one instance, the unique answer is
then $\{\auth(A{\to}B)\}$, and with two, $\auth(A{\to}B)$ is in every
weakest trust but no longer suffices alone. On Lowe's fix, the same atoms
give $\{\Non(sk_A)\}$, $\{\auth(A{\to}B)\}$, or $\{\conf(B{\to}A)\}$ with
one instance. With two, the authentic branch again requires
$\auth(B{\to}A)$ as well. Hijacking changes nothing there. Thus, channel
assumptions enter the trust lattice as candidate assumptions
like any other, and in this bounded model they turn an unrealisable goal
into a realisable one.

Trust composes as Theorem~\ref{thm:composition} predicts. The three pairs of
Section~\ref{sec:consequences} compose the signed challenge-response with an
attestation service that signs with the same key, with or without tags on
the two signed formats, or with a separate key. For each pair, we ran
Algorithm~\ref{alg:loop} on each component with its own goal, and on the
composition with each goal and with their conjunction.
Table~\ref{tab:composition} gives the answers, and every one agrees with
evaluating every trust using the same oracle. We call a minimal attack of
the composition \emph{cross-protocol} when its stopping set contains no
stopping set of a component's minimal attack. With tags, and with separate
keys, no minimal attack is cross-protocol, and the composition's weakest
trust is the join of the components' answers. Without tags, the condition
of Theorem~\ref{thm:composition} fails. The composition's answer is then
$\mathsf{false}$, strictly above the join, and two of its minimal attacks
are cross-protocol, with stopping sets $\{\Non(sk_C)\}$ for $\chi_1$ and
$\emptyset$ for $\chi_2$. Both protocols run unchanged in these attacks, and
the adversary passes a signature from one off as a message of the other.
These are therefore multi-protocol attacks, not attacks with a protocol the
adversary designs. In all three pairs, the weakest trust for the
conjunction of the goals equals the join of the weakest trusts for each, as
Corollary~\ref{cor:conjunction} states. The gap between the join and the
composition's answer is the cross-protocol attacks, and the loop finds them
without being told to look for them.

\begin{table}[tbp]
\caption{Channel assumptions on Needham--Schroeder and Lowe's fix. $U$ contains the two long-term keys and, in the channel rows, $\auth(c)$ and $\conf(c)$ for $c \in \{A{\to}B, B{\to}A\}$. ``hijack'' adds the transition that re-ascribes a message to another channel without learning it. Results are given for one instance per role (least witnesses, with calls, $|\attacks|$, and agreement with evaluating every trust using the same oracle, as in Table~\ref{tab:results}), and for two instances per role, with the states explored by the reduced search.}
\label{tab:channels}
\centering
\setlength{\tabcolsep}{3pt}
\footnotesize
\begin{tabular}{l l c l r r l r}
\toprule
protocol & $U$ & hijack & $\weakest$, 1 instance & $|\attacks|$ & calls & $\weakest$, 2 instances & states\\
\midrule
Needham--Schroeder & keys & no & $\mathsf{false}$ & 1 & 2 & $\mathsf{false}$ & 414\\
 & keys, channels & no & $\{\auth(A{\to}B)\}$, & 1 & 3 & $\{\conf(B{\to}A)\}$, & 482\,837\\
 & & & $\{\conf(B{\to}A)\}$ & & & $\{\Non(sk_A), \auth(A{\to}B)\}$, & \\
 & & &  & & & $\{\auth(A{\to}B), \auth(B{\to}A)\}$ & \\
 & keys, channels & yes & $\{\auth(A{\to}B)\}$ & 1 & 2 & $\{\Non(sk_A), \auth(A{\to}B)\}$, & 147\,942\\
 & & &  & & & $\{\auth(A{\to}B), \auth(B{\to}A)\}$, & \\
 & & &  & & & $\{\auth(A{\to}B), \conf(B{\to}A)\}$ & \\
\midrule
Lowe's fix & keys & no & $\{\Non(sk_A)\}$ & 1 & 2 & $\{\Non(sk_A)\}$ & 404\,625\\
 & keys, channels & no & $\{\Non(sk_A)\}$, & 1 & 4 & $\{\Non(sk_A)\}$, & 817\,327\\
 & & & $\{\auth(A{\to}B)\}$, & & & $\{\conf(B{\to}A)\}$, & \\
 & & & $\{\conf(B{\to}A)\}$ & & & $\{\auth(A{\to}B), \auth(B{\to}A)\}$ & \\
 & keys, channels & yes & $\{\Non(sk_A)\}$, & 1 & 4 & $\{\Non(sk_A)\}$, & 817\,327\\
 & & & $\{\auth(A{\to}B)\}$, & & & $\{\conf(B{\to}A)\}$, & \\
 & & & $\{\conf(B{\to}A)\}$ & & & $\{\auth(A{\to}B), \auth(B{\to}A)\}$ & \\
\bottomrule
\end{tabular}
\end{table}

\subsection{Threats to Validity}
\label{sec:eval-threats}

\paragraph{Bounded pools.} Section~\ref{sec:eval-bound} states what the
analyser's answers support. Within the time limit, the pool could be
enlarged only to two instances per role (Section~\ref{sec:lim-bound}),
since three instances timed out on every row except Needham--Schroeder. The channel
rows change at two instances, and nothing is known about three. The
analyser's answers about freshness depend on the earlier session in the
pool, and type-flaw attacks are outside its model.

\paragraph{Row selection and goals written by us.} The CPSA rows follow a
rule fixed before any run, with a cap of \numLRowCap{} rows in total and a
limit of \numLSelectTimeout\,s per lattice point during selection.
Protocols that need longer at some point would have been excluded. Every
CPSA row with answer $\mathsf{false}$ is a row of the rule. Otway--Rees
with $A \neq B$ and Kerberos with agreement on the session key were added
after the first results as second readings of the two rows whose goal
fails at every trust for a reason in the model or the goal, and we report
both readings. No shipped \texttt{defgoal} was used, so all
\numLNCpsaRows{} CPSA goals are ours. Appendix~\ref{app:protocols} gives
every goal verbatim so that a reader can check what was asked. The
candidate assumptions $U$ and the situation $\sigma$ follow the rule too,
and a different choice gives a different lattice and possibly different
answers. In particular, freshness that a role declares stays in $\sigma$
and is never reported as needed.

\paragraph{CPSA's fragment.} A CPSA verdict is a theorem only as far as CPSA
is complete for the goal language on these protocols and its default bounds
are not reached. Every run ended with its search exhausted. The replay
check depends on how CPSA reports inputs it cannot treat as skeletons,
which we established by reading its output rather than from a
specification. The reading
of the preskeleton and dead marks is tested on
saved outputs and checked against evaluating every trust on every row.
The CPSA rows concern the CPSA model of each protocol, with our goal, $U$,
and $\sigma$.

\paragraph{Timing.} All times come from one machine. The load average was
sampled only at the start and end of every step, so it says nothing about
the load during a step. Over the timed steps of the main session these
samples ranged from \numLLoadMin{} to \numLLoadMax{}. For the CPSA step,
rerun alone later, they were
\numLCpsaLoadMin{} and \numLCpsaLoadMax{}. Other processes were not controlled. Timing runs were
sequential and repeated. The call counts, on which the comparisons of
Sections~\ref{sec:eval-rq2} and~\ref{sec:eval-scale} rest, do not depend on
load. The scale times are wall-clock times of an OCaml process without
flambda optimisation.

\paragraph{Synthetic instances.} The scale families are instances of the
reduction, whose components do not interact by construction, and their
oracle is a check on families of sets, exact and monotone by construction.
They check the implementation against Theorem~\ref{thm:cost} and measure
the blocker computation and the number of oracle calls. They say nothing
about the cost of one oracle call on a large protocol, which is where an
analyser spends its time. The gadget check confirms the reduction only on
instances with $n \le \numLGadgetMaxN$, because exhaustive search over the
gadget is exponential in the number of sets.

\paragraph{Implementation errors.} Agreement with evaluating every trust
checks the loop against the same oracle, and the differential tests check
the OCaml code against a prototype written by the same authors. Both would
miss an error shared by the two implementations or present in the oracle.
The agreement also cannot detect stopping sets that are too large. An atom
wrongly kept in a clause costs calls but leaves the answer unchanged
(Section~\ref{sec:backend}). Stopping sets that are too small do change
the answer, which is how the replay error on Otway--Rees with $A \neq B$
was found. The only external check is that
the answers on the CPSA rows agree with the attacks and fixes known from
the literature.

\section{Scope and Limitations}
\label{sec:limits}

The results rest on the six modelling assumptions of
Section~\ref{sec:lattice} and on the oracles of Section~\ref{sec:eval}.
Every result concerns the symbolic model, and none bears on
computational security.

\subsection{Finitely Many Candidates and Trust Hypotheses Only (A1, A2)}
\label{sec:lim-a2}

We assume that the part of the hypothesis that varies is the trust, a
subset of a finite set $U$ of atoms whose values the situation fixes.
Thus, the answer concerns a bounded scenario rather than a protocol
quantified over unboundedly many sessions. Rowe, Guttman, and Liskov ask
their question for every conjunction of atoms of the goal language,
including role-position, parameter, and ordering atoms \cite{rowe}. We
answer it only for hypotheses built from trust atoms.

Without A2, A1 fails. Atoms that mention nodes and parameters are not
bounded by the situation, since a hypothesis may name any number of
further strands. With infinitely many candidate assumptions,
Theorem~\ref{thm:completion} and the dualisation of
Section~\ref{sec:blocker} lose the finiteness they use. Two things
survive. First, achievement is still monotone in the hypothesis and
antitone in the conclusion, so the Galois connection exists when the
hypotheses are completed by their up-sets, which form a complete lattice.
Second, Lemma~\ref{lem:separable} does not use finiteness: a conjunctive
hypothesis fails on a run exactly when one of its conjuncts does. What is
lost is the description of $\weakest(\chi)$ as a finite antichain. An
up-set of hypotheses need not be generated by finitely many minimal
elements. In that case, $\weakest(\chi)$ exists only as an up-set and has
no formula in the goal language extended by finite disjunctions.

Lifting the restriction requires two steps. The first is order-theoretic.
If the hypotheses, ordered by strength modulo the situation, form a
well-quasi-order, every up-set is the up-closure of a finite set, and
Theorem~\ref{thm:completion} holds with finite antichains. We do not know
whether the conjunctions of atoms of \cite{rowe} under their order have
this property. The second is computational. Line~3 of
Algorithm~\ref{alg:loop} takes a blocker over $U$, and over an infinite
set of atoms the candidate step must change, for instance to a finite $U$
grown from the atoms that appear in the attacks found so far. The Galois
connection needs neither step, but computing with it needs both.

\subsection{Either/Or Hypotheses}
\label{sec:lim-completion}

Corollary~\ref{cor:diagnosis} obtains a unique answer by allowing hypotheses
that are disjunctions of conjunctions of trust atoms. This extends the
goal language of \cite{rowe}. Without the extension, a weakest hypothesis
exists only under the condition of Theorem~\ref{thm:unique}.

Deciding achievement for an either/or hypothesis requires nothing new. The
goal $\forall\bar{x}.\,((\Phi_1 \vee \Phi_2) \Rightarrow \Psi)$ is
equivalent to $\forall\bar{x}.\,(\Phi_1 \Rightarrow \Psi) \wedge
\forall\bar{x}.\,(\Phi_2 \Rightarrow \Psi)$, so it is achieved exactly when
each of finitely many goals of the original fragment is, and one
enrich-by-need run per disjunct decides it. Thus, results stated goal by
goal carry over. However, results that rely on the shape analysis sentence
of a single hypothesis need not carry over. Enrich-by-need starts from one
skeleton, an either/or hypothesis has one set of skeletons per disjunct,
and its strongest guarantee is the disjunction of theirs (the
remark after Corollary~\ref{cor:conjunction}). The strength orderings of
\cite{rowe} and the goal-preservation results of
\cite{guttmanestablishing} are of this kind.

Lifting this limit means restating those results for sets of skeletons and
checking which survive. We have not done so. The claim is also narrow in
another direction. The completion gives a unique answer here because
sufficiency is monotone (A4) and the hypotheses form a finite lattice
(A1). An abduction problem without these properties need not have a
unique weakest explanation, and we claim nothing about abduction in
general.

\subsection{Safety Goals (A3)}
\label{sec:lim-safety}

We assume that every goal violation has a finite witness, as for secrecy,
aliveness, and non-injective agreement. A1 alone makes $\attacks(\chi)$
and $\weakest(\chi)$ finite. A3 makes failure semi-decidable: a violation,
when there is one, is a finite run that a search can exhibit. Deciding
achievement requires more, namely an oracle that terminates, as CPSA does on
every row we report. Without termination,
line~4 of Algorithm~\ref{alg:loop} may not return on a sufficient
candidate, and the procedure is a semi-algorithm. The description
$\weakest(\chi) = \blocker(\attacks(\chi))$ and Theorem~\ref{thm:unique}
depend on neither. Only the oracle does.

Injective agreement is excluded for a different reason. Its statement
requires an equality between the partner strands matched to two runs of
the observed role, which a goal over one fixed situation $\sigma$ does not
state. The exclusion concerns expressibility rather than the length of a
violating run.

Lifting A3 requires an oracle that certifies success for the goal class,
such as an unbounded verifier whose analysis may not terminate. The loop
needs nothing else from it. A partial run is still useful. Every member of
$H$ at any point of Algorithm~\ref{alg:loop} is a conflict computed from a
real counterexample, by Proposition~\ref{prop:clause}. Thus, if the oracle
does not return on a candidate, every sufficient trust still meets every member
of $H$. The partial output is a sound statement of necessity, of the same
kind as the analyser rows of Table~\ref{tab:results}.

\subsection{No Equational Theory (A5)}
\label{sec:lim-algebra}

We assume perfect cryptography over a free term algebra. The proof of
Lemma~\ref{lem:separable} uses the free algebra only through the fact
that a run either violates an atom or does not. With an equational
theory, new atoms, for instance about exponents, would have to be checked
against that reading before Proposition~\ref{prop:hitting} applies, and
the oracle would have to decide achievement modulo the theory. For this
reason, the Diffie--Hellman protocols shipped with CPSA are not evaluated
in Section~\ref{sec:eval}.

Lifting A5 requires an oracle for the equational theory and a new proof of
Lemma~\ref{lem:separable} for the atoms the theory adds.

\subsection{The Strongest Goal and Termination}
\label{sec:lim-channels}

The strongest goal $\strongest(\tau)$ exists when enrich-by-need
terminates from $\sigma \wedge \tau$, and we assume this for every trust,
over origination and channel atoms alike (the termination assumption of
Section~\ref{sec:galois}). Rowe et al.\ \cite[Theorem~2]{rowe}
characterise $\strongest(\tau)$ as the shape analysis sentence
\cite{guttmantosca,ramsdell} when the analysis terminates, but they do not
supply termination. For channel atoms, we also assume that the
characterisation carries over to the restricted model, which \cite{rowe}
does not cover. Without these assumptions, $\strongest$ may be undefined on
some trusts, and the results that use it (the right adjoint in
Theorems~\ref{thm:galois} and~\ref{thm:completion}, and
Theorem~\ref{thm:correspondence}) are conditional there. The left adjoint
$\weakest$, its description as a blocker, and Algorithm~\ref{alg:loop} use
only $\Suff(\chi)$ and do not depend on it. An enrich-by-need analysis
with channels in the sense of \cite{kamillowe}, with a termination
argument, would remove the assumption.

\subsection{Tight Trusts}
\label{sec:lim-tight}

Theorem~\ref{thm:correspondence} pairs tight trusts with attainable goals,
but it does not discard assumptions. Conclusions of the goal language may
contain origination atoms, and $\strongest(\tau)$ retains those of $\tau$.
As a result, $\weakest(\strongest(\tau))$ retains them as well, and nearly
every trust is tight. Computing
$\weakest(\chi)$ for a goal identifies which stated assumptions that goal
can do without. A tight core that discards unneeded assumptions
would need conclusions free of trust atoms, for instance $\strongest(\tau)$ with its
origination atoms removed. Whether the correspondence, or some
replacement for it, holds for such conclusions is open.

\subsection{A Bounded Analyser}
\label{sec:lim-bound}

Table~\ref{tab:results} has \numLNAnalyserRows{} rows on our bounded
analyser (the channel row is shown at two pool sizes), whose sufficiency verdicts hold only within its pool of strands.
On these rows, a claim of sufficiency is an observation about the bounded
model. Within \numLBoundTimeout\,s per run, the pool could be enlarged
only to two instances per role, as Section~\ref{sec:eval-bound} reports.

Lifting this limit requires an oracle that is complete without a bound for
the atoms in question. CPSA, when its search is exhausted, is such an
oracle for the \texttt{non} and \texttt{uniq} facts our encoding uses. The
channel assumptions would require an encoding of the channel levels in a complete
analysis, which we have not written. Short of that, a pool rule fixed per
goal, stating which earlier sessions and how many instances the goal's
attacks need, would turn the sensitivity table into a check.

\subsection{Cost on Large Protocols}
\label{sec:lim-scale}

Every protocol row has at most \numLMaxU{} candidate assumptions, and the
scale study uses synthetic instances whose components do not interact.
Hence, neither measures the cost of the loop when each oracle call is a
full analysis of a large protocol. That
cost grows with $|U|$ in a second way. Algorithm~\ref{alg:stops} makes $|U
\setminus T|$ replays per counterexample, and on the CPSA rows, the
replays already dominate the verifier runs.
Measuring the combined cost would require a protocol with ten or more
candidate assumptions over keys, nonces, and channels, analysed by CPSA.
We have not run this experiment.

\subsection{The Composition Conjecture}
\label{sec:lim-composition}

The equality condition of Theorem~\ref{thm:composition} is a condition on
the composition's minimal attacks, which a designer cannot read off the
components. Remark~\ref{rem:disjoint} conjectures that disjoint encryption
\cite{guttmanthayer} implies it. Until the conjecture is settled, the
theorem gives a lower bound for any pair in which every component attack
extends to an attack on the composition that still violates the goal and
whose stopping set is contained in the original's. It gives equality for
a given pair only by computing the composition's minimal attacks, as
Table~\ref{tab:composition} does. Protocols that share a key
infrastructure can fail the condition, as the untagged pair there does,
and for them the lower bound is strict.

Proving the conjecture means following stopping sets through the
transformation by which \cite{guttmanthayer} maps a run of the
composition to behaviour of a component. It also means showing that the
image of an attack is ruled out by no atom that does not also rule out the
original.

\section{Related Work}
\label{sec:related}

The weakest trust, the minimal ways to stop every minimal attack, draws on
three bodies of work: protocol analysis in strand spaces, the search for
weakest assumptions in verification, and the combinatorics of hitting
sets. The closest work evaluates security claims over lattices of
adversary or threat models.

\subsection{The Forward Direction}

Our setting is the strand-space model of Thayer, Herzog, and Guttman
\cite{thayer} with the goal language of Guttman
\cite{guttmanshapes,guttmanestablishing}, in which a security goal is a
geometric formula and a protocol achieves it when every run satisfies
it. Guttman's enrich-by-need analysis computes the shapes of a skeleton,
the minimal executions compatible with it, and CPSA \cite{cpsaspec}
implements it. The shape analysis sentence, which states everything one
terminating run of the analysis establishes, is due to Guttman
\cite{guttmantosca} and Ramsdell \cite{ramsdell}. Rowe, Guttman, and
Liskov \cite{rowe} use it to order goals and protocols by strength. They
show that for a fixed hypothesis on which the analysis terminates, the
achieved goals have a strongest member, the shape analysis sentence. For a
fixed conclusion, however, the sufficient hypotheses need not have a
weakest one. They report knowing of no tools designed to support the
search for the maximal hypotheses and leave it for future research. We
take their fragment and achievement relation as given and answer that
question for hypotheses built from trust atoms. Their own method for the
inverse direction drops conjuncts from a sufficient hypothesis until the
goal fails and returns one maximal hypothesis. Section~\ref{sec:eval} compares
it with Algorithm~\ref{alg:loop}.

Dougherty, Guttman, and Ramsdell \cite{dougherty} compute minimal
executions of a protocol in a context, combining CPSA with an SMT-based
model finder. Their minimality concerns runs, while ours concerns the
assumptions under which no run violates the goal. The two meet in
Algorithm~\ref{alg:loop}, whose witnesses are minimal runs and which
minimises their stopping sets. Theorem~\ref{thm:cost} records that a
minimal run need not have a minimal stopping set. Lowe's hierarchy of
authentication goals \cite{lowe} supplies the goals of
Table~\ref{tab:results}.

\subsection{Weakest Assumptions}

Several lines of work ask which assumptions a property needs. Yatapanage
and Jones \cite{yatapanage} use rely/guarantee reasoning to make explicit,
by hand, the environment assumptions on which the Needham--Schroeder
public-key protocol depends. Alves-Foss and Soule \cite{alvesfoss}
generate necessary preconditions for protocol goals in the SvO belief
logic, by a weakest-precondition calculus.
Belief logics such as BAN \cite{ban} require the initial assumptions to be
stated and derive goals from them, but they do not minimise the assumptions.
None of these works uses a Dolev--Yao semantics with a verifier as
oracle, and none treats the non-uniqueness of the answer.

Outside security, assumption generation for compositional verification
asks for the weakest environment under which a component satisfies a
property. Giannakopoulou, P\u{a}s\u{a}reanu and Barringer
\cite{giannakopoulou} construct it from the component's transition
system, and Cobleigh, Giannakopoulou and P\u{a}s\u{a}reanu \cite{cobleigh}
learn sufficient assumptions with an automata-learning algorithm. There,
the assumptions are sets of traces, the environments under which a safety
property holds are closed under union, and the weakest assumption is
unique. For trusts, Theorem~\ref{thm:nofilter} shows that the
corresponding closure fails.

Chatterjee, Henzinger and Jobstmann \cite{chatterjee} compute environment
assumptions that make an unrealisable specification realisable in a game,
in polynomial time for safety assumptions, and show that minimal liveness
assumptions are NP-hard. They also show that weakest environment
assumptions need not be unique, with an example in which either of two
assumptions suffices and neither is implied by the other. Our two-key
variant is an instance of the same phenomenon. Their assumptions are
conditions on environment behaviour in a game graph, whereas ours are
atoms about keys, values, and channels.
Their safety and liveness cases differ in complexity, while here the
distinction affects only the oracle.

Bi-abduction \cite{calcagno} infers, during a compositional
separation-logic analysis, the missing part of a precondition together
with the unused part of the state, and chooses among the possible answers
heuristically. Its hypothesis language is much richer than ours, and we do
not claim that a completion turns it into an adjunction.

Dijkstra's weakest precondition \cite{dijkstra} exists because the
predicate transformer is universally conjunctive, which makes it a right
adjoint. Achievement is monotone in the trust but does not preserve
meets, so $\weakest$ requires the completion.

\subsection{Hitting Sets and Minimal Cut Sets}

The minimal ways to stop every minimal attack are minimal hitting sets,
and the combinatorics is Reiter's theory of diagnosis \cite{reiter}. The
minimal diagnoses of a system are the minimal hitting sets of its minimal
conflict sets. The hitting-set tree, corrected by Greiner, Smith, and
Wilkerson \cite{greiner}, computes them with a theorem prover that returns
conflicts. Algorithm~\ref{alg:loop} is that loop with a protocol verifier
in place of the theorem prover. A stopping set plays the role of a
conflict, and a trust plays that of a diagnosis. The loop is also
counterexample-guided inductive synthesis \cite{cegis} with achievement as
the specification. In satisfiability, the same duality relates minimal
correction sets and minimal unsatisfiable subsets, and Liffiton and
Sakallah \cite{liffiton} compute both on this basis. MARCO \cite{marco}
enumerates both by alternating between them, much as
Algorithm~\ref{alg:loop} alternates between candidates and attacks.

The bound of Theorem~\ref{thm:cost} on goal queries is an instance of
known results. Achievement is a monotone Boolean function of the trust,
and a goal query is a membership query on it. Angluin \cite{angluin}
learns monotone DNF formulas exactly from membership and equivalence
queries, and Bshouty \cite{bshouty} learns a Boolean function with a
number of queries polynomial in the sizes of its DNF and CNF. For a
monotone function, these are the minimal true points and the minimal
false clauses. Bioch and Ibaraki \cite{biochibaraki} show that
identifying a monotone function from membership queries, which yields
both families, has an incrementally polynomial algorithm exactly when
dualisation does. Gunopulos et
al.\ \cite{gunopulos} give the same connection for finding the most
specific sentences of a monotone predicate, with Dualize-and-Advance,
whose number of predicate evaluations is bounded in terms of the sizes of
the positive and negative borders. In those terms, $\weakest(\chi)$ is the
set of minimal true points, and the complements of the members of
$\attacks(\chi)$ are the maximal false points. What is specific here is
the oracle: a protocol verifier, with a replay query that recovers a
stopping set from a run.

The involution of blockers is due to Edmonds and Fulkerson \cite{edmonds}
and to Isbell \cite{isbell}, and is in Berge's treatment of hypergraphs
\cite{berge}. Checking a candidate dualisation is quasi-polynomial by
Fredman and Khachiyan \cite{fredman}. Whether minimal transversals can be
enumerated in output-polynomial time is open, and Eiter and Gottlob
\cite{eiter} relate it to equivalent problems in logic and databases.
Minimum hitting set is NP-complete \cite{karp}. We use these results as
stated. What we add is that $\weakest(\chi)$ and $\attacks(\chi)$ are a
blocker pair, and Theorem~\ref{thm:unique} is a direct consequence of the
involution.

In fault-tree analysis \cite{vesely,ruijters}, a minimal cut set is a
minimal set of basic events whose joint occurrence causes the top event,
and a minimal path set is a minimal set of components whose functioning
prevents it. The two families are blockers of each other. With the failure
of an assumption as a basic event, $\attacks(\chi)$ is the family of
minimal cut sets, and $\weakest(\chi)$ is the family of minimal path sets.
Attack trees \cite{schneier,mauw} and attack--defence trees \cite{kordy}
carry the same duality into security. Mauw and Oostdijk give an attack
tree a semantics as a set of multisets of basic attack steps, and choosing
defences that counter every attack is a hitting-set problem. In all three,
the tree is written by an analyst and given to the analysis. Here, no tree
is given. The minimal attacks follow from the protocol's semantics and are
reached only through a verifier, and Algorithm~\ref{alg:loop} finds them
one counterexample at a time.

The order theory is also standard. Birkhoff's polarities
\cite{birkhoff}, on which formal concept analysis builds \cite{ganter},
give an antitone Galois connection between sets of trusts and sets of
goals. The isotone connection between single trusts and single goals of
Section~\ref{sec:galois} exists exactly when every goal has a single
weakest trust. The disjunctive completion that supplies it in general is
the construction of abstract interpretation \cite{cousot}.

\subsection{Lattices of Adversaries and Threat Models}

Basin and Cremers \cite{basincremersesorics,basincremers} define adversary
models by combining compromise rules: reveal of long-term keys, under
conditions on whether the owner is a partner of the session under attack
and on whether the reveal happens after that session ends; reveal of
session keys; reveal of local state; and reveal of random values. These
rules generate a lattice of adversary models. For each protocol and
claim, they use an extension of Scyther to determine the models against
which the claim holds. They order protocols by these sets into
protocol-security hierarchies, in which protocols can be incomparable.

Their work began the line that is closest to ours. The set of their models
against which a claim holds is down-closed in adversary capability. Read
in the opposite order, this set corresponds to $\Suff(\chi)$, and its
maximal adversaries correspond to the members of $\weakest(\chi)$. An
adversary that may not reveal a key corresponds to assuming that key
uncompromised. Their hierarchies order protocols by containment of these
sets, as $\preceq$ does in Theorem~\ref{thm:tenet}. Thus, for their
lattice of adversary models, the inverse question is already answered by
tool support, by checking each model.

The two settings differ in three ways. First, the candidate assumptions
differ. Theirs are uniform rules over all agents and sessions, some with
temporal conditions, and their lattice is fixed in advance for every
protocol. Ours are atoms about particular values and channels in a
situation, and the lattice has $2^{|U|}$ points that depend on the
protocol and the situation. Neither set subsumes the other. A reveal
permitted only after the session ends is not expressible by our
non-origination atoms, which constrain the whole run, and freshness of a
particular value or authenticity of a particular channel is not one of
their rules. Second, the method differs. They evaluate each model. We
describe the minimal elements as a blocker, give the uniqueness criterion
of Theorem~\ref{thm:unique}, and compute them with a number of goal
queries bounded, for minimal witnesses, by the number of minimal attacks
and weakest trusts, by Theorem~\ref{thm:cost}, independent of the size of
the lattice. Third, the use differs. They rank protocols by the
adversaries they resist, whereas we compute, for one protocol and goal,
which assumptions the goal requires.

Later work evaluates claims across lattices of threat models and computes
their extreme elements. Girol et al.\ \cite{girol} analyse every Noise
handshake pattern under combinations of compromise capabilities, and
Jackson et al.\ \cite{jackson} analyse protocols under a range of symbolic
models of signatures weaker than the ideal one. Jacomme and Kremer
\cite{jacommemfa} analyse multi-factor authentication protocols under a
lattice of threat models built from atomic capabilities, such as malware
on a device, phishing, and users who omit protocol steps, and compare
protocols by
the threat models under which each property holds. Cheval et al.\
\cite{chevalhash} build a lattice of threat models from weaknesses of
hash functions. Jacomme, Klein, Kremer, and Racouchot \cite{jacommeedhoc}
compute, for each property of EDHOC, the maximal threat models under which
it holds and the minimal ones under which it fails. They work over the
inclusion lattice of atomic capabilities, using enumeration with monotone
pruning. Rothstein-Morris, Sun, and Chattopadhyay \cite{rothsteinmorris}
generate attackers, as sets of capabilities, by bounded model checking, and
classify them by the requirements each breaks, with pruning that uses the
same monotonicity. Cremers and Dehnel-Wild \cite{cremers5g} make the
channel assumptions of 5G-AKA explicit and show which of them its goals
rely on.

If a capability is read as the negation of one of our atoms, a maximal
threat model under which a property holds corresponds to a member of
$\weakest(\chi)$, and a minimal one under which it fails corresponds to a
member of $\attacks(\chi)$. Thus, the dual question has already been
computed for several protocols, by enumeration over fixed capability
lattices. Three things differ here. First, the question is posed in the
goal language of \cite{rowe}, with atoms about particular values and
channels of a situation. Second, the answer is characterised as a
blocker, with the uniqueness criterion of Theorem~\ref{thm:unique}. Third,
it is computed by a counterexample-guided loop over a verifier rather than
by enumeration. The loop's goal queries are bounded by $|\attacks(\chi)| +
|\weakest(\chi)|$ when witnesses are minimal, which is the learning bound
recalled above.

Rowe \cite{roweorderings} studies orderings used in security models.
Rowe, Guttman, and Ramsdell \cite{rowedb} analyse distance-bounding
protocols with CPSA from an assumption-based perspective and compare the
protocols by the assumptions each needs. Neither work computes a weakest
assumption set.

\subsection{Channels as Assumptions}

Dilloway and Lowe \cite{dilloway} specify secure transport channels by
the adversary behaviours each forbids and order the specifications in a
hierarchy. M\"odersheim and Vigan\`o \cite{modersheim}, building on that
line, give authentic, confidential, and secure channels as assumptions in
a protocol notation, with a pseudonymous variant in which an endpoint is
identified only by a pseudonym. Kamil and Lowe \cite{kamillowe} specify
channels in strand spaces as restrictions on the adversary, including
faking messages and hijacking, which changes the sender or receiver
identity of a message on the same channel. These works analyse a protocol
under given channel assumptions, and Cremers and Dehnel-Wild
\cite{cremers5g} show for 5G-AKA which channel assumptions its goals rely
on. We use the authentic and confidential levels as candidate
assumptions. The inject/read reading of Section~\ref{sec:channels}, and
its hijack variant, which moves a message between channels, are ours.
Under both readings, Lemma~\ref{lem:generators} shows that the secure
level is the conjunction of the other two, so the trust lattice remains
Boolean.

\subsection{Composition}

Guttman and Thayer \cite{guttmanthayer} prove that protocols satisfying
disjoint encryption, in which the regular strands of one protocol neither
create nor open the encrypted units of the other, do not interfere. A run
of the composition can be transformed into behaviour of the primary
protocol, so the primary's goals are preserved. Guttman
\cite{guttmanestablishing} develops goal preservation in the goal language
we use. Kelsey, Schneier, and Wagner \cite{kelsey} show the opposite
direction: for a protocol whose keys are shared, an adversary can design a
second protocol that breaks it. Cremers \cite{cremersares} shows that
attacks using messages of another protocol already in use, called
multi-protocol attacks, are common among published protocols that share
keys. Theorem~\ref{thm:composition} states interference in terms of
trust. When every component attack extends to an attack on the
composition with no larger stopping set, the join of the components'
weakest trusts is a lower bound for the composition. Equality holds when
every minimal attack on the composition is stopped by no less than some
component attack. Remark~\ref{rem:disjoint} conjectures that disjoint
encryption is a sufficient condition. The untagged pair of
Table~\ref{tab:composition} has two multi-protocol attacks, in the sense
of Cremers, that make the composition's weakest trust
strictly greater than the join.

\section{Conclusion}
\label{sec:conclusion}

The weakest trust sufficient for a goal is exactly the set of minimal
hitting sets of all minimal attacks, $\weakest(\chi) =
\blocker(\attacks(\chi))$. The answer is unique when either/or hypotheses
over trust atoms are allowed. A single
weakest trust exists exactly when every minimal attack has one stopping
assumption. The weakest trust and the strongest
goal of Rowe, Guttman, and Liskov are the two sides of one Galois
connection, and the goal language allows disjunction on one side only.
Reiter's hitting-set loop computes this answer
using a protocol verifier as an oracle. Using our analyser and CPSA, for
all \numNRows{} goals in Table~\ref{tab:results}, the results agree with
evaluating every trust using the same oracle.

The main open question is whether these results extend to the full goal
language of \cite{rowe}, where hypotheses may also contain role-position
and parameter atoms. The Galois connection exists there in the lattice of
up-sets of hypotheses. It is not known whether the weakest hypothesis is
always a finite disjunction, and thus a formula of the extended goal
language, or whether it can be computed.
Section~\ref{sec:limits} lists narrower open questions. These include
whether disjoint encryption implies the equality condition of
Theorem~\ref{thm:composition}, and whether a tight core that discards
unneeded assumptions exists for conclusions free of trust atoms.

The advice in Section~\ref{sec:intro} was to trust a few. For a protocol
and a goal, the few are now computable: they are the minimal ways to stop
every minimal attack.

\section*{Acknowledgments}
Aalok Thakkar is supported by the Anusandhan National Research Foundation, under the Prime Minister Early Career Research Grant ANRF/ECRG/2025/001136/ENS, and by the Vachani School of Advanced Computing, Ashoka University, under the Young Faculty Grant (00281).



\appendix
\section{The CPSA Oracle}
\label{app:oracle}

Listing~\ref{lst:oracle} is the CPSA oracle, abridged from the artifact. The
candidate trust is added to the hypothesis of a CPSA goal. On failure, each
failing shape is replayed under further atoms to compute its clause, the
atoms whose assumption removes the shape (Algorithm~\ref{alg:stops}), and
the smallest clause is returned. One call of \texttt{oracle} is one goal
query, the unit counted as a call in Section~\ref{sec:eval}. The CPSA runs
made by \texttt{replay} inside it are counted separately, as replay runs.

\begin{lstlisting}[language={[Objective]Caml},caption={The CPSA oracle: the trust is the goal's hypothesis.},label={lst:oracle}]
let stops tmpl trust shape =                (* grow M while x survives *)
  let rest = Aset.diff (universe tmpl) trust in
  let m = Aset.fold (fun a m ->
      if survives (replay tmpl shape (Aset.add a m))
      then Aset.add a m else m) rest Aset.empty in
  Aset.diff rest m                          (* the clause U \ (T u M) *)

let oracle tmpl trust =
  let g = goal_run tmpl trust in            (* trust -> non/uniq facts *)
  if achieves g then Loop.Achieves
  else
    let clauses = List.map (stops tmpl trust) g.failing in
    Loop.Fails { stops = smallest clauses; descr = "" }
\end{lstlisting}

\section{Protocol Models}
\label{app:protocols}

For every row of Table~\ref{tab:results}, this appendix gives what was
analysed. For the analyser rows, it lists the strand pool and the goal. For
the CPSA rows, it lists the template file of the artifact, the source of the protocol
text, the candidate assumptions $U$ with the CPSA facts they stand for, the role-level
declarations that stay in the situation $\sigma$, and the goal verbatim. In a
template, \texttt{@TRUST@} marks where the facts of the candidate trust are
inserted into the goal's hypothesis, and the empty trust inserts nothing. CPSA
runs each template with no options.

\subsection{Analyser Rows}
\label{app:protocols-analyser}

In the analyser every strand is a sequence of sends and receives over message
templates. A variable is bound by the earliest receive that mentions it, to
an atom of a fixed list. A starred strand is the test strand, whose
completion the goal is about. Copy $j \ge 2$ of a pool (Table~\ref{tab:bound})
repeats every strand with fresh nonces of its own and is never a test
strand.

\paragraph{Signed challenge-response and the two-key variant.} Write
$r(P,Q,n)$ for the response: $\sig{sk_P}{n, P, Q}$ in the signed row and
$\{\!|\{\!|n, P, Q|\!\}_{k_1}|\!\}_{k_2}$ in the two-key variant. The pool is
\begin{center}
\footnotesize
\setlength{\tabcolsep}{4pt}
\begin{tabular}{ll}
\toprule
strand & steps\\
\midrule
$\mathrm{Chal}(B,A)^*$ & send $N_t$; receive $r(A,B,N_t)$\\
$\mathrm{Prov}(A,B)$ & receive $n$; send $r(A,B,n)$\\
$\mathrm{Prov}(B,A)$ & receive $n$; send $r(B,A,n)$\\
$\mathrm{Chal}(A,B)$ & send $N_a$; receive $r(B,A,N_a)$\\
$\mathrm{Prov}(A,I)$ & receive $n$; send $r(A,I,n)$\\
\bottomrule
\end{tabular}
\end{center}
with $U = \{\Non(sk_A), \Non(sk_B), \Unq(N)\}$ (respectively $\Non(k_1),
\Non(k_2)$). The adversary knows the names, the public keys, $sk_I$, and every
key of $U$ not assumed. With the earlier session, it also holds the transcript
$N_{old}, r(A,B,N_{old})$ of a completed $\mathrm{Prov}(A,B)$ run, and the
test nonce $N_t$ is either the fresh $N$ or, unless $\Unq(N)$ is assumed, the
stale $N_{old}$. The goal is recent agreement: when the test strand completes,
some $\mathrm{Prov}(A,B)$ strand sent $r(A,B,N_t)$ after the test strand sent
$N_t$.

\paragraph{Needham--Schroeder and Lowe's fix.} The pool is
\begin{center}
\footnotesize
\setlength{\tabcolsep}{4pt}
\begin{tabular}{ll}
\toprule
strand & steps\\
\midrule
$\mathrm{Init}(A,B)$ & send $\{N_{a0}, A\}_{pk_B}$; receive $\{N_{a0}, N_b\}_{pk_A}$; send $\{N_b\}_{pk_B}$\\
$\mathrm{Init}(A,I)$ & send $\{N_{a1}, A\}_{pk_I}$; receive $\{N_{a1}, N_b\}_{pk_A}$; send $\{N_b\}_{pk_I}$\\
$\mathrm{Resp}(B,A)^*$ & receive $\{N_a, A\}_{pk_B}$; send $\{N_a, N_{b2}\}_{pk_A}$; receive $\{N_{b2}\}_{pk_B}$\\
\bottomrule
\end{tabular}
\end{center}
and in Lowe's fix the second message carries the responder's name,
$\{N_a, N_b, B\}_{pk_A}$. $U = \{\Non(sk_A), \Non(sk_B)\}$, and in the channel
rows also $\auth(c)$ and $\conf(c)$ for $c \in \{A{\to}B, B{\to}A\}$. Each
message travels on the channel from its sender to its intended receiver. The
adversary knows the names, the public keys, $sk_I$, and every key of $U$ not
assumed. The goal is the responder's non-injective agreement: when the test
strand completes, some $\mathrm{Init}(A,B)$ strand has completed with the same
$N_a$ and $N_b$. With the earlier session, the adversary also holds the three
messages of a completed run between $A$ and $B$ on fresh nonces of its own.

\paragraph{Composition pairs.} $\Pi_1$ is the signed challenge-response
with one strand per role, and $\Pi_2$ has three strands: $C$ sends
$\sig{sk_C}{m_C, B}$; $A$ receives $\sig{sk_C}{m, B}$ and sends
$\sig{k}{m, A, B}$; $B$ receives $\sig{k}{M, A, B}$. In the tagged pair the
two signed formats carry the constants $\mathsf{cr}$ and $\mathsf{att}$. The key $k$
is $sk_A$ in the shared-key pairs and a separate $sk'_A$ in the third.
The candidate assumptions are $\Non(sk_A)$, $\Non(sk_C)$, and $\Unq(N)$, and also
$\Non(sk'_A)$ for separate keys. $\chi_1$ is recent agreement of $\Pi_1$'s $B$ with $A$ on $N$, and
$\chi_2$ is agreement of $\Pi_2$'s $B$ with $A$ on $M$.

\subsection{CPSA Rows}
\label{app:protocols-cpsa}

Every CPSA row except the adopted \textsc{pkinit} fix uses a protocol
shipped in the \texttt{tst/} directory of CPSA~\numLCpsaVersion{}, with the protocol text unchanged except where noted.
None of the shipped files has a \texttt{defgoal} for these protocols, so we
wrote every goal, by the rule of Section~\ref{sec:eval-protocols}, and no
shipped \texttt{defgoal} was used. The
agreement goals fix the observed role's completed strand and its parameters
in the hypothesis and ask for a strand of the partner role that agrees on
the parameters named. They are non-injective. The atom names in $U$ are those
of Table~\ref{tab:results}.

\paragraph{\textnormal{\textsc{pkinit}}, flawed and fixed.} The flawed row is
\texttt{tst/pkinit.scm}, renamed. The fix is not shipped with CPSA. It is
our model of the adopted fix, written in the same style, in which the
server signs the session key together with a hash of the client's request.
The standard's fix uses a keyed checksum of the request; our model
simplifies it to the unkeyed \texttt{hash} of the whole request. In both
rows, the client declares its nonces $n_1, n_2$ \texttt{uniq-orig}, and the
server declares its keys $k, ak$, so their freshness is part of $\sigma$
and not of $U$. The goal is the client's authentication of the server,
that is, agreement with an \texttt{auth} strand on the client, the server,
and the session key, and the answers are for this goal only.
\par\smallskip\noindent\begin{tabular}{@{}l p{0.8\textwidth}@{}}
file & \texttt{cpsa/protocols/pkinit-flawed.scm}\\
$U$ & $\Non(sk(S))$ = \texttt{(non (privk as))}, $\Non(sk(C))$ = \texttt{(non (privk c))}\\
in $\sigma$ & \texttt{(uniq-orig n1 n2)} in \texttt{client}; \texttt{(uniq-orig k ak)} in \texttt{auth}\\
\end{tabular}
\begin{lstlisting}[language={},aboveskip=4pt]
(defgoal pkinit-flawed
  (forall ((c as name) (k skey) (z strd))
    (implies
      (and (p "client" z 2)
           (p "client" "c" z c)
           (p "client" "as" z as)
           (p "client" "k" z k)
           @TRUST@)
      (exists ((z-0 strd))
        (and (p "auth" z-0 1) (p "auth" "as" z-0 as)
             (p "auth" "k" z-0 k) (p "auth" "c" z-0 c))))))
\end{lstlisting}

\par\smallskip\noindent\begin{tabular}{@{}l p{0.8\textwidth}@{}}
file & \texttt{cpsa/protocols/pkinit-fix2.scm}\\
$U$ & $\Non(sk(S))$ = \texttt{(non (privk as))}, $\Non(sk(C))$ = \texttt{(non (privk c))}\\
in $\sigma$ & \texttt{(uniq-orig n1 n2)} in \texttt{client}; \texttt{(uniq-orig k ak)} in \texttt{auth}\\
\end{tabular}
\begin{lstlisting}[language={},aboveskip=4pt]
(defgoal pkinit-fix2
  (forall ((c as name) (k skey) (z strd))
    (implies
      (and (p "client" z 2)
           (p "client" "c" z c)
           (p "client" "as" z as)
           (p "client" "k" z k)
           @TRUST@)
      (exists ((z-0 strd))
        (and (p "auth" z-0 1) (p "auth" "as" z-0 as)
             (p "auth" "k" z-0 k) (p "auth" "c" z-0 c))))))
\end{lstlisting}

\paragraph{ISO 9798-3, flawed and corrected.} The flawed protocol
is shipped as \texttt{tst/isoreject.scm}. The corrected version,
\texttt{tst/isoreject-corrected.scm}, adds distinct tags to the two signed
messages. The goal is the responder's agreement with
the initiator on the names and the responder's nonce. In the flawed version,
$B$'s responder strand accepts a signature that $A$ produced as a responder
in a session with $B$ started by the adversary, since the two signed messages
have the same format.
\par\smallskip\noindent\begin{tabular}{@{}l p{0.8\textwidth}@{}}
file & \texttt{cpsa/protocols/isoreject.scm}\\
$U$ & $\Non(sk_A)$ = \texttt{(non (privk a))}, $\Non(sk_B)$ = \texttt{(non (privk b))}, $\Unq(n_A)$ = \texttt{(uniq na)}, $\Unq(n_C)$ = \texttt{(uniq nc)}\\
in $\sigma$ & \texttt{(uniq-orig nb)} in \texttt{resp}\\
\end{tabular}
\begin{lstlisting}[language={},aboveskip=4pt]
(defgoal isoreject
  (forall ((a b name) (na nb nc text) (z strd))
    (implies
      (and (p "resp" z 3) (p "resp" "a" z a) (p "resp" "b" z b)
           (p "resp" "na" z na) (p "resp" "nb" z nb) (p "resp" "nc" z nc)
           @TRUST@)
      (exists ((z-0 strd))
        (and (p "init" z-0 3) (p "init" "a" z-0 a) (p "init" "b" z-0 b)
             (p "init" "nb" z-0 nb))))))
\end{lstlisting}

\par\smallskip\noindent\begin{tabular}{@{}l p{0.8\textwidth}@{}}
file & \texttt{cpsa/protocols/isofix.scm}\\
$U$ & $\Non(sk_A)$ = \texttt{(non (privk a))}, $\Non(sk_B)$ = \texttt{(non (privk b))}, $\Unq(n_A)$ = \texttt{(uniq na)}, $\Unq(n_C)$ = \texttt{(uniq nc)}\\
in $\sigma$ & \texttt{(uniq-orig nb)} in \texttt{resp}\\
\end{tabular}
\begin{lstlisting}[language={},aboveskip=4pt]
(defgoal isofix
  (forall ((a b name) (na nb nc text) (z strd))
    (implies
      (and (p "resp" z 3) (p "resp" "a" z a) (p "resp" "b" z b)
           (p "resp" "na" z na) (p "resp" "nb" z nb) (p "resp" "nc" z nc)
           @TRUST@)
      (exists ((z-0 strd))
        (and (p "init" z-0 3) (p "init" "a" z-0 a) (p "init" "b" z-0 b)
             (p "init" "nb" z-0 nb))))))
\end{lstlisting}

\paragraph{Blanchet, flawed and fixed.} Both protocols come from
\texttt{tst/blanchet.scm}, with comments dropped. The goal is the secrecy
of $d$ from the responder's point of view, expressed as the file's analysis
of the responder with a listener for $d$, written as a goal whose
conclusion is \texttt{false}. In the flawed version, $A$ runs the initiator with a compromised peer, and the
adversary re-encrypts $A$'s signed key for $B$. The fix puts $B$'s name inside
the signature.
\par\smallskip\noindent\begin{tabular}{@{}l p{0.8\textwidth}@{}}
file & \texttt{cpsa/protocols/blanchet.scm}\\
$U$ & $\Non(sk_A)$ = \texttt{(non (privk a))}, $\Non(sk_B)$ = \texttt{(non (privk b))}, $\Unq(s)$ = \texttt{(uniq s)}, $\Unq(d)$ = \texttt{(uniq d)}\\
\end{tabular}
\begin{lstlisting}[language={},aboveskip=4pt]
(defgoal blanchet
  (forall ((a b name) (s skey) (d data) (z z-0 strd))
    (implies
      (and (p "resp" z 2) (p "resp" "a" z a) (p "resp" "b" z b)
           (p "resp" "s" z s) (p "resp" "d" z d)
           (p "" z-0 1) (p "" "x" z-0 d)
           @TRUST@)
      (false))))
\end{lstlisting}

\par\smallskip\noindent\begin{tabular}{@{}l p{0.8\textwidth}@{}}
file & \texttt{cpsa/protocols/blanchet-fixed.scm}\\
$U$ & $\Non(sk_A)$ = \texttt{(non (privk a))}, $\Non(sk_B)$ = \texttt{(non (privk b))}, $\Unq(s)$ = \texttt{(uniq s)}, $\Unq(d)$ = \texttt{(uniq d)}\\
\end{tabular}
\begin{lstlisting}[language={},aboveskip=4pt]
(defgoal blanchet-fixed
  (forall ((a b name) (s skey) (d data) (z z-0 strd))
    (implies
      (and (p "resp" z 2) (p "resp" "a" z a) (p "resp" "b" z b)
           (p "resp" "s" z s) (p "resp" "d" z d)
           (p "" z-0 1) (p "" "x" z-0 d)
           @TRUST@)
      (false))))
\end{lstlisting}

\paragraph{Otway--Rees, and with $A \neq B$.} \texttt{tst/or.scm}. The goal
is the responder's agreement with the initiator on $(a, b, s, m)$. Under the
full trust, every failing shape has $a = b$. $B$ runs the responder with
itself, and the server accepts $B$'s own block in the initiator's position,
because with $a = b$ the two blocks have the same form and key. These are runs
of the model, but we know of no published attack on Otway--Rees that uses
them. They arise because the goal does not exclude $a = b$. The second row
adds $a \neq b$ to the situation, using a rule \texttt{neq} that makes
\texttt{(fact neq x x)} contradictory, appended to the protocol, and
\texttt{(fact neq a b)} in the hypothesis. Its answer requires the
freshness of $B$'s nonce, since otherwise $B$ accepts the server's reply
for another session that reused $n_B$.
\par\smallskip\noindent\begin{tabular}{@{}l p{0.8\textwidth}@{}}
file & \texttt{cpsa/protocols/otway-rees.scm}\\
$U$ & $\Non(k_{AS})$ = \texttt{(non (ltk a s))}, $\Non(k_{BS})$ = \texttt{(non (ltk b s))}, $\Unq(M)$ = \texttt{(uniq m)}, $\Unq(n_B)$ = \texttt{(uniq nb)}\\
in $\sigma$ & \texttt{(uniq-orig k)} in \texttt{serv}\\
\end{tabular}
\begin{lstlisting}[language={},aboveskip=4pt]
(defgoal or
  (forall ((a b s name) (m nb text) (k skey) (z strd))
    (implies
      (and (p "resp" z 4) (p "resp" "a" z a) (p "resp" "b" z b)
           (p "resp" "s" z s) (p "resp" "m" z m) (p "resp" "nb" z nb)
           (p "resp" "k" z k)
           @TRUST@)
      (exists ((z-0 strd))
        (and (p "init" z-0 1) (p "init" "a" z-0 a) (p "init" "b" z-0 b)
             (p "init" "s" z-0 s) (p "init" "m" z-0 m))))))
\end{lstlisting}

\par\smallskip\noindent\begin{tabular}{@{}l p{0.8\textwidth}@{}}
file & \texttt{cpsa/protocols/otway-rees-neq.scm}\\
$U$ & $\Non(k_{AS})$ = \texttt{(non (ltk a s))}, $\Non(k_{BS})$ = \texttt{(non (ltk b s))}, $\Unq(M)$ = \texttt{(uniq m)}, $\Unq(n_B)$ = \texttt{(uniq nb)}\\
in $\sigma$ & \texttt{(uniq-orig k)} in \texttt{serv}\\
rules & \texttt{neq}\\
\end{tabular}
\begin{lstlisting}[language={},aboveskip=4pt]
(defgoal or
  (forall ((a b s name) (m nb text) (k skey) (z strd))
    (implies
      (and (p "resp" z 4) (p "resp" "a" z a) (p "resp" "b" z b)
           (p "resp" "s" z s) (p "resp" "m" z m) (p "resp" "nb" z nb)
           (p "resp" "k" z k) (fact neq a b)
           @TRUST@)
      (exists ((z-0 strd))
        (and (p "init" z-0 1) (p "init" "a" z-0 a) (p "init" "b" z-0 b)
             (p "init" "s" z-0 s) (p "init" "m" z-0 m))))))
\end{lstlisting}

\paragraph{Yahalom.} \texttt{tst/yahalom.scm}. The goal is the responder's
agreement with the initiator on the names, the server and the session key.
\par\smallskip\noindent\begin{tabular}{@{}l p{0.8\textwidth}@{}}
file & \texttt{cpsa/protocols/yahalom.scm}\\
$U$ & $\Non(k_{AS})$ = \texttt{(non (ltk a c))}, $\Non(k_{BS})$ = \texttt{(non (ltk b c))}, $\Unq(n_A)$ = \texttt{(uniq n-a)}, $\Unq(n_B)$ = \texttt{(uniq n-b)}\\
in $\sigma$ & \texttt{(uniq-orig k)} in \texttt{serv}\\
\end{tabular}
\begin{lstlisting}[language={},aboveskip=4pt]
(defgoal yahalom
  (forall ((a b c name) (n-a n-b text) (k skey) (z strd))
    (implies
      (and (p "resp" z 4) (p "resp" "a" z a) (p "resp" "b" z b) (p "resp" "c" z c)
           (p "resp" "n-a" z n-a) (p "resp" "n-b" z n-b) (p "resp" "k" z k)
           @TRUST@)
      (exists ((z-0 strd))
        (and (p "init" z-0 3) (p "init" "a" z-0 a) (p "init" "b" z-0 b)
             (p "init" "c" z-0 c) (p "init" "k" z-0 k))))))
\end{lstlisting}

\paragraph{Denning--Sacco.} \texttt{tst/denning-sacco.scm}, the shipped
version, which already includes both names in the signed message. The goal is
the responder's agreement with the initiator on the names and the session key.
CPSA's \texttt{pubk} is a function of the name, so the server's certificates
bind nothing that the key does not already bind.
\par\smallskip\noindent\begin{tabular}{@{}l p{0.8\textwidth}@{}}
file & \texttt{cpsa/protocols/denning-sacco.scm}\\
$U$ & $\Non(sk_A)$ = \texttt{(non (privk a))}, $\Non(sk_B)$ = \texttt{(non (privk b))}, $\Non(sk_{KS})$ = \texttt{(non (privk ks))}, $\Unq(K)$ = \texttt{(uniq k)}, $\Unq(t_A)$ = \texttt{(uniq ta)}\\
\end{tabular}
\begin{lstlisting}[language={},aboveskip=4pt]
(defgoal denning-sacco
  (forall ((a b ks name) (k skey) (ta text) (z strd))
    (implies
      (and (p "resp" z 1) (p "resp" "a" z a) (p "resp" "b" z b)
           (p "resp" "ks" z ks) (p "resp" "k" z k) (p "resp" "ta" z ta)
           @TRUST@)
      (exists ((z-0 strd))
        (and (p "init" z-0 3) (p "init" "a" z-0 a) (p "init" "b" z-0 b)
             (p "init" "k" z-0 k))))))
\end{lstlisting}

\paragraph{Neuman--Stubblebine.} \texttt{tst/neuman-stubblebine.scm}. The goal
is the responder's agreement with the initiator on the names and the session
key.
\par\smallskip\noindent\begin{tabular}{@{}l p{0.8\textwidth}@{}}
file & \texttt{cpsa/protocols/neuman-stubblebine.scm}\\
$U$ & $\Non(k_{AS})$ = \texttt{(non (ltk a ks))}, $\Non(k_{BS})$ = \texttt{(non (ltk b ks))}, $\Unq(r_A)$ = \texttt{(uniq ra)}, $\Unq(r_B)$ = \texttt{(uniq rb)}, $\Unq(t_B)$ = \texttt{(uniq tb)}\\
in $\sigma$ & \texttt{(uniq-orig k)} in \texttt{keyserver}\\
\end{tabular}
\begin{lstlisting}[language={},aboveskip=4pt]
(defgoal neuman-stubblebine
  (forall ((a b ks name) (ra rb tb text) (k skey) (z strd))
    (implies
      (and (p "resp" z 3) (p "resp" "a" z a) (p "resp" "b" z b)
           (p "resp" "ks" z ks) (p "resp" "ra" z ra) (p "resp" "rb" z rb)
           (p "resp" "tb" z tb) (p "resp" "k" z k)
           @TRUST@)
      (exists ((z-0 strd))
        (and (p "init" z-0 3) (p "init" "a" z-0 a) (p "init" "b" z-0 b)
             (p "init" "k" z-0 k))))))
\end{lstlisting}

\paragraph{Kerberos (simplified), and agreement on $K$.}
\texttt{tst/kerberos.scm}, protocol text unchanged, in both rows. The key
server declares the session key $k$ \texttt{uniq-orig}. The lifetime $l$ is
not a nonce and has no \texttt{uniq} atom. The first row has the goal the
rule produced: the initiator's agreement with the responder on the names
and the session key. Its answer is $\mathsf{false}$. The key server's two
tickets $\{t, l, k, b\}_{ltk(a,ks)}$ and $\{t, l, k, a\}_{ltk(b,ks)}$ have
the same form, and the reply $\{t'\}_k$ names no one. The adversary hands
the key server's message to an initiator $B$ with peer $a$ with the tickets
swapped, forwards $B$'s authenticator to $A$ acting as responder, and $A$'s
own responder strand answers $A$'s initiator. The initiator completes, but
the responder that answered is $A$ with peer $b$, so agreement on the names
fails. It fails at every trust, since no atom of $U$ rules the run out.
This is a property of the simplified model and not a known attack on
Kerberos~V5.
\par\smallskip\noindent\begin{tabular}{@{}l p{0.8\textwidth}@{}}
file & \texttt{cpsa/protocols/kerberos-names.scm}\\
$U$ & $\Non(k_{AS})$ = \texttt{(non (ltk a ks))}, $\Non(k_{BS})$ = \texttt{(non (ltk b ks))}, $\Unq(t)$ = \texttt{(uniq t)}, $\Unq(t')$ = \texttt{(uniq t-prime)}\\
in $\sigma$ & \texttt{(uniq-orig k)} in \texttt{keyserver}\\
\end{tabular}
\begin{lstlisting}[language={},aboveskip=4pt]
(defgoal kerberos
  (forall ((a b ks name) (t t-prime l text) (k skey) (z strd))
    (implies
      (and (p "init" z 4) (p "init" "a" z a) (p "init" "b" z b)
           (p "init" "ks" z ks) (p "init" "t" z t) (p "init" "t-prime" z t-prime)
           (p "init" "l" z l) (p "init" "k" z k)
           @TRUST@)
      (exists ((z-0 strd))
        (and (p "resp" z-0 2) (p "resp" "a" z-0 a) (p "resp" "b" z-0 b)
             (p "resp" "k" z-0 k))))))
\end{lstlisting}

The second row, added after the first results, weakens the conclusion to
the initiator's agreement with some responder strand on the session key
alone.
\par\smallskip\noindent\begin{tabular}{@{}l p{0.8\textwidth}@{}}
file & \texttt{cpsa/protocols/kerberos-key.scm}\\
$U$ & $\Non(k_{AS})$ = \texttt{(non (ltk a ks))}, $\Non(k_{BS})$ = \texttt{(non (ltk b ks))}, $\Unq(t)$ = \texttt{(uniq t)}, $\Unq(t')$ = \texttt{(uniq t-prime)}\\
in $\sigma$ & \texttt{(uniq-orig k)} in \texttt{keyserver}\\
\end{tabular}
\begin{lstlisting}[language={},aboveskip=4pt]
(defgoal kerberos
  (forall ((a b ks name) (t t-prime l text) (k skey) (z strd))
    (implies
      (and (p "init" z 4) (p "init" "a" z a) (p "init" "b" z b)
           (p "init" "ks" z ks) (p "init" "t" z t) (p "init" "t-prime" z t-prime)
           (p "init" "l" z l) (p "init" "k" z k)
           @TRUST@)
      (exists ((z-0 strd))
        (and (p "resp" z-0 2) (p "resp" "k" z-0 k))))))
\end{lstlisting}

\paragraph{Excluded candidates.} Wide-Mouthed Frog
(\texttt{tst/wide-mouth-frog.lsp}): at the four points of the lattice that
assume both long-term keys, CPSA does not terminate, since each server output
is again a valid server input with the roles swapped, and the file's own header
notes infinitely many shapes. Woo--Lam (\texttt{tst/woolam.scm}): the
responder role declares its long-term key \texttt{non-orig} and its nonce
\texttt{uniq-orig}, which leaves $|U| = 1$. The two-pass unilateral ISO 9798-3
exchange (\texttt{tst/unilateral.scm}): its text fields are not nonces, which
leaves $|U| = 1$. The Diffie--Hellman and station-to-station files require
an equational theory (A5). Needham--Schroeder with symmetric keys, Andrew's
secure RPC, and a TLS-like handshake are not shipped, and the CPSA versions of
Needham--Schroeder and Lowe's fix are covered by the analyser rows.

\end{document}